\documentclass[a4paper]{article}

\usepackage[english]{babel}
\usepackage[utf8x]{inputenc}
\usepackage[T1]{fontenc}
\usepackage{authblk}
\usepackage{parskip}
\usepackage[onehalfspacing]{setspace}

\usepackage{amsthm}
\usepackage{amsmath}
\newtheorem{prop}{Proposition}
\DeclareMathOperator*{\argmin}{argmin}
\usepackage{amssymb}
\usepackage{mathrsfs}
\usepackage{amsfonts}
\usepackage{enumitem}

\usepackage{bbm}
\usepackage{bigints}
\usepackage{dsfont}
\usepackage{mathtools}
\usepackage{stmaryrd}

\usepackage[title]{appendix}

\DeclareMathAlphabet{\mathpzc}{OT1}{pzc}{m}{it}
\theoremstyle{definition}
 \newtheorem{step}{Step}

\usepackage[a4paper,top=3cm,bottom=2cm,left=3cm,right=3cm,marginparwidth=1.75cm, heightrounded]{geometry}

\usepackage{graphicx}
\usepackage{subcaption}
\usepackage[colorinlistoftodos]{todonotes}
\usepackage[colorlinks=true, allcolors=blue, pdfencoding=auto, psdextra]{hyperref}
\usepackage[authoryear,comma,semicolon]{natbib}
\usepackage{booktabs}

\title{Beyond unidimensional poverty analysis using distributional copula models for mixed ordered-continuous outcomes}
\author[1]{Maike Hohberg\thanks{Corresponding author: mhohber@uni-goettingen.de. The author received funding from the Ministry for Science and Culture of Lower Saxony as a part of the project ``Reducing Poverty Risks in Developing Countries'' and the German Science Foundation within the research project KN 922/9-1.}}
\author[2]{Francesco Donat\thanks{This paper should not be reported as representing the views of the Single Resolution Board. The views expressed are those of the authors and do not necessarily reflect those of the Board.}} 
\author[3]{Giampiero Marra}
\author[1]{Thomas Kneib} 

\affil[1]{Chair of Statistics, University of Goettingen, Germany}
\affil[2]{Single Resolution Board, Brussels, Belgium}
\affil[3]{Department of Statistical Science, University College London, UK}

\begin{document}
\maketitle

\begin{abstract}
\noindent Poverty is a multidimensional concept often comprising a monetary outcome and other welfare dimensions such as education, subjective well-being or health, that are measured on an ordinal scale. In applied research, multidimensional poverty is ubiquitously assessed by studying each poverty dimension independently in univariate regression models or by combining several poverty dimensions into a scalar index. This inhibits a thorough analysis of the potentially varying interdependence between the poverty dimensions. We propose a multivariate copula generalized additive model for location, scale and shape (copula GAMLSS or distributional copula model) to tackle this challenge. {By relating the copula parameter to covariates, we specifically examine if certain factors determine the dependence between poverty dimensions. Furthermore, specifying the full conditional bivariate distribution, allows us to derive several features such as poverty risks and dependence measures coherently from one model for different individuals.}
We demonstrate the approach by studying two important poverty dimensions: income and education. Since the level of education is measured on an ordinal scale while income is continuous, we extend the bivariate copula GAMLSS to the case of mixed ordered-continuous outcomes. The new model is integrated into the \texttt{GJRM} package in \texttt{R} and applied to data from Indonesia. Particular emphasis is given to the spatial variation of the income-education dependence and groups of individuals at risk of being simultaneously poor in both education and income dimensions.      
\end{abstract}

\section{Introduction}

Although poverty is widely regarded a multidimensional phenomenon  and poverty measures moving beyond a single monetary dimension -- such as the Multidimensional Poverty Index \citep[MPI,][]{Alkire.2012} -- have emerged, little progress has been made on \textit{analysing} poverty as a multidimensional concept. To study poverty at the micro level, univariate linear regression is the standard tool of the empirical economist. Despite their widespread use, however, univariate models for poverty analyses require either studying each poverty dimension separately in different equations, or using as response variable an index that subsumes all dimensions in a single number \citep[e.g.][]{Alkire.2018}. Both approaches neglect the interdependence between poverty dimensions and ignore that the dependence itself should be part of the analysis. In fact the level of poverty and well-being depends on the strength of the dependence \citep{Duclos.2006}: for example, lower tail dependence can explain persisting poverty where performing low in one dimension is strongly associated with a low outcome in the other dimensions.

To overcome such limitations, multivariate regression can be used to tackle multidimensionality in poverty analyses. The relationship between two or more outcomes can also modeled using copulas which have been proven to be useful and flexible tools in this regard \citep[see][for an introduction to copula theory]{Nelson.2006}. %\textcolor{red}{The basic idea of a copula model is to separate two or more potentially different marginal distributions from the dependence structure. XXXX FD: I see what you mean here, but I think it should express it in a better way: it took me a while to get the meaning. XXXX}
A second issue in poverty analysis concerns distributional aspects. Especially for program targeting and risk factor analysis, it is important that poverty studies move beyond the simple mean effects. In fact, concepts like vulnerability to poverty -- a forward-looking measure of individuals' exposure to poverty -- look at both the location and scale of the target distribution. Previous studies on vulnerability to poverty used a step-wise procedure to explicitly make the scale parameter dependent on covariates \citep[see][for recent works]{Gunther.2009, Calvo.2013, Nguyen.2015, Zereyesus.2016}. Another example is inequality, which has become growingly relevant for both the political agenda and for projects implemented in developing countries. The World Bank, for example, centers its shared prosperity initiative around the goal to reduce inequality \citep{Worldbank.2018}. Hence, it is necessary to analyse not only effects on the mean but also on the other parameters characterising the distribution of the outcomes of interest. Generalized additive models for location, scale, and shape \citep[GAMLSS,][]{Rigby.2005} are able to capture the effects of covariates on the whole conditional distribution of a single poverty dimension.

Both issues of multidimensionality and distributional aspects can be addressed with a combination of GAMLSS and multivariate copula models, also referred to as copula GAMLSS. These models are implemented in the \texttt{R} package \texttt{GJRM} \citep{GJRM} and comprise a wide range of potential marginal distributions (continuous, binary, discrete) and copulas. A Bayesian version of this model class is implemented in the software \texttt{BayesX} \citep{Belitz.2015} while \citet{Klein2016} provide the related literature. The advantage of embedding copula regression into GAMLSS is that each parameter of the marginals and the copula association parameter can be modeled to depend flexibly on covariates. This allows us to not only measure the strength of the dependence, which has been the focus of previous literature on interrelated poverty dimensions, but also to analyse which factors related to household location and composition drive this dependence. This latter aspect has not been previously considered in poverty studies.

When studying poverty, it often occurs that one dimension is reported in ordered categories whereas the other is continuous. For example, two possible dimensions of interest could be income (measured on the continuous scale) and the highest level of education, which is often assessed in ordered categories such as ``no schooling'', ``elementary school'', ``high school'', and ``higher education''. This is a very relevant case, especially in economics and poverty research where several outcomes are measured on the ordinal scale (health, education, subjective well-being, etc.). 

The aim of this paper is twofold. First, to theoretically extend copula GAMLSS to a mixed ordered-continuous case. Second, to practically demonstrate how multidimensional poverty analysis can benefit from flexible models that allow for covariate effects on the interdependence between the poverty dimensions. 

For the theoretical part, we rely on the latent variable approach relating the ordered categories to an underlying continuous variable as in \citet{Donat2018}. In this way we can follow the approach developed in  \citet{Marra.2017}, which estimates the copula dependence and marginal distribution parameters simultaneously within a penalized likelihood framework using a trust region algorithm. %(KLEIN ET AL) also used the latent variable representation for mixed binary-continuous  outcomes.  
The new model is incorporated into the \texttt{R} package \texttt{GJRM} \citep{GJRM}.  

Regarding the application to multidimensional poverty, there is an extensive literature dealing with the measurement of multidimensional poverty. Yet, the methods proposed for analyzing multidimensional poverty, including its determinants and poverty profiles, are rather limited. For example, \citet{Alkire.2015} suggest employing Generalized Linear Models using a single number index as the response variable. To demonstrate how a more comprehensive poverty analysis can be conducted by researchers, the empirical study in this paper applies copula GAMLSS in this context. Our application deals with two important poverty dimensions that are interrelated: income and education. %Due to the possible dependence, analyzing their determinants in two separate equations might not be appropriate and would lead to biased coefficients of covariate effects on the marginals. 
In many developing countries, there is potentially a vicious cycle of poor education and low income. This cycle is also called poverty trap and is a long-established concept in economics: capable children stay under-educated due to their parents' restricted resources and hence remain poor when grown-up \citep{BARHAM.1995}. Understanding what determines the interdependence between poverty dimensions helps designing strategies to interrupt this cycle. 
To this end, we model the income-education dependency in Indonesia and draw an in-depth picture of monetary and education poverty across the population. We address the following questions: 1) Which factors determine the distributions of household income per capita and individual education and their interdependence? 2) How does this dependence differ spatially across Indonesia? 3) What are the probabilities of being poorly educated \textit{and} income poor for different population's sub-groups? We will answer these questions using a rich dataset from Indonesia which is made publicly available by the RAND corporation \citep{RAND.2017}. 

The dependencies between different poverty dimensions have been widely addressed in the economics literature during the last two decades. However the literature on using copulas to model multidimensional poverty is scarce and, to the best of our knowledge, restricted to the \textit{measurement} of the strength of such dependence. Existing approaches do not place the model into a regression framework and hence neither relate the copula association parameter nor the other parameters characterising the marginal distributions to covariates. For example, \cite{Quinn.2007} quantifies the dependence between income and an ordinal health measure in four industrial countries. \citet{Decancq.2014} uses copula models to measure dependence over time between income, health, and schooling (all of them  assumed to be continuous) in Russia. A similar approach was used by \citet{Perez.2015} to study the dependence between income, material needs and work intensity in Spain. \citet{Kobus.2018} analyse the distributions of health and education. In contrast to \cite{Quinn.2007} and this paper, that make use of a latent variable approach to represent the ordered categories of education, \citet{Kobus.2018} overcome the unidentifiability issue when using copulas with discrete marginals by concordance ordering. In a Bayesian context, \cite{Tan.2018} re-construct the MPI using a one-factor copula model and data from East-Timur. These examples emphasize once more the importance of extending copula GAMLSS also to the case of mixed ordered-continuous outcomes when these models are applied to poverty analyses.

%Main issues include identifying determinants of the dependency parameter between income  and health, analyzing its spatial distribution and assessing inequalities in income and health.  

The remainder of the paper is organised as follows: Section \ref{sec:copulaGAMLSS} introduces a bivariate copula GAMLSS for mixed ordered and continuous outcomes. Section \ref{sec:infer} presents the estimation procedure. Finally, Section \ref{sec:app} studies poverty dimensions with copula GAMLSS using data from Indonesia  and discusses practical approaches to model selection. Section \ref{sec:concl} concludes the paper.  

%vinh bivariate copula und hat auch measures abgeleitet, aber nicht conditional
%bayesian copula

%%%%%%%%%%%%%%%%%%%%%%%%%%%%%%%%%%%%%%%%%%%%%%%%%%%%%%%%%%%%%%%%%%%%%%
%%%%%%%%%%%%%%%%%%%%%%%%%%%%%%%%%%%%%%%%%%%%%%%%%%%%%%%%%%%%%%%%%%%%%%

\section{Model definition}\label{sec:copulaGAMLSS}

\subsection{Bivariate mixed ordered-continuous model}\label{sec:mixed_ordered_model} 

The model considered in this paper deals with a pair of random variables, $(Y_1,Y_2)^\prime$, with support $\mathcal R\times\mathds R$, where $(\mathcal R,\preceq)$ is a totally ordered set under the ordering relation $\preceq$. The elements of $\mathcal R$ are denoted by $r$ and represent the levels of the categorical variable $Y_1$, namely $\mathcal R:=\{1,\ldots,r,\ldots,R+1\}$ with $R+1<\infty$. The variable $Y_2$ is assumed to be continuous. In the case study of Section \ref{sec:app}, response $Y_2$ will represent the income and $Y_1$ the highest level of education attained by each individual surveyed.  

We are interested in building up a statistical model for the joint distribution of the response variables $(Y_1, Y_2)^\prime$ where their dependence structure is represented by means of a copula specification. The bivariate cumulative distribution function can then be written as
\begin{equation}
F_{12}(r,y_2)=\mathcal C(F_1(r),F_2(y_2))\in[0,1],\label{eq:BivCopula}
\end{equation}
where the copula function is $\mathcal C:[0,1]^2\longrightarrow[0,1]$, with $F_1(r):=\mathbb P(Y_1\preceq r)$ and $F_2(y_2):=\mathbb P(Y_2\le y_2)$ being the marginal distributions. A significant advantage of the copula representation is that it decomposes the joint distribution into two marginals distributions, that may come from different families, and a copula function $\mathcal C$ that binds them together. The dependence structure of the two marginals is captured by an association parameter $\gamma$ that is specific to the copula employed as described below. 

If both $F_1$ and $F_2$ are continuous, Sklar's theorem ensures that the copula function is uniquely determined \citep{sklar.1959}. However, since $Y_1$ is categorical in our case, the uniqueness of the copula does not apply directly. We address this limitation by representing the ordinal variable as a coarse version of a latent continuous variable.

%\subsubsection{Continuous representation of an ordinal response}

Let $Y_1^*\in\mathds R$  denote the unobserved (or latent) continuous variable that drives the decision for the observed categories in $\mathcal R$. This continuous latent variable can be modeled as
\begin{align} \label{eq:lat_reg}
Y_{1}^*=  \mathbf x_{1}^\prime\boldsymbol\beta_1+\epsilon_{1},
\qquad  \epsilon_1 \stackrel{iid}{\sim} N(0,1),  
\end{align}
where $\boldsymbol \beta_1$ is a vector of regression coefficients, $\mathbf{x}_{1}$ a vector of covariates, and $\epsilon_{1}$ the error term with density $f_1^*$ and cumulative distribution function (CDF) $F_1^*$.
Later on, the latent variable in (\ref{eq:lat_reg}) will be placed into the more sophisticated GAMLSS framework, but this model formulation with only linear effects shall serve as a starting point. In line with \citet{McKelvey.1975}, the following observation rule linking the latent to the observed variable is applied:
\begin{align}\label{eq:thresholdMech} 
%\{Y_1\preceq k\}\Longleftrightarrow\{Y_1^*\le\theta_k\},\\
%
\{Y_1 = r\} \Longleftrightarrow \{\theta_{r-1}< Y_1^* \leq \theta_r\}, 
\qquad r = 1, \dots, R+1,
\end{align}
where $\theta_r$ is a cut point on the latent continuum related to the level $r$ of $Y_1$. We observe category $r$ if the latent variable is between the cutoffs $\theta_{r-1}$ and $\theta_r$. 
There is a total of $R+2$ cut points: $-\infty = \theta_0 <\theta_1 < \ldots <\theta_{R+1}= \infty$. However, only $R$ of them are estimable, namely $\{\theta_1, \ldots, \theta_R\}$.

%In the univariate context it is required that 
%
%\begin{equation}
%\{\theta_k\le\theta_{\tilde k}\}\mbox{ in }\mathds R\Longleftrightarrow\{\theta_k\le\theta_{\tilde k}\}\mbox{ in }\mathcal K\nonumber
%\end{equation}
%
%for any $k,\tilde k\in\mathcal K$ \textcolor{red}{[TO BE PRECISE WE SHOULD TALK TO SOME SUBSETS OF $\mathds R$]}. 

To guarantee the monotonicity of the cut points, we apply the transformation $\theta_1^*:=\theta_1$ and $\theta_r^*:=\sqrt{\theta_r-\theta_{r-1}}$ for any $r>1$ \citep{Donat2017}. This implies that $\theta_r \ge \theta_{r-1}$ for any $r \in \mathcal R$ and $\theta_r \in \mathbb R$. However, the equality $\theta_r=\theta_{r-1}$ can be problematic in practice because it results in estimated parameters at the boundary of the parameter space. This happens, for example, wherever a given level of $Y_1$ has no observations in the sample \citep{Haberman.1980}. 
%\theta_r =(\theta_r^*)^2+\theta_{r-1}=\sum_{1\prec\tilde r\preceq r}(\theta_{\tilde r}^*)^2+\theta_1\ge\theta_{r-1}.

From equation (\ref{eq:lat_reg}) and (\ref{eq:thresholdMech}), we derive the cumulative link model 
\begin{align}\label{eq:cum_model}
\mathbb P(Y_{1} \preceq r) &= \mathbb P(Y_{1}^*\leq\theta_r) %\\
= \mathbb P(\epsilon_{1}\leq{\theta_r-\mathbf x_{1}^\prime\boldsymbol\beta_1}) %\\
:= F_1^*(\underbrace{\theta_r-\mathbf x_{1}^{\prime}\boldsymbol\beta_1}_{:=\eta_{1r}}),
\end{align}
where $\eta_{1r}$ is the predictor associated with the ordinal categorical response in the model. It depends on the observed level $r$ of $y_{1}$ through cut point $\theta_r$. In Section \ref{sec:cop_gamlss} the predictor $\eta_{1r}$ will be replaced with a generalized additive form.  With this information in hand, equation (\ref{eq:BivCopula}) can equivalently be written as 
%
%By virtue of equivalence (\ref{eq:thresholdMech}), it is implied that $\mathbb P(Y_1\preceq k)=\mathbb P(Y_1^*\le\theta_k)=:F_1^*(\eta_k)$,  
%
\begin{align}\label{eq:biv_model}
F_{12}(r,y_2)=F_{12}^*(\eta_{1r},y_2)=
\mathcal C(F_1(r),F_2(y_2))
= \mathcal C(F_1^*(\eta_{1r}),F_2(y_2)). 
\end{align}
Since both marginals are now continuous, the applicability of Sklar's theorem in ensured.
%

%\subsubsection{Density of a bivariate copula for mixed ordered-continuous outcomes %$f_{1,2}$
%}
Finally, deriving the analytical form of the density function $f_{12}^*$ yields
%
%\begin{equation} 
%f_{12}(r,y_2)=
%\left\{
%\begin{aligned}
%& \frac{\partial \mathcal C(F_1(1), F_2(y_2)) f_2(y_2)}{\partial F_2(y_2)} 
%\qquad \textrm{for} \quad y_1 = 1\\
%
%& \left(\frac{\partial\mathcal C(F_1(r),F_2(y_2))}{\partial F_2(y_2)}
%-\frac{\partial\mathcal C(F_1(r-1),F_2(y_2))}
%{\partial F_2(y_2)}\right)
%f_2(y_2) \qquad \textrm{for} \quad 1<y_1\leq R+1 \nonumber
%\end{aligned}
%\right.
%\end{equation}
%
%and 
%
%\begin{align*}
%f_{12}^*(\eta_{1r},y_2)
%&=\Delta_r \left(\frac{\partial\mathcal C(F_1^*(\eta_{1r}),F_2(y_2))}{\partial F_2(y_2)}\right) \frac{\partial F_2(y_2)}{\partial y_2} \\
%
%&=\left(\frac{\partial\mathcal C(F_1^*(\eta_{1r}),F_2(y_2))}{\partial F_2(y_2)}-\frac{\partial\mathcal C(F_1^*(\eta_{1r-1}),F_2(y_2))}{\partial F_2(y_2)}\right)f_2(y_2), 
%\end{align*}
%
%where $F_1(r) := F_1^*(\eta_{1r})$. 
\begin{equation} 
f_{12}^*(\eta_{1r},y_2)=
\left\{
\begin{aligned}
& \frac{\partial \mathcal C(F_1^*(\eta_{1r}), F_2(y_2)) f_2(y_2)}{\partial F_2(y_2)} 
\qquad \textrm{for} \quad r = 1\\
& \left(\frac{\partial\mathcal C(F_1^*(\eta_{1r}),F_2(y_2))}{\partial F_2(y_2)}
-\frac{\partial\mathcal C(F_1^*(\eta_{1r-1}),F_2(y_2))}
{\partial F_2(y_2)}\right)
f_2(y_2) \qquad \textrm{for} \quad 1<r\leq R+1. \nonumber
\end{aligned}
\right.
\end{equation}

This will form the basis for the derivation of the penalized log-likelihood function in Section \ref{sec:estim}.

\subsection{Copula GAMLSS}\label{sec:cop_gamlss}

The bivariate copula model is embedded into the distributional regression framework to model flexibly both the dependence parameter and the marginal distributions. To this end, the response vector $\boldsymbol y_i = (y_{1i}^*, y_{2i})^{\prime}$, $i=1,\ldots,n,$ is assumed to follow a parametric distribution where potentially all parameters, except of the cut-points, are related to a regression predictor and consequently to covariates. 
We write the joint conditional density as $f_{12}^*(\vartheta_{1i}, \dots , \vartheta_{Ki}|\boldsymbol \nu_i)$, where the vector $\boldsymbol \nu_i$ collects any covariates associated to the parameters $\vartheta_{ki}, k= 1, \dots, K$ of density $f_{12}^*$.  
Accordingly, the distributional parameter vector  $\boldsymbol\vartheta_i=(\theta_1^*, \dots, \theta_R^*, \vartheta_{1i}, \dots, \vartheta_{Ki})^{\prime}$ includes the transformed cut-points $\{\theta_r^*\}$, the location parameter of the first marginal distribution,  all other distributional parameters related to the second marginal distribution, and the copula parameter $\gamma_i$. Subscript $i$ attached to parameters is made explicit to stress their potential dependence on individual-level covariates. For the ordinal response, logit and probit link functions can be applied and the scale parameter for density $f_1$ is set to one in order to achieve identification as for a probit/logit model. The second marginal distribution can be selected from a wide range of options that are available in \texttt{GJRM} and listed in \citet{Marra.2017}. At the current stage, some of them are not implemented for the mixed-ordinal case, but will be made available in the near future. 
In this paper we only consider one-parameter copulas; some available options are summarized in Table \ref{tab:mixed.copulae} (Appendix~\ref{sec:copulas}) although rotated versions are also implemented in \texttt{GJRM}. 
Since the copula parameter $\gamma_i$ is not directly comparable over different models, we relate it to the Kendall's $\tau$ which can be used for interpreting the dependence. For optimisation and modelling purposes, an appropriate transformation of the copula parameter, $\gamma_i^*$, is used in the estimation algorithm as highlighted in the last column of Table \ref{tab:mixed.copulae} (Appendix).

%To give an example for a parameter vector, we anticipate our  application in Section \ref{sec:app}. There the categorical variable $Y_1$ has 5 levels for education and for the second outcome variable $Y_2$, expenditures, the two parameter lognormal distribution provided a good fit. As an appropriate copula we identified the normal copula. That is, the final parameter vector comprises of $\boldsymbol \vartheta = (\theta_1^*,\dots, \theta^*_4, \mu_1, \mu_2, \sigma_2, \gamma)$

In the spirit of the GAMLSS approach, each distributional element in the parameter vector is related to an additive predictor via
\begin{align} \label{eq:predictor}
\vartheta_{ki} = h_k(\eta^{\vartheta_k}_i) \quad \textrm{and} \quad
 \eta^{\vartheta_k}_i = g_k(\vartheta_{ki}),  
 \end{align}
where $\eta^{\vartheta_k}_i \in \mathds R$ is the predictor belonging to distributional parameter $\vartheta_{ki}$, and $h_k=g_k^{-1}$ is a response function mapping the real line into the domain of $\vartheta_{ki}$. 

For the ordinal response, $\eta_{1ri}$ in equation (\ref{eq:cum_model}) can now be represented as $\eta^{\mu_1}_{ri} = \theta_r - \eta^{\mu_1}_i,$ where $\eta^{\mu_1}_i$ is a predictor as in (\ref{eq:predictor}). 
The predictor $\eta^{\vartheta_k}_i$ takes the additive form 
\begin{equation*}
\eta^{\vartheta_k}_i =  \sum^{J_k}_{j=1} s_{j}^{\vartheta_k}(\boldsymbol \nu_i),
\end{equation*}
%
%\beta_0^{\vartheta_k} +
%
where %$\beta_0^{\vartheta_k}$ denotes the overall level of the predictor, 
functions $s_j^{\vartheta_k}(\boldsymbol \nu_i), j = 1, \dots, J_k,$ can be chosen to model a range of different effects of (a subset) of explanatory variables $\boldsymbol \nu_i$. In particular,   
\begin{itemize}
\item Linear effects are represented by setting $s_j^{\vartheta_k}(\boldsymbol \nu_i) = \nu_{ji}^{\vartheta_k} {\beta}_j^{\vartheta_k}$, where $\nu_{ji}^{\vartheta_k}$ is a singleton element of $\boldsymbol \nu_i$ and $\beta_j^{\vartheta_k}$ a regression coefficient to be estimated. For the second marginal, this also includes an intercept with $s_j^{\vartheta_k}(\boldsymbol \nu_i) = \beta_0^{\vartheta_k}$ to denote the overall level of the predictor while for the ordinal equation the intercept is already accounted for by the cut-point $\theta_r$.
%Linear effects are captured by linear functions $s_j^{\vartheta_k}(\boldsymbol \nu_i) = \mathbf x'_i \boldsymbol{\beta}_j^{\vartheta_k}$, where $\mathbf x'_i $ is a subvector of $\boldsymbol{\nu}_i$ and $\boldsymbol{\beta}_j^{\vartheta_k}$ a vector of coefficients of the same length. 
%
\item For continuous covariates, nonlinear effects are achieved by including smooth functions $s_j^{\vartheta_k }(\boldsymbol \nu_i)$
%$s_j^{\vartheta_k }(\boldsymbol \nu_i) = s_j^{\vartheta_k }(x_i)$ where $x_i$ is a single element.
represented by penalized regression splines. \citet{ruppert.2003} and \citet{wood2017generalized} provide various definitions and options for computing basis functions and related penalties.  
\item An underlying spatial pattern can be accounted for by specifying 
%$s_j^{\vartheta_k }(\boldsymbol \nu_i) = s_j^{\vartheta_k }(x_i)$, where $x_i$ is some type of 
spatial information such as geographical coordinates or administrative units in $\boldsymbol \nu_i$. Smoothing penalties can account for the neighbourhood structure and ensure that effects are similar for adjacent regions. \citet{rue.2005} interpret this penalty as the assumption that the vector of spatial effects for all regions follows a Gaussian Markov random field. 
\item If the data are clustered, random effects $s_j^{\vartheta_k }(\boldsymbol \nu_i) = \beta_{jc_i}^{\vartheta_k }$ can be included with $c_i$ denoting the cluster the observations are grouped into.
\end{itemize}
%
%\textcolor{red}{Linear, non-linear and spatial effects can all be represented by linear combinations of design matrices appropriately constructed -- and comprising covariates, bases and adjacent matrices -- and unknown regression parameters to be estimated [REFERENCE HERE!!]. We can thus re-state each model predictor as}
%
%\textcolor{red}{
%\begin{equation}
%\begin{array}{rcll}
%\eta_{ti}^{\mu_1} &=& \theta_i - \mathbf x_{1i}^\prime\boldsymbol\beta^{\mu_1} & \mbox{for }t = 1,\ldots,T\\
%\eta_i^{\vartheta_k}&=&\mathbf x_{ki}^\prime\boldsymbol\beta^{\vartheta_k} & \mbox{for } \vartheta_k\neq\mu_1
%\end{array}.\nonumber
%\end{equation}
%}
%
%\textcolor{red}{An intercept term can be incorporated in the model by letting the first element of $\mathbf x_{ki}$ be 1. However, for the ordinal equation the intercept is already accounted in cut point $\theta_t$; it should therefore be omitted from $\mathbf x_{1i}$. Regression coefficients are collected in vector $\boldsymbol\beta :=(\theta_1^*,\ldots,\theta_T^*,\boldsymbol\beta^{\vartheta_1},\ldots,\boldsymbol\beta^{\vartheta_K})^\prime$, whereas design matrices are for the $i$th individual are labelled as $\mathbf x_i :=(\mathbf x_{1i},\ldots,\mathbf x_{Ki})^\prime$.}

%%%%%%%%%%%%%%%%%%%%%%%%%%%%%%%%%%%%%%%%%%%%%%%%%%%%%%%%%%%%%%%%%%%%%%
%%%%%%%%%%%%%%%%%%%%%%%%%%%%%%%%%%%%%%%%%%%%%%%%%%%%%%%%%%%%%%%%%%%%%

\section{Estimation} \label{sec:infer}

\subsection{Maximum penalized likelihood}

From the analytical expression of the bivariate density $f_{12}$ given in Section \ref{sec:mixed_ordered_model}, the model's log-likelihood function is derived as
\begin{equation}
\ell(\boldsymbol\beta)=\sum_{i=1}^n\left(\sum_{r\in\mathcal R}\mathds 1_{\{y_{1i}=r\}}\left(\log\{F_{12.2}(\eta_{1ri}, y_{2i})-F_{12.2}(\eta_{1r-1i}, y_{2i})\}\right)+\log\{f_2(y_{2i})\}\right),
\end{equation}
%
%\begin{equation}
%\ell(\boldsymbol\beta)=\sum_{i=1}^n\left(\sum_{r\in\mathcal R}\mathds 
%\nonumber
%\end{equation}
%
where $\mathds 1_{\{\cdot\}}$ is a Boolean operator that takes on value $1$ if condition $\{\cdot\}$ is verified, and $0$ otherwise.
We define %The term $F_{12.2}$ is defined as
\begin{align*}
F_{12.2}(\eta_{1ri}, y_{2i}):=\frac{\partial\mathcal C(F_1^*(\eta_{1ri}),F_2( y_{2i}))}{\partial F_2(y_{2i})} \qquad \textrm{with} \qquad
F_{12.2}(\eta_{1,1-1,i}, y_{2i}) = 0.
\end{align*}
The log-likelihood function is maximized with respect to the complete vector of regression coefficients $\boldsymbol\beta = (\theta^*_1, \dots, \theta^*_R, \boldsymbol \beta^{\vartheta_1}, \ldots, \boldsymbol \beta^{\vartheta_K})^{\prime}$. Each vector of regression coefficients $\boldsymbol \beta^{\vartheta_k}$ includes the coefficients for one parameter $\vartheta_k$.   

Embedding the model into the distributional regression framework with highly flexible predictors, including regression spline components, typically requires penalization to avoid overfitting. 
The penalized log-likelihood $\ell_p(\boldsymbol\beta)$ with ridge-type penalty can be written as 
\begin{align} \label{lik_pen}
\ell_p(\boldsymbol \beta) = \ell (\boldsymbol \beta) - \frac{1}{2} \boldsymbol\beta ^{\prime} \boldsymbol S_{\boldsymbol\lambda} \boldsymbol\beta,
\end{align}
where $\boldsymbol S_{\boldsymbol\lambda}$ 
% = \textrm{diag} (\mathbf{0}, \lambda_1^{\prime \vartheta_1} \boldsymbol D_1^{\vartheta_1}, \ldots, \lambda_{J_1}^{\prime \vartheta_1} \boldsymbol D_{J_1}^{\vartheta_1}, \ldots, \lambda_{J_K}^{\prime \vartheta_K} \boldsymbol D_{J_K}^{\vartheta_K})$
is a block diagonal matrix consisting of the penalties associated to each model parameter. For un-penalized parameters (like the cut points or categorical covariates) the corresponding block of $\boldsymbol S_{\boldsymbol\lambda}$ is set to $\boldsymbol0$. Penalty matrices are associated with smoothing parameters $\boldsymbol \lambda = (\lambda_1^{\vartheta_1}, \ldots, \lambda_{J_K}^{\vartheta_K})^\prime $. %$J_k$ indexes the basis functions effects. 

\subsection{Parameter estimation using the trust region algorithm}  \label{sec:estim}

\cite{Marra.2017} proposed maximizing the penalized likelihood in equation  (\ref{lik_pen}) using a trust region algorithm with integrated automatic selection of the smoothing parameters. As in \citet{Radice2016}, \citet{Marra.2017} and \cite{klein.marra.2019}, the estimation proceeds in two steps:

\vspace{6pt}
\begin{step}
At iteration $a$, equation (\ref{lik_pen}) is maximized for a given parameter vector $\boldsymbol \beta^{[a]}$ holding $\boldsymbol \lambda^{[a]}$ fixed at a vector of values. A trust region algorithm is applied as follows
\begin{align}
\boldsymbol \beta^{[a+1]} = \boldsymbol \beta^{[a]} + \underbrace{\argmin_{\boldsymbol p: \| \boldsymbol p\| \leq \Delta^{[a]}} \breve{\ell}_p (\boldsymbol \beta^{[a]})}_{:=\boldsymbol p^{[a+1]} }  , \label{eq:lp_min} \\
\breve{\ell}_p(\boldsymbol \beta^{[a]}) :=
-\{\ell_p(\boldsymbol\beta^{[a]} 
+ \boldsymbol p^{\prime} \boldsymbol g_p(\boldsymbol \beta^{[a]}) 
+ \frac{1}{2} \boldsymbol p^{\prime} \boldsymbol {H}_p^{[a]} \boldsymbol{p} \},
\nonumber
\end{align}
where the Euclidean norm is denoted by $\|\cdot\|$ and $\Delta^{[a]}$ is the radius of the trust region. The radius is adjusted in each iteration \citep[see][for details]{Geyer.2015}. The gradient vector at iteration $a$ is given by $\boldsymbol g_p^{[a]} = \boldsymbol g^{[a]} - \boldsymbol S_{\boldsymbol\lambda} \boldsymbol\beta^{[a]}$ is  and $\boldsymbol H_p^{[a]}= \boldsymbol H^{[a]} - \boldsymbol S_{\boldsymbol\lambda}$ is the Hessian matrix - both penalized  by matrix $\boldsymbol S_{\boldsymbol\lambda}$. 

The vector $\boldsymbol g(\boldsymbol \beta^{[a]})$ consists of
\begin{align*}
\boldsymbol g^{[a]} (\boldsymbol \beta^{[a]}) 
=\left(\left.\frac{\partial \ell(\boldsymbol \beta)}{\partial \theta_1^*}\right |_{ \theta_1^* = \theta_1^{*[a]}},\ldots,
\left.\frac{\partial \ell(\boldsymbol \beta)}{\partial \theta_R^*}\right |_{ \theta_R^* =  \theta_R^{*[a]}},
\left.\frac{\partial \ell(\boldsymbol \beta)}{\partial \boldsymbol \beta^ {\vartheta_1}}\right |_{\boldsymbol \beta^{\vartheta_1} = \boldsymbol \beta^{\vartheta_1^{[a]}}}, 
\ldots,
\left.\frac{\partial \ell(\boldsymbol \beta)}{\partial \boldsymbol \beta^ {\vartheta_K}}\right |_{\boldsymbol \beta^{\vartheta_K} = \boldsymbol \beta^{\vartheta_K^{[a]}}}\right)^\prime
\end{align*}
and the elements of the Hessian matrix are 
\begin{align*}
\boldsymbol H(\boldsymbol \beta^{[a]})^{l,m} = \left.\frac{\partial^2 \ell(\boldsymbol \beta)}{\partial \boldsymbol \beta^l \partial \boldsymbol\beta^{m\prime}} \right| _{\boldsymbol \beta^l = \boldsymbol \beta^{l[a]}, \boldsymbol \beta^m = \boldsymbol \beta ^{m[a]}}, \qquad l,m = \vartheta_1, \ldots, \vartheta_K. 
\end{align*}

The second-order partial derivatives of the log-likelihood with respect to cut points $\theta_1^*,\ldots,\theta_R^*$ are derived similarly. At each iteration step, the minimization of equation (\ref{eq:lp_min}) uses a quadratic approximation of $\ell_p(\boldsymbol \beta^{[a]})$, and the solution $\boldsymbol p^{[a+1]}$ is chosen such that it falls within a trust region with centre $\boldsymbol \beta ^{[a]}$ and radius $\Delta^{[a]}$. %\footnote{The $\Delta^{[a]}$ is adjusted in the following way: If $\frac{\ell_p(\boldsymbol{\beta}^{[a+i]}- \ell_p(\boldsymbol{\beta}^{[a]})}{\boldsymbol p^{\prime}\boldsymbol g_p^{[a]} + \frac{1}{2}\boldsymbol p^{\prime} \boldsymbol {H}_p^{[a]} \boldsymbol{p}}$ < 1/4, then do not move to iteration $[a+1]$ but decrease radius to $1/4\, \Delta^{[a]}$.
%
%If $\frac{\ell_p(\boldsymbol{\beta}^{[a+i]}- \ell_p(\boldsymbol{\beta}^{[a]})}{\boldsymbol p^{\prime}\boldsymbol g_p^{[a]} + \frac{1}{2}\boldsymbol p^{\prime} \boldsymbol {H}_p^{[a]} \boldsymbol{p}}$ > 3/4, then do not move to iteration $[a+1]$ but increase radius to $2 \Delta^{[a]}$ \citep[see][for details]{Geyer.2015}.}. 
\end{step}

\vspace{6pt}
\begin{step}
Holding the parameter vector value fixed at $\boldsymbol \beta^{[a+1]}$, the following problem is solved 
\begin{align} \label{eq:step2}
\boldsymbol \lambda^{[a+1]} = \argmin_{\boldsymbol \lambda}
\| \boldsymbol M^{[a+1]} - \boldsymbol A^{[a+1]} \boldsymbol M^{[a+1]} \|^2 
- Kn + 2\mathrm{tr}(\boldsymbol A^{[a+1]}),
\end{align}
where, after defining $\boldsymbol{\cal I}^{[a+1]} = - \boldsymbol H ^{[a+1]}$, the key quantities are
\begin{align*}
\boldsymbol M^{[a+1]} = \sqrt[]{\boldsymbol{\cal I} (\boldsymbol\beta^{[a+1]})} \boldsymbol \beta^{[a+1]}
+ \sqrt[]{\boldsymbol{\cal I} (\boldsymbol\beta^{[a+1]})}^{-1} \boldsymbol g (\boldsymbol \beta^{[a+1]}), \\
\boldsymbol A^{[a+1]} = \sqrt[]{\boldsymbol{\cal I} (\boldsymbol\beta^{[a+1]})}
(\boldsymbol{\cal I}(\boldsymbol \beta^{[a+1]}) + \boldsymbol S_{\boldsymbol\lambda})^{-1} \,\,
\sqrt[]{\boldsymbol{\cal I} (\boldsymbol\beta^{[a+1]})}, 
\end{align*}
$\mathrm{tr}(\boldsymbol A^{[a+1]})$ is the number of effective degrees of freedom (edf) of the penalized model while $K$ is the number of penalized parameters in vector $\boldsymbol \vartheta$.   
The expression in (\ref{eq:step2}) is solved using the method proposed by \citet{wood.2004}. The gradient vector $\boldsymbol g$ and the Hessian $\boldsymbol H$ are obtained as a side product in step 1. Both are analytically derived in a modular fashion for each parameter, see Appendix \ref{sec:apx_gradient} for details.     
\end{step}
%

%%% old version: paragraph on trust region here

%%

Step 1 and 2 are iterated until they no longer improve the objective function, that is until the following criterion is met: 
\begin{align*}
\frac{|\ell(\boldsymbol \beta^{[a+1]}) - \ell(\boldsymbol \beta^{[a]})|}{0.1 + |\ell(\boldsymbol \beta^{[a+1]})|} < 1e^{-0.7}.
\end{align*}

To obtain the starting values for the marginals' parameters and the cut-off value, a generalized additive model is fitted using \texttt{gam()} \citep{wood2017generalized} or a GAMLSS using the \texttt{gamlss()} function within the \texttt{GJRM} package. A transformed Kendall's $\tau$ between the responses is used as a starting value for the copula parameter. Further details on the trust region algorithm and smoothing parameter selection can be found in Appendix \ref{sec:apx_trust} while asymptotic considerations on the proposed maximum penalized likelihood estimator are reported in Appendix \ref{sec:apx_asymp}. 
%%

%%%%%%%%%%%%%%%%%%%%%%%%%%%%%%%%%%%%%%%%%%%%%%%%%%%%%%%%%%%%%%%%%%%%%%
%%%%%%%%%%%%%%%%%%%%%%%%%%%%%%%%%%%%%%%%%%%%%%%%%%%%%%%%%%%%%%%%%%%%%

\subsection{Confidence intervals}
At convergence, reliable point-wise confidence intervals are constructed based on Bayesian large sample approximation as in \cite{wood2017generalized}, for generalized additive models (GAM), i.e.  
\begin{align*}
\hat{\boldsymbol \beta} \overset{a}{\sim} N(\boldsymbol \beta, - \boldsymbol H_p(\hat{\boldsymbol \beta})^{-1}).
\end{align*}
The result for the Bayesian covariance matrix $\boldsymbol V_{\boldsymbol\beta} = - \boldsymbol H_p^{-1}$ is an alternative to the frequentist covariance matrix $\boldsymbol V_{\hat{\boldsymbol\beta}} = - \boldsymbol H_p^{-1} \boldsymbol{H} \boldsymbol H_p^{-1}$. For unpenalized models, the two matrices are equal. Applying the Bayesian framework to the GAM or copula GAMLSS context, follows the notion that penalisation in the estimation implicitly assumes certain prior beliefs about the model's features \citep{wahba.1978}. In this view, a normal prior for the parameter vector $\boldsymbol\beta$, i.e. $f_{\boldsymbol\beta} \propto \exp(-1/2 \boldsymbol\beta^{\prime} \boldsymbol S_{\boldsymbol \lambda} \boldsymbol\beta)$ means that wiggly models are less likely than smoother ones \citep{wood2006}. %Under the large sample assumption, $\boldsymbol H (\boldsymbol \beta)$ can be treated as fixed and the usual Bayesian assumption on the prior of $\boldsymbol \beta $ for smooth models is made. 
\citet{Marra.Wood.2012} give a full justification for using the above approximation and show that $\boldsymbol{V}_{\boldsymbol{\beta}}$ gives close to across-the-function frequentist coverage probabilities since it includes bias and variance components in a frequentist sense, which is not the case for $\boldsymbol{V}_{\boldsymbol{\hat{\beta}}}$.  

%Following \citet{Radice2016}, for a generic smooth term $s_{j}^{\vartheta_k}(\nu_{j,i}^{\vartheta_k})$, point-wise confidence intervals can be constructed using 
%\begin{align*}
%{\cal N} \left(s_{j}^{\vartheta_k}(\nu_{j,i}^{\vartheta_k}),
%\boldsymbol B_{j}^{\vartheta_k} (\nu_{j,i}^{\vartheta_k}) ^{\prime} \boldsymbol V_{\boldsymbol \beta_{j}^{\vartheta_k} }  \boldsymbol B_{j}^{\vartheta_k} (\nu_{j,i}^{\vartheta_k}) \right),
%\end{align*}
%where $\boldsymbol V_{\boldsymbol \beta_{j}^{\vartheta_k}}$ denotes the sub matrix related to the regression spline parameter belonging to $s_{j}^{\vartheta_k}(\nu_{j,i}^{\vartheta_k})$ and $\boldsymbol B_{j}^{\vartheta_k}(\nu_{j,i}^{\vartheta_k})$ is the $i$th vector of basis functions evaluated at observation $\nu_{j,i}^{\vartheta_k}$. 

%%%%%%%%%%%%%%%%%%%%%%%%%%%%%%%%%%%%%%%%%%%%%%%%%%%%%%%%%%%%%%%%%%%%%%
%%%%%%%%%%%%%%%%%%%%%%%%%%%%%%%%%%%%%%%%%%%%%%%%%%%%%%%%%%%%%%%%%%%%%

%\subsection{Parametric bootstrap inference for non-linear functions}

To obtain intervals for non-linear functions of the model parameters (e.g.\ Kendall's $\tau$), \citet{Radice2016} simulate from the posterior distribution of $\boldsymbol \beta$ and give examples of interval construction. They propose the following procedure: 
\begin{enumerate}[leftmargin=1.3cm]
    \item[\textbf{Step 1}] Draw $n_{sim}$ random vectors $\tilde{\boldsymbol\beta}_m, m = 1 \dots, n_{sim}$,  from $\mathcal{N}(\hat{\boldsymbol\beta}, \hat{\boldsymbol V_{\boldsymbol \beta}})$.
\item [\textbf{Step 2}] Calculate $n_{sim}$ realizations of the function under consideration, say $R(\tilde{\boldsymbol\beta_m})$.
\item [\textbf{Step 3}] Calculate the $(\zeta/2)$-th and $(1-\zeta/2)$-th quantile of the realizations where $\zeta$ is typically set to 0.05. The confidence interval is then constructed as $CI_{1-\zeta} = [R(\tilde{\boldsymbol\beta_m)}_{\zeta/2}, R(\tilde{\boldsymbol\beta_m)})_{1-\zeta/2}]$
\end{enumerate}
A value of $n_{sim}$ equal to 100 typically produces reliable results although it can be increased if more precision is required.

\subsection{Simulation study}

To evaluate the effectiveness and implementation of the proposed methodology, we conducted a simulation study with four scenarios that differ in terms of the continuous marginal distribution and the copula specification. All four scenarios are assessed using sample sizes of $n = 1,000$, $3,000$ and $10,000$. The data generating process and detailed results can be found in Appendix~\ref{apx:simulation}. Our approach is able to capture the effect of both linear and nonlinear covariates fairly well and performance improves significantly with increasing sample size. {In addition to recovering the coefficients, we calculated the AIC in every simulation run for the bivariate model and the corresponding independence model. The share of a runs in which the bivariate model had a smaller AIC was 1, providing evidence for the ability of our model to identify dependence between the responses if this is indeed required by the data generating process. We refrain from detailled simulations concerning the selection of marginal distributions and/or copulas since these have been considered before in the literature on copula GAMLSS, albeit for the case of two continuous marginal distributions \citep[e.g.,][]{Marra.2017,Radice2016}.}

%%%%%%%%%%%%%%%%%%%%%%%%%%%%%%%%%%%%%%%%%%%%%%%%%%%%%%%%%%%%%%%%%%%%%%
%%%%%%%%%%%%%%%%%%%%%%%%%%%%%%%%%%%%%%%%%%%%%%%%%%%%%%%%%%%%%%%%%%%%%

%%%%%%%%%%%%%%%%%%%%%%%%%%%%%%%%%%%%%%%%%%%%%%%%%%%%%%%%%%%%%%%%%%%%%%
%%%%%%%%%%%%%%%%%%%%%%%%%%%%%%%%%%%%%%%%%%%%%%%%%%%%%%%%%%%%%%%%%%%%%

%%%%%%%%%%%%%%%%%%%%%%%%%%%%%%%%%%%%%%%%%%%%%%%%%%%%%%%%%%%%%%%%%%%%%%
%%%%%%%%%%%%%%%%%%%%%%%%%%%%%%%%%%%%%%%%%%%%%%%%%%%%%%%%%%%%%%%%%%%%%

\section{Multidimensional poverty in Indonesia} \label{sec:app}

\subsection{The IFLS dataset}

To analyse poverty dimensions in a bivariate copula model and to identify 1) the determinants of the income-education relation, 2) its spatial distribution and  3) groups at risk of being both consumption and education poor, we rely on the most recent wave (IFLS 5) of the Indonesian Family Life Survey (IFLS). The IFLS is a publicly available, longitudinal survey on individual, household, and community level that is designed to study the health and socioeconomic situation of Indonesia’s population. The first wave was implemented in 1993 and covered individuals from 7,224 households representing 83 percent of the population from 13 out of 27 Indonesian provinces \citep{Strauss.2016}. The sample was drawn by stratifying the population on provinces and urban/rural areas before randomly selecting enumeration areas and households within the strata. Due to a large number of split-off households the sample grew up to 16,204 households interviewed in IFLS 5. 

In the IFLS, individuals of an IFLS-household older than 15 years were asked to fill in an “adult individual book” containing questionnaires on subjects such as income, education, employment, and subjective health. We use the level of education as the ordinal response variable, and income as the continuous response. Education is proxied by the highest educational institution attended and can take on five different levels: 1 "no schooling", 2 "primary school", 3 "middle school", 4 "high school", 5 "tertiary education". In analyses for developing countries, income is often proxied by expenditures for consumption. Expenditures are calculated at the household level and then divided by the number of household members. Our income variable is thus precisely expenditures per capita. Due to different price levels in the provinces, we used the province specific minimum wage to adjust expenditures across provinces. Data on the individuals from the “adult individual book” are extracted and merged with relevant information on the household head, such as gender and education, and complemented with information on the household's location, such as province or whether the household lives in an urban area. 
We only included complete cases and individuals from the age of 18 as most of them already attained or are studying towards their highest education level. The final dataset contains 32,884 individuals.      

%%%%%%%%%%%%%%%%%%%%%%%%%%%%%%%%%%%%%%%%%%%%%%%%%%%%%%%%%%%%%%%%%%%%%%
%%%%%%%%%%%%%%%%%%%%%%%%%%%%%%%%%%%%%%%%%%%%%%%%%%%%%%%%%%%%%%%%%%%%%

\subsection{Model building}\label{sec:model_building}

Applying flexible bivariate copula GAMLSS requires the researcher to decide on the specification of multiple parameters, on the form of the continuous marginal distribution and of the copula.   

\subsubsection*{Continuous marginal distribution}

In line with, e.g. \citet{Klein.2015.multi} and \cite{Marra.2017}, we propose to use normalized quantile residuals for selecting the continuous marginal first. This allows us to assess graphically the appropriateness of the chosen distribution which could firstly be done using separate univariate models. A normalized quantile residual $\hat{q}_{mi}$ for the second, i.e. the continuous marginal, is defined as: 
\begin{align*}
\hat{q}_{2i} = \Phi^{-1}\{ \hat{F}_2 (y_{2i})\} \qquad \textrm{for } \, i= 1, \ldots, n,
\end{align*}
where $\hat{F}_2(\cdot)$ is the estimated marginal CDF for the continuous response component, and $\Phi^{-1}(\cdot)$ is the quantile function of a standard normal distribution. If $\hat{F}_2$ is close to the true distribution, $\hat{q}_{2i}$ approximately follows the standard normal distribution. Quantile residuals are fairly robust to the specification of the distribution parameters' specification \citep{Klein.2015.multi, Marra.2017}. %This can be assessed using function \texttt{post.check()} which produces the QQ-plot for the continuous margin of the bivariate model. 

We fit univariate models and select the model by inspecting the corresponding QQ-plots. Good distribution candidates  for income and expenditure are generally the lognormal distribution, the Singh-Maddala distribution and the Dagum distribution \citep[e.g.][]{Kleiber.2003}. Figure \ref{fig:qq} shows the QQ-plots for the univariate income model and all potential covariates using a lognormal and the Dagum distribution. Fitting the model with the Singh-Maddala distribution leads to converge failure, which may signal an inappropriate choice of the marginal distribution. The QQ-plots suggest an appropriate fit for the lognormal distribution and it is hence used. Note that once the final bivariate model is built, the QQ-plot for the continuous margin are re-examined. However, we find that the plot (shown in Appendix \ref{apx_application}) looks almost identical to Figure \ref{fig:qq} (left panel) indicating that a good fit for the continuous margin of the proposed copula model has been obtained.   
\begin{figure}[ht]
    \centering
    \begin{subfigure}{0.5\textwidth}
        \centering
        \includegraphics[scale=0.8]{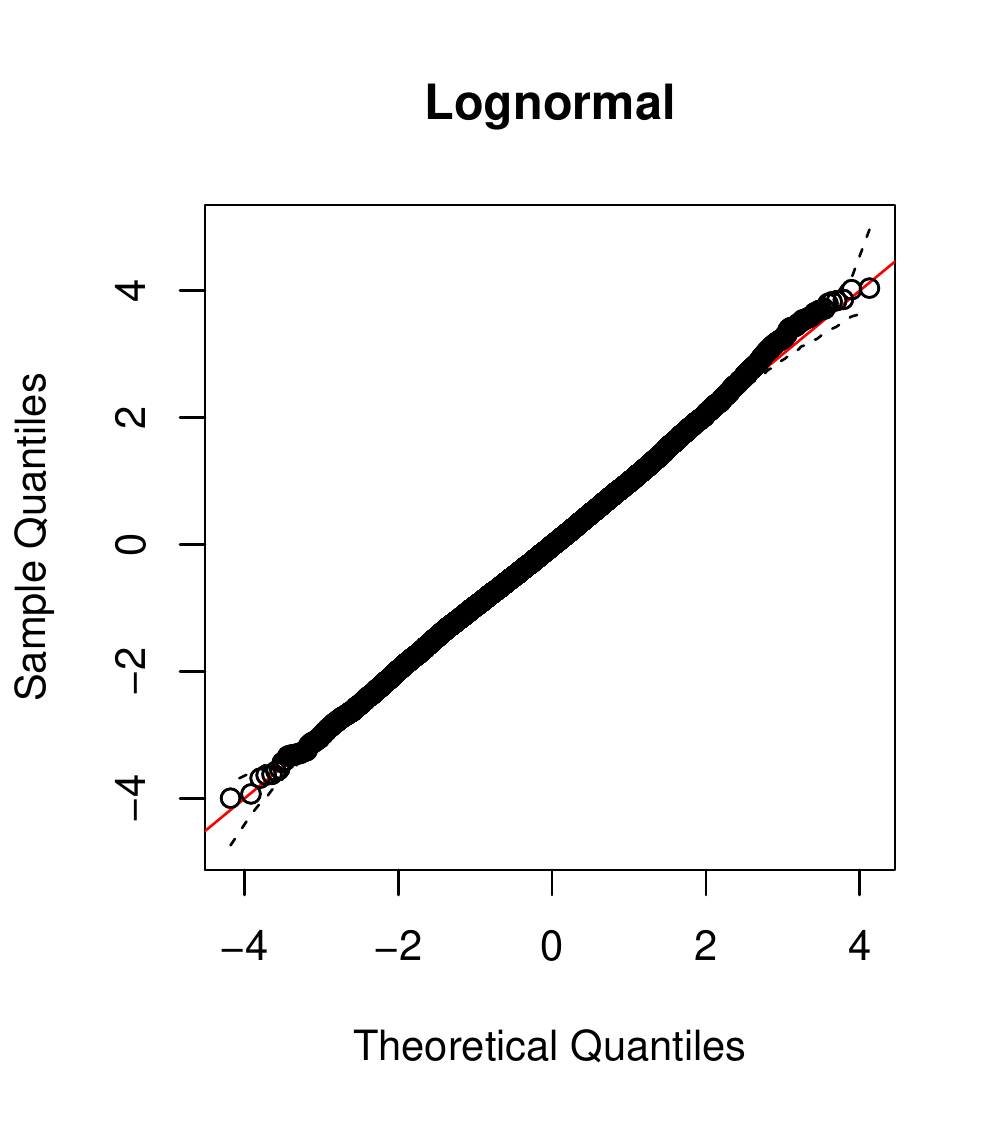}
        %\caption{lognormal}
    \end{subfigure}%
    \begin{subfigure}{0.5\textwidth}
        \centering
        \includegraphics[scale=0.8]{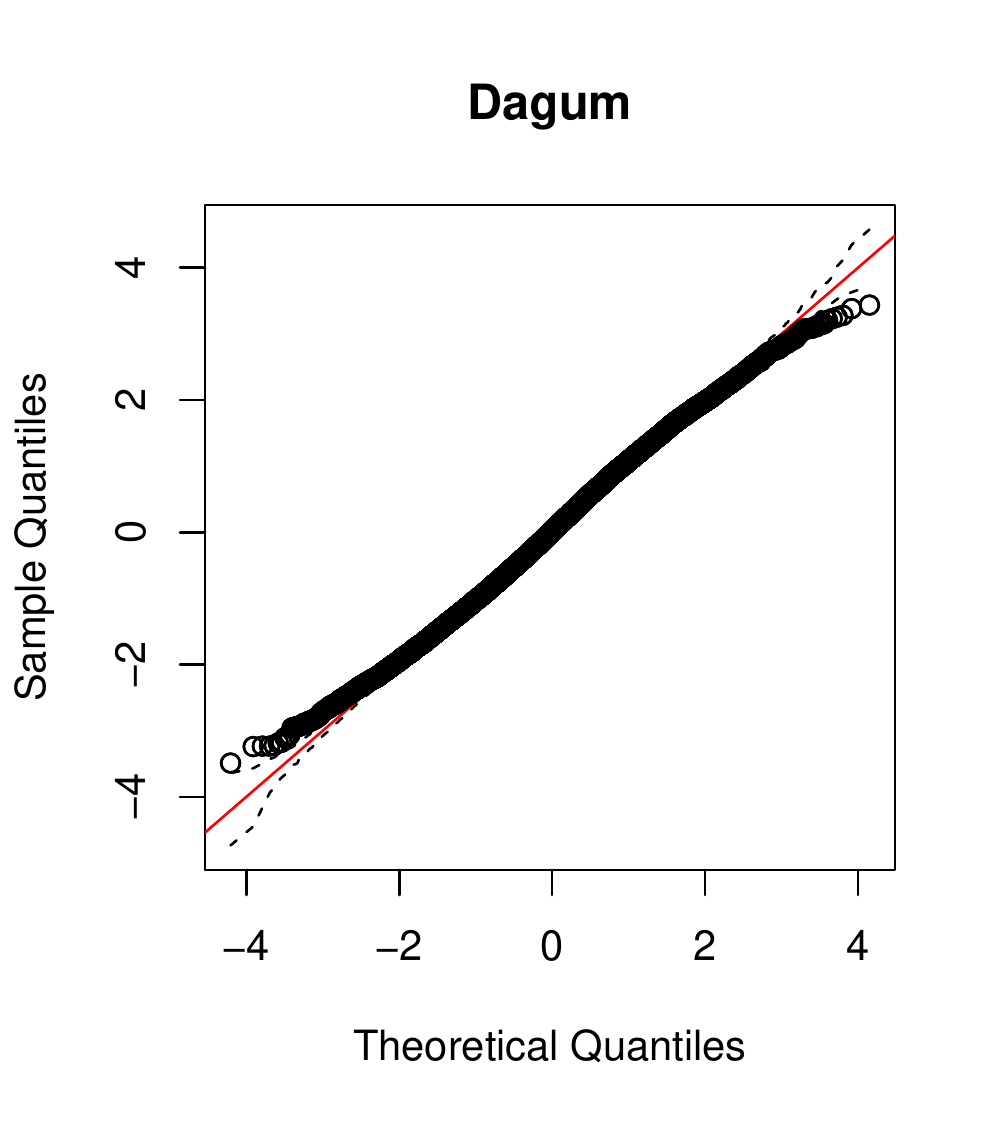}
        %\caption{dagum}
    \end{subfigure}
    \caption{Normal QQ-plots for the univariate income model and different distributions with 95\% reference bands.}\label{fig:qq}
\end{figure}

\subsubsection*{Variable selection}

%In a recent paper, \citet{Liu.2018} introduced a surrogate approach to analyse residuals in ordinal response models. Although their proposal could be employed to assess potential the impact of potential misspecifications with respect to mean structures, link functions and heteroscedasticity, such implementation to the present case is beyond the scope of this paper. 
For the specification of the link function of the ordinal response, as well as variable selection for the bivariate model, and the choice of the copula function, the Akaike's Information Criterion (AIC) and the Bayesian Information Criterion (BIC) can be used. These are defined as
\begin{align*}
\textrm{AIC} &:= -2 \ell(\hat{\boldsymbol \beta}) + 2\mathrm{edf}, \\
\textrm{BIC} &:= -2 \ell(\hat{\boldsymbol \beta}) + \log(n)\mathrm{edf},
\end{align*}
where $\ell(\hat{\boldsymbol \beta})$ is the log likelihood of the bivariate model evaluated at the penalized parameter estimate and $\mathrm{edf} = \mathrm{tr}(\hat{\boldsymbol{A}})$ as defined in Section \ref{sec:estim}. Theoretical knowledge about the problem at hand  facilitates the variable selection procedure by pre-selecting candidate predictors. \cite{Radice2016} also suggest to start with a model specification where all distributional and the association parameter depend on all covariates. In case the algorithm does not converge, an instance that often indicates that the sample size is too small for the model's complexity, they recommend trying out a series of more parsimonious specifications. To test smooth components for equality to zero, we have adapted the results of \citet{wood2017generalized} to the current context.

%In the bivariate case, the ordinal outcome and the copula parameter enter the model. 
To fit the bivariate model, we specify an equation for each distributional and the copula parameter %according to the strategies described in Section \ref{sec:model_building} 
as follows: We start with a set of variables selected according to economic reasoning. Note that often in income or expenditure equations, household size is used as a covariate in addition to number of children and elderly. However, we do not wish to separate the child effect in a ``pure'' child effect and children as additional household member effect. Moreover, the outcome variable is already adjusted for household size. Religion is included because it defines minority groups.
{While education on the individual level is part of the response vector, the level of education of the household head is included as a control in the predictor for  capita income.
For the bivariate model, we fit a full specification for the location parameter of each marginal distribution, i.e. ${\mu_1}_{educ}$ and ${\mu_2}_{inc}$ and perform variable selection using the AIC for the scale parameter  of the second marginal, ${\sigma_2}_{inc}$, and for the copula parameter, ${\gamma}$. More specifically, a backwards selection procedure is applied for ${\sigma_2}_{inc}$ given a full specification for ${\gamma}$ and then a second backwards selection is performed on ${\gamma}$ given the reduced model for ${\sigma_2}_{inc}$.}      
This excludes only four variables for the scale predictor and two variables for the copula parameter specification, i.e. we arrive at: %. We compare the full model and the reduced models by means of the AIC and BIC. The AIC and BIC support the following specification:  % based on economic reasoning and in line with the poverty literature (REFERENCE) as follows:
\begin{alignat*}{2}
\eta^{\mu_1}_{educ} &=  \theta_r - \{ s(age) + 
                     \beta^{\mu_1}_{1} \cdot hhmarstat + 
                     \beta^{\mu_1}_{2} \cdot hhmale + 
                     \beta^{\mu_1}_{3} \cdot urban   + \\
                     & \quad \, \beta^{\mu_1}_{4} \cdot num\_child + 
                     \beta^{\mu_1}_{5} \cdot elderly + 
                     \beta^{\mu_1}_{6} \cdot relig \}  \\
\eta^{\mu_2}_{inc} &=  \beta^{\mu_2}_0 + s(age) + 
                    \beta^{\mu_2}_{1} \cdot hhmarstat + 
                    \beta^{\mu_2}_{2} \cdot hhmale + 
                    \beta^{\mu_2}_{3} \cdot urban  + \\
                    & \quad \, \beta^{\mu_2}_{4} \cdot num\_child + 
                    \beta^{\mu_2}_{5} \cdot elderly + 
                    \beta^{\mu_2}_{6} \cdot relig  + 
                    \beta^{\mu_2}_{7} \cdot hheduc    + 
                    s(prov) \\
\eta^{\sigma_2}_{inc}  &=  \beta^{\sigma_2}_0 + s(age) + 
                    \beta^{\sigma_2}_{1} \cdot hhmarstat + 
                    \beta^{\sigma_2}_{2} \cdot num\_child +                    % & \quad \, \beta^{\sigma_2}_{3} \cdot elderly + 
                    \beta^{\sigma_2}_{3} \cdot relig + 
                   s(prov) \\
\eta^{\gamma} &=  \beta^{\gamma}_0 + s(age) + 
                    \beta^{\gamma}_{1} \cdot hhmarstat + 
                    \beta^{\gamma}_{2} \cdot urban  + \\
                    & \quad \, \beta^{\gamma}_{3} \cdot num\_child + 
                    \beta^{\gamma}_{4} \cdot elderly + 
                    \beta^{\gamma}_{5} \cdot hheduc + 
                    s(prov).
\end{alignat*}

Continuous variables enter the equations with smooth non-parameteric effects \textit{s()} represented via thin plate regression splines with ten bases and second order derivative penalties. Spatial effects of the provinces and their neighbourhood structure are modeled using Markov random fields. We choose to model the spatial effect at the province level since minimum wages are set at the province level affecting individual wages and thus expenditure measure as well \citep{Hohberg.2015}.  

\subsubsection*{Ordinal model}

The ordinal outcome education is fitted using an ordered model. Table \ref{tab:probit} compares the AIC and BIC between a probit and a logit link of the bivariate model using the lognormal as the continuous marginal and a Gaussian copula. Both AIC and BIC favor the logit model for the first marginal.  

\begin{table}[ht]
    \centering
    \caption{AIC and BIC of bivariate ordered-continuous model using the logit and probit links.}
    \begin{tabular}{lcc}
    \hline
         & AIC & BIC  \\ 
         \hline
    logit & 1,075,191 & 1,076,241\\
    probit & 1,075,588 & 1,076,637 \\
    \hline 
    \end{tabular}
    \label{tab:probit}
\end{table}

\subsubsection*{Choice of the copula}

For the copula selection, a good starting point would be the use of a Gaussian copula and then consider all consistent alternatives depending on the direction of the dependence \citep{Radice2016, klein.marra.2019}. Again, AIC and BIC can help choosing among several candidate copulas.    

%To find an appropriate copula, we 
Starting off with the Gaussian yields an average value for the copula parameter (with 95\% confidence interval in brackets) of $\gamma = 0.163~(0.104,0.221)$. Building on this finding, we test a range of suitable possible candidates. After checking convergence, we can only eliminate the un-rotated Joe and Clayton copula.  
The remaining candidates  are  compared using the AIC and BIC; see Table \ref{tab:cop_AIC} for the results. The AIC and BIC indicate that a Gaussian copula should be used for our model, and all copula models should be favoured over the independence model. Using the Gaussian copula suggests that the dependence between per capita expenditures and education is symmetric with asymptotically independent extremes. 

\begin{table}[ht]
\centering
\caption{AIC and BIC for different copula specifications.} \label{tab:cop_AIC}
\begin{tabular}{rrr}
  \toprule
 & AIC & BIC \\ 
  \midrule
 Gaussian & 1,075,191 & 1,076,241 \\ 
  F & 1,075,233 & 1,076,295 \\ 
  FGM & 1,075,280 & 1,076,335 \\ 
  PL & 1,075,226 & 1,076,286 \\ 
  AMH & 1,075,298 & 1,076,359 \\ 
  C0 & 1,075,448 & 1,076,508 \\ 
  C180 & 1,075,379 & 1,076,380 \\ 
  J0 & 1,075,514 & 1,076,462 \\ 
  J180 & 1,075,516 & 1,076,573 \\ 
  G0 & 1,075,351 & 1,076,360 \\ 
  G180 & 1,075,306 & 1,076,341 \\ 
  Independence & 1,075,892 & 1,076,678 \\
    \bottomrule
\multicolumn{3}{p{8cm}}{ Note:  Abbreviations correspond to Frank, Farlie-Gumbel-Morgenstern, Plackett, Ali-Mikhail-Haq, Clayton, rotated Clayton (180 degrees),  Joe, rotated Joe (180 degrees), Gumbel, rotated Gumbel (180 degrees), respectively.
}  
\end{tabular}
\end{table}

%After fitting the final model and checking its convergence, the fit is assessed graphically using the \texttt{post.check()} function. The plots in Figure \ref{fig:mod_check} show that a good fit for the continuous margin of the proposed copula model has been obtained. 

%\begin{figure}[h!]
%    \centering
%           \includegraphics[scale=0.6]{figures/post_check_edu.pdf}
%        \caption{Histogram and normal Q-Q plots for the log-normal continuous margin of the chosen copula model.}\label{fig:mod_check}
%\end{figure}

\subsection{Model evaluation}
After deciding on the marginal distributions, the copula and covariates by comparing different candidates, we check the final bivariate model. To this end, we use a multivariate generalization of the quantile residuals introduced in Section \ref{sec:model_building} that was proposed by \citep{Kalliovirta.2008}. 
Multivariate quantile residuals for two continuous responses are defined as
\begin{align*}
\hat{\boldsymbol{q}}_i = 
\begin{pmatrix}
\hat{q}_{1i} \\ \hat{q}_{2i}
\end{pmatrix} 
= 
\begin{pmatrix}
\Phi^{-1}(\hat{F}_1(y_{1i})) \\ \Phi^{-1}(\hat{F}_{2|1}(y_{2i}|y_{1i})), 
\end{pmatrix} 
\end{align*}
where $\hat{F}_{2|1}$ is the (estimated) conditional CDF of $Y_2$ given $Y_1$. In our case, the first marginal is discrete such that we resort to randomized quantile residuals where uniformly distributed random variables on the interval corresponding to cumulated probabilities are plugged into $\Phi^{-1}(\cdot)$.
If the model is correctly specified, then $\hat{\boldsymbol{q}}$ approximately follows a bivariate standard normal distribution.

The contour plot for the bivariate model in Figure \ref{fig:mod_eval} shows the density of the quantile residuals $\hat{\boldsymbol{q}}$ by means of a multivariate kernel density estimator. This estimated density is compared to the density of the standard normal distribution. The contour lines of both densities are close to each other indicating a good fit of the bivariate copula model.      
\begin{figure}[h]
  \begin{subfigure}[t]{0.5\textwidth}
   \centering
           \includegraphics[scale=0.75]{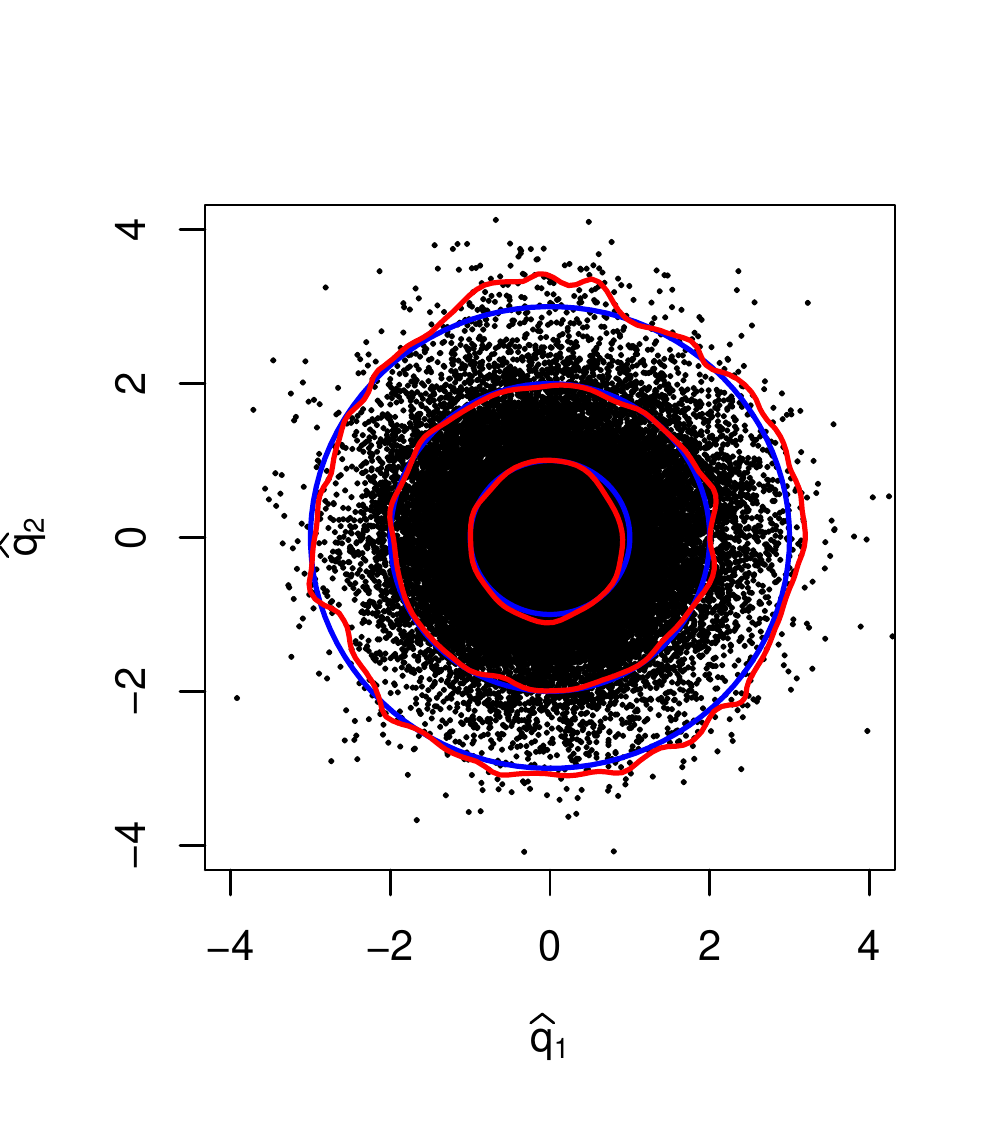}
        \caption{Contour plot of multivariate quantile residuals. The red lines indicate the density of the quantile residuals estimated by a multivariate kernel density estimator. The blue circles are the contour lines of the density of the standard normal distribution with radius $1,2$ and $3$.  }\label{fig:mod_eval}
\end{subfigure}\hspace{3mm}
\begin{subfigure}[t]{0.5\textwidth}
 \centering
           \includegraphics[scale=0.75]{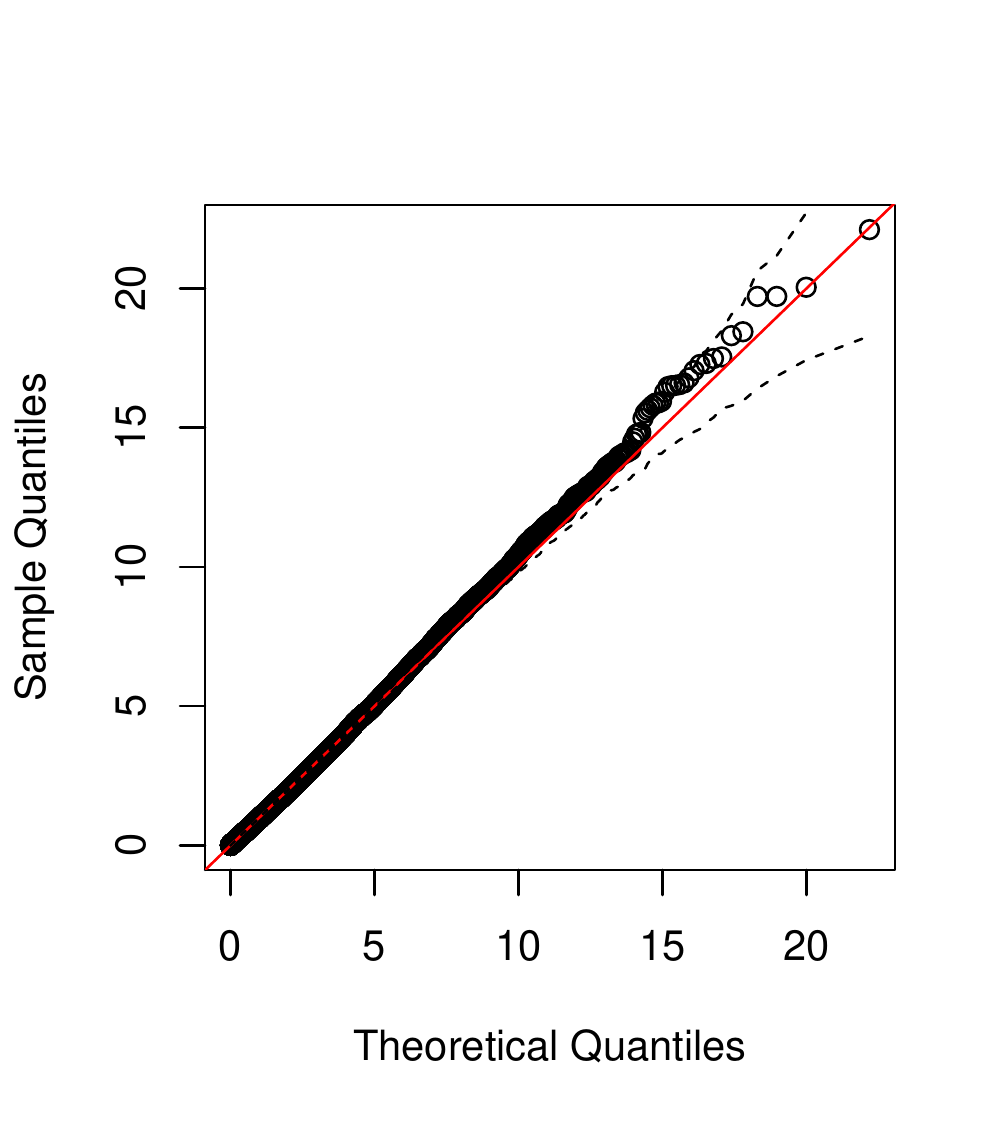}
       \caption{QQ-plot depicting the sum of the squared elements of the multivariate quantile residuals with 95\% reference bands.}\label{fig:mod_eval2}
\end{subfigure}
\caption{Multivariate quantile residuals of the bivariate model.}
\end{figure}

In Figure \ref{fig:mod_eval2}, the sum of the squared elements of the multivariate quantile residuals are considered. That is, $\hat{\boldsymbol{q}}_i^{\prime} \hat{\boldsymbol{q}}_i = \hat{q}_{1i}^2 + \hat{q}_{2i}^2,$ where $\hat{\boldsymbol{q}}_i$ is the multivariate quantile residual for the $i$-th individual. Since $\hat{q}_{1i} \stackrel{a}{\sim} \mathcal{N}(0,1)$ and $\hat{q}_{2i} \stackrel{a}{\sim} \mathcal{N}(0,1)$, it follows that $\hat{\boldsymbol{q}}_i^{\prime} \hat{\boldsymbol{q}}_i \stackrel{a}{\sim} \chi^2(2)$ which is assessed in the  QQ-plot in Figure \ref{fig:mod_eval2}.
The reference bands are obtained by repeatedly  simulating from a $\chi^2(2)$ distribution. We draw $n_{rep} = 100$ samples and compute the $2.5\%$ and $97.5\%$ quantiles for each observations across the sorted $n_{rep}$ samples. The plot supports our model choice.

%%%%%%%%%%%%%%%%%%%%%%%%%%%%%%%%%%%%%%%%%%%%%%%%%%%%%%%%%%%%%%%%%%%%%%
%%%%%%%%%%%%%%%%%%%%%%%%%%%%%%%%%%%%%%%%%%%%%%%%%%%%%%%%%%%%%%%%%%%%%
\newpage

\subsection{Results}

This section demonstrates the results that poverty researchers can derive from fitting a copula GAMLSS model for the joint analysis of two inter-related poverty dimensions. We summarise briefly the covariate effects on the marginals before taking a closer look at the dependence structure and risk groups. This way, we are able to answer questions on how dependence and risk are affected by a household's location or composition.     

\subsubsection*{Effects of covariates on the marginal distributions}

The full list of effects on each parameter of the marginals and on the copula parameter is included in the Appendix \ref{apx_application}. Note that the effects are subject to a \textit{ceteris paribus} interpretation. For example, more schooling is associated with more income with tertiary education having the highest effect. Households with more children or elderly people living in the household  have on average a lower income per capita and less education (except for the effect of one elderly on education which is negative and not statistically significant). Urban households are associated with both better income and better education. A little surprising is that living in a household where the head is married is correlated with a reduction in income and education compared to not (yet) married households (Table~\ref{tab:income} and \ref{tab:educ}). One possible explanation is that non-married households compared to married households include a larger share of young, single persons that do not need to share their income and have a comparatively higher level of education. Non-Muslim households are associated with higher income per capita and more education although they represent a minority in Indonesia. For a male household head, the effects are not as expected since it is negative for income but positive for education. For the second parameter of the continuous margin, i.e. the scale parameter for the income equation, the number of children and a marital status other than not married have negative effects while the effects of elderly and other religions compared to Islam are positive. All of the covariates have a positive effect on the copula parameter though not all effects are significant.  

Age is modeled in a non-linear way and Figure \ref{fig:smooth_age} displays the smooth effects of age on each parameter. Education attainments are lower for higher ages which can be explained by the education expansion that Indonesia has undergone since the 1970's and younger individuals benefited from. For example, between 1974 and 1978 over 61,000 primary schools were built. In 1984, compulsory education was set to six years which was extended to nine years in 1994 \citep{Akita.2017, duflo.2004}.   
The effect of income on the location parameter is inverted u-shaped until the age of 60 with a peak around the age of 40. After 80, the confidence intervals become very wide due to a lower number of observations in this age span and the effect is thus less clear. The effects on the scale parameter are around zero. For the copula parameter, the age effect indicates that the dependence is decreasing for individuals up to their mid 30ies and stays around zero afterwards until it decreases again after the age of 60.

\begin{figure}[h!]
            \centering 
          \includegraphics[scale=0.8]{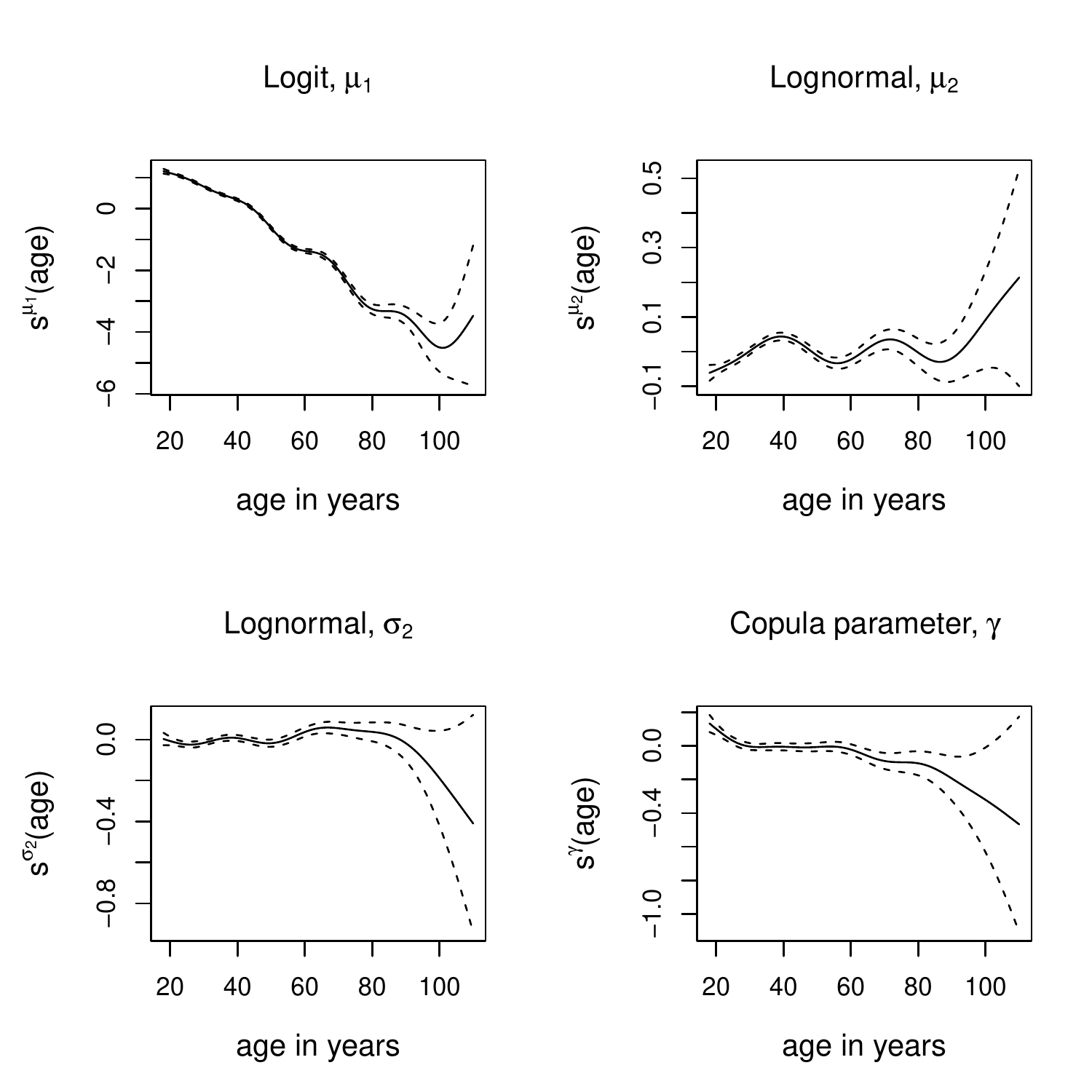}
            \caption
        {Estimated smooth functions of age and respective point-wise 95\% confidence intervals.}\label{fig:smooth_age}
    \end{figure}

Figure \ref{fig:spat_para } shows the effect of the underlying spatial pattern on the parameters $\mu_2$ and $\sigma_2$ of the second response, and on the copula parameter $\gamma$.  
%\textcolor{red}{für die erste equation gibts in gam einen fehler wenn spatial mit drin is. wir damit umgehen?} 
Households located in provinces of Java seem to have higher income per capita compared to the observed provinces in Sumatra, Borneo and Sulawesi. Provinces with a negative effect on the location parameter have higher effects on the scale parameter except for Borneo whose scale effect is negative.  

\begin{figure}
            \centering 
            \includegraphics[scale=1.2]{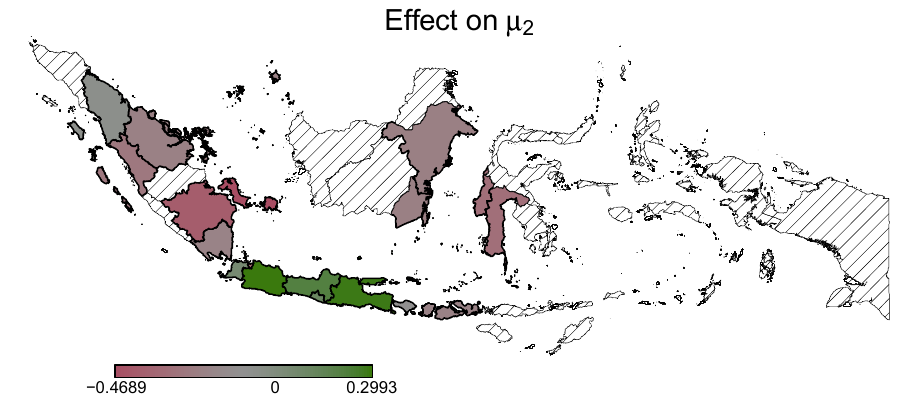}
            \includegraphics[scale=1.2]{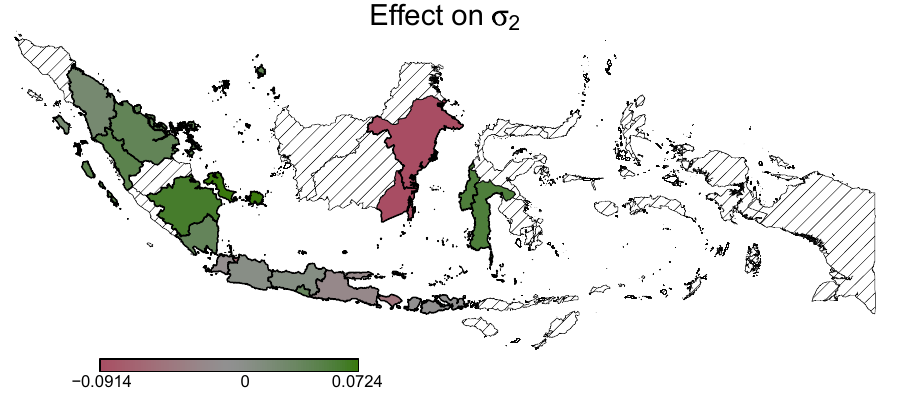}
            \includegraphics[scale=1.2]{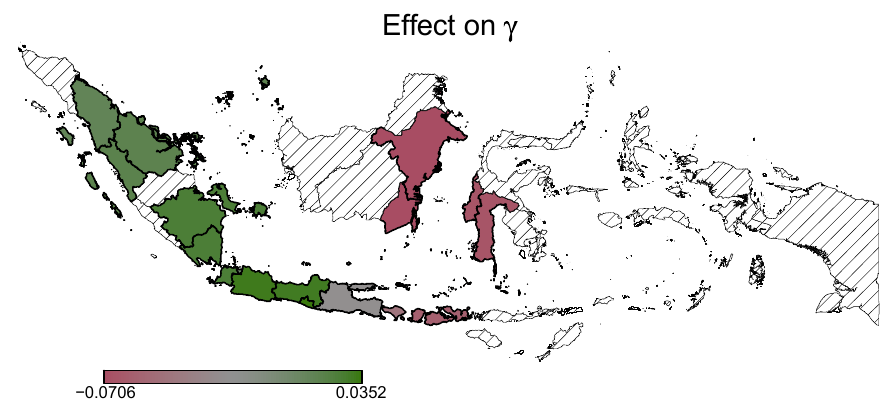}
            \caption {Spatial effects of the provinces on the distribution parameters of income and on the copula parameter. 
            } \label{fig:spat_para }
    \end{figure}

\subsubsection*{Dependence structure}

When focusing on the dependence, its structure can be represented via contour plots for specific covariate combinations. For example, we focus here on the location of the household (urban/rural and province) as they are the significant drivers of the copula parameter, while the remaining plots can be found in Appendix \ref{apx_application}. For comparison purposes, we include provinces with the highest frequencies in the dataset but select only one of Java's provinces. To compare the dependence structure across different locations, we create an example of typical individual whose characteristics, other than the one under consideration, are set to their mean value or to their most frequent observation. The only exception is the education of the household head which is set to the second most frequent observation. For a household head with "high school" degree the dependence is a bit more pronounced we hence selected this level for demonstration purposes. This is the covariates' combination that we call an ``example individual'' henceforward. 

Figure \ref{fig:contour_urban} shows that the dependence is stronger for individuals in urban households compared to rural households. One reason might be that average education levels are lower in rural areas (x-axis) while at the same time high paid job opportunities are restricted in a rural environment, resulting in more equal incomes compared to an urban environment.  
Figure \ref{fig:contour_prov} compares the dependence structure across selected provinces. The dependence seems weakest in the province of Nusa Tenggara Barat, which is one of the poorest provinces in Indonesia. It is surprising that the average per capita consumption (straight line) of Jakarta is about the same level than that one of Nusa Tenggara Barat. Most likely this is due to the high price level in Jakarta and the deflation measure applied which scaled down our expenditure measure maybe too drastically. Though the copula coefficient for Jakarta is similar to Jawa Timur, the latter has more variation in incomes and the contour levels lie further apart. Considering all -- and not just the selected provinces -- the value of the copula parameter for the example individual ranges from 0.16 for Kalimantan Timur and Kalimantan Selatan up to 0.27 for Yogyakarta.  
%structure in Jawa Barat is similar to the one in Yogyakarta which might be explained by their close location. Jakarta Raya that is also located on the island of Java does not show much dependence between the two outcomes. The plots in which a positive dependence is visible, exhibit stronger dependence for smaller incomes and lower education compared to higher incomes and higher education. \textcolor{red}{when using clayton, plots look almost same }. The dependence structure for Sumatara Utara seems to lie in between Jakarta and Jawa Barat.  

\begin{figure}
\centering
\begin{subfigure}{0.8\textwidth}
\hspace{-4mm}\includegraphics[scale=0.71]{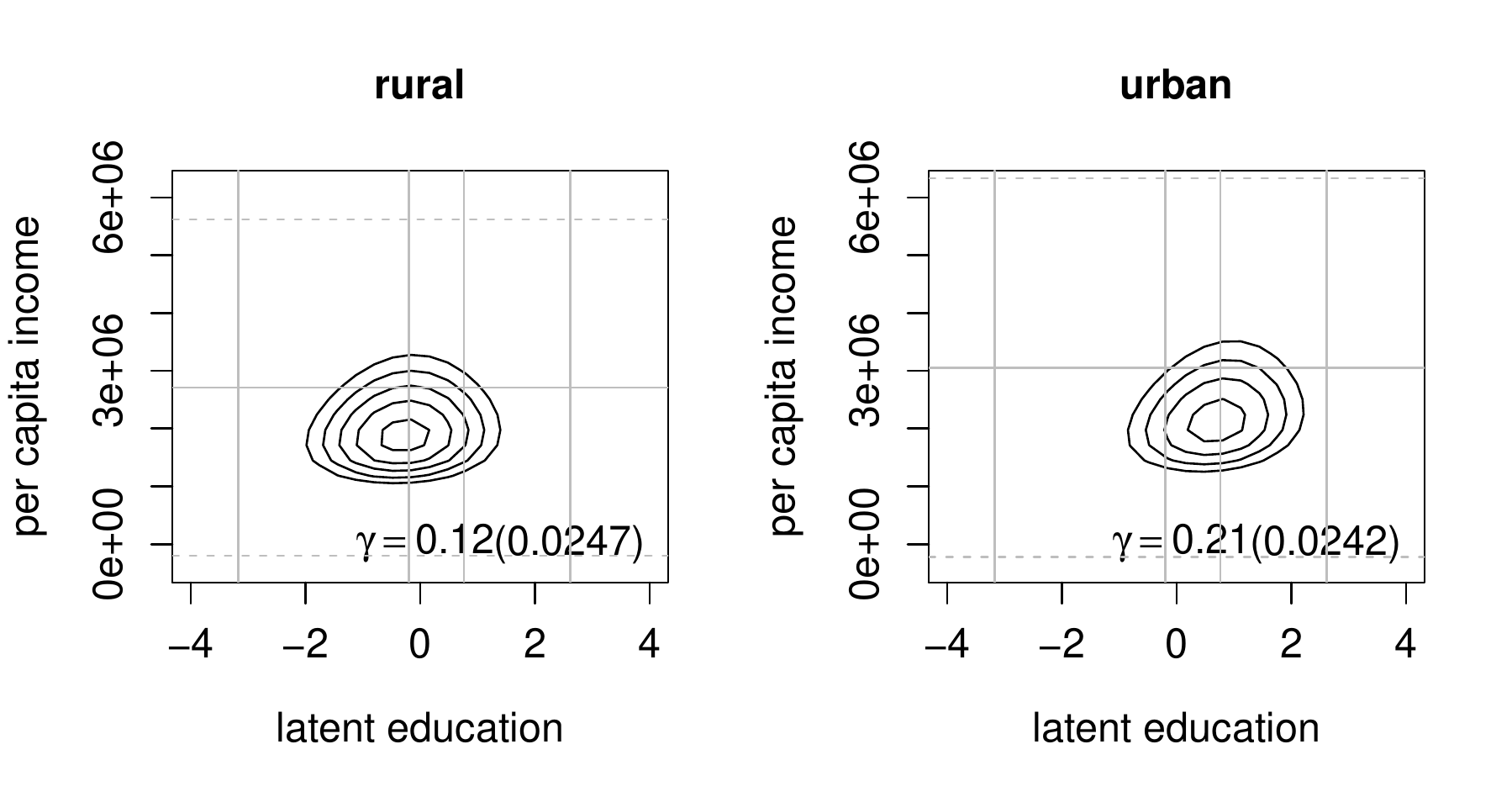}   
\caption{Contour plots for an example individual in an urban or rural household in the province of Jawa Timur. \label{fig:contour_urban} }
\end{subfigure} 
%\vskip 12pt
%\begin{minipage}{\textwidth}
%    \captionof{figure}{Contour plots for an average individual in an urban or rural household \label{fig:contour_urban}}
%\end{minipage}
\begin{subfigure}{0.8\textwidth}
\includegraphics[scale=0.7]{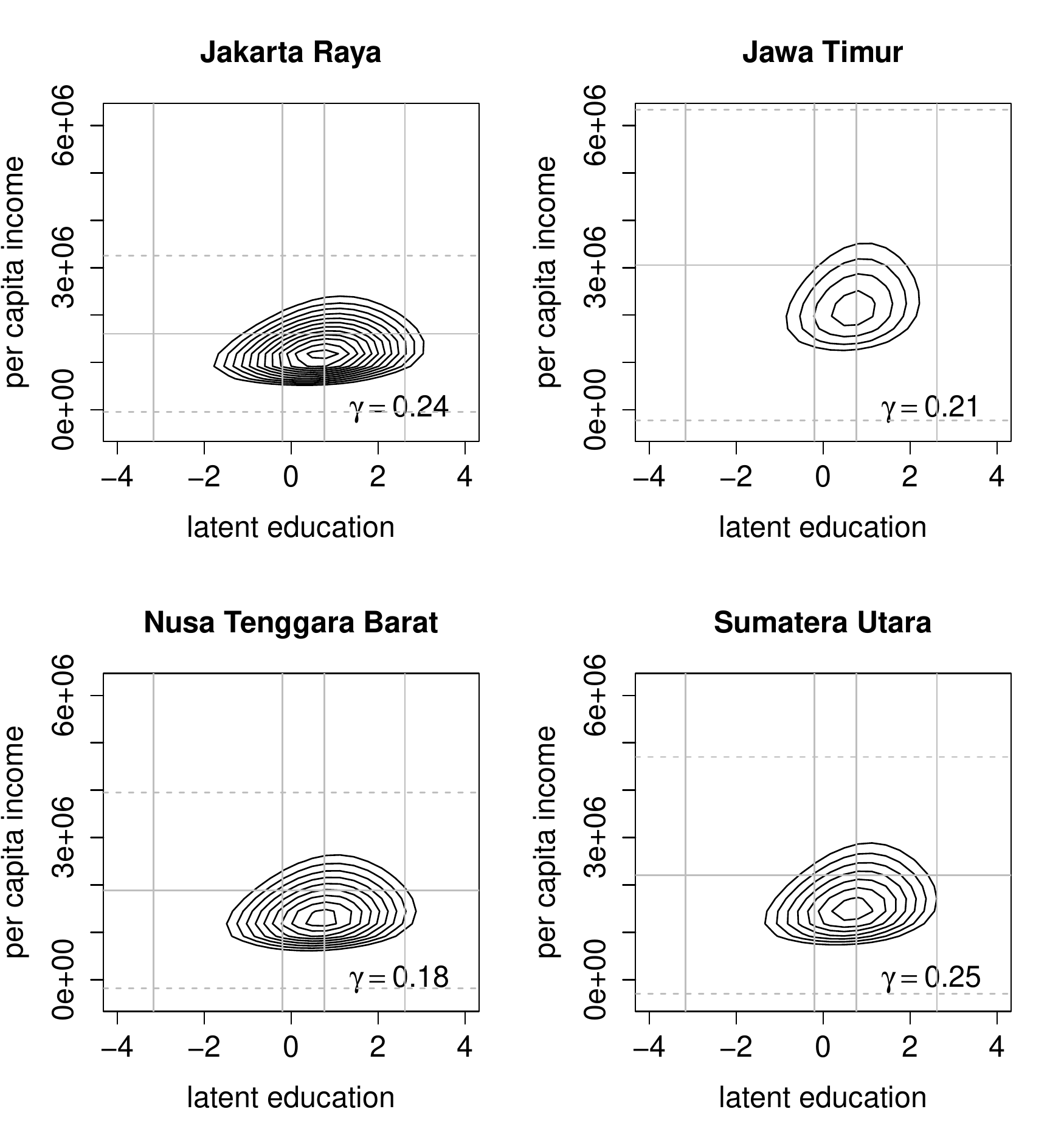}
\caption{Contour plots for an example individual in an urban household in different provinces. \label{fig:contour_prov}}
\end{subfigure}
\caption{Contour plots for (education, income)' and a Gaussian copula by households' location. Contour lines of densities are at levels from 0.00000005 to 0.00000025 in 0.00000001 steps. The vertical straight lines represents the cut off values for the education categories, horizontal straight lines are the consumption average, and dashed horizontal line are at two standard deviations around this average. }
\end{figure}
%\vskip 12pt
%\begin{minipage}{\textwidth}
%    \captionof{figure}{Contour plots for an average individual in different provinces\label{fig:contour_prov}}
%\end{minipage}

\begin{figure}[ht]
\includegraphics[scale=1.6]{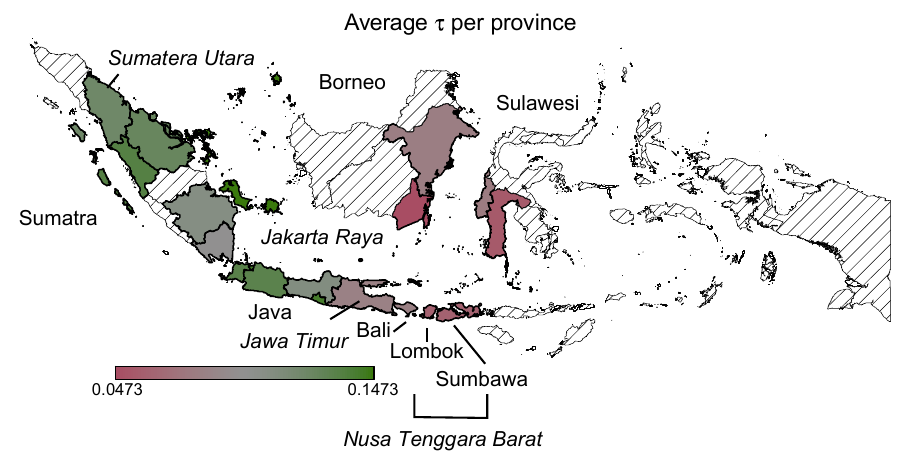}
\caption{Kendall's $\tau$ for each individual averaged within provinces. %\textcolor{red}{XXXX FD: Aren't we using the Gaussian copula? Why mapping the correlation coefficient to the Kendall's tau?} --> i think it's more standard to use tau, also it has some advantages (invariant to monotonic transformations)
\label{fig:kendall}}
\end{figure}

A policy maker might be interested to know in which locations each individual in the dataset, and not the example individual, have higher dependencies in order to efficiently design policy strategies. Although the lower panel of Figure \ref{fig:spat_para } shows the effect of the provinces on the copula parameter, it might be more helpful for interpretation to transform it into the Kendall's $\tau$, an association measure that takes on values on $[-1,1]$. Each individual with his/her specific covariates' combination is related to an individual-specific $\tau$. One way to present the differences across provinces is to average the $\tau$ over all individuals in a particular province. This is shown in Figure \ref{fig:kendall}. The Kalimantan Selatan (South Borneo) is the province with the lowest average of Kendall's $\tau$ with a value of 0.0468 and Kepulauan Riau (Riau Islands, northwest of Borneo) has a value of 0.1467 which is the highest average value that also indicates spatial heterogeneity in the strength of the dependence. The provinces of Sumatra seem to have higher dependence between income and education than provinces in Borneo or Sulawesi. %XXXX FD: are provinces of Sumatra richer than the other mentioned? XXXX --> some are, some are not .. 
Interestingly, for Java and its neighbouring smaller islands on the east, the dependence seem to decrease from west to east.

\subsubsection*{Joint probabilities}

Other results we can derive from a copula GAMLSS models are joint probabilities for different sub-groups. That is, we calculate the probability for the example individual of being poor in both the education and income dimensions. As an example, we again focus here on household location and additionally consider  household composition. To define poverty, we classify individuals that have only primary or even less education as education poor and set a relative poverty line of 60\% of the median of unique values of per capita expenditures. Note that Indonesia also has a national absolute poverty line that is, however, based on a different expenditure measure than the one we constructed from the IFLS data. We thus decided to use a relative poverty line. An individual that is poor in both dimensions has expected values below each of these thresholds. One of the sub-groups is set to one and the probability of the other sub-group is compared to this base category. Figure \ref{fig:joint_probs} shows that the probability for being poor in both dimensions is 2 times higher for the example individual in a rural household compared to the same individual in an urban household. Compared to Jawa Timur the joint probabilities of Jakarta Raya, Nusa Tenggara Barat, and Sumatera Utara are about 8 times, 6 times, and 4 times higher, respectively. Not surprisingly, the risk of being poor in both dimensions increases with the number of children and elderly in the household.

\begin{figure}
\centering
\includegraphics[scale=0.5]{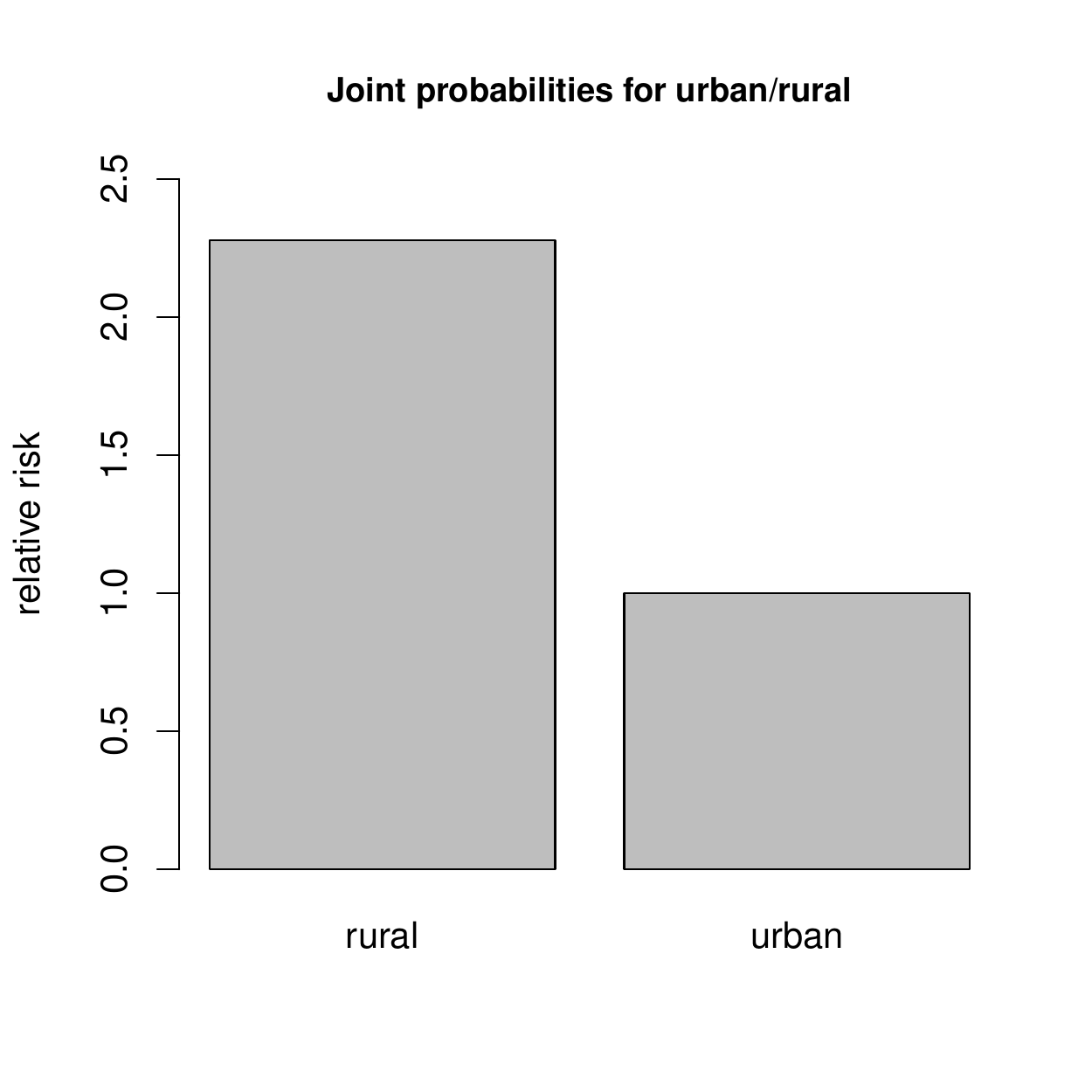} 
 \includegraphics[scale=0.5]{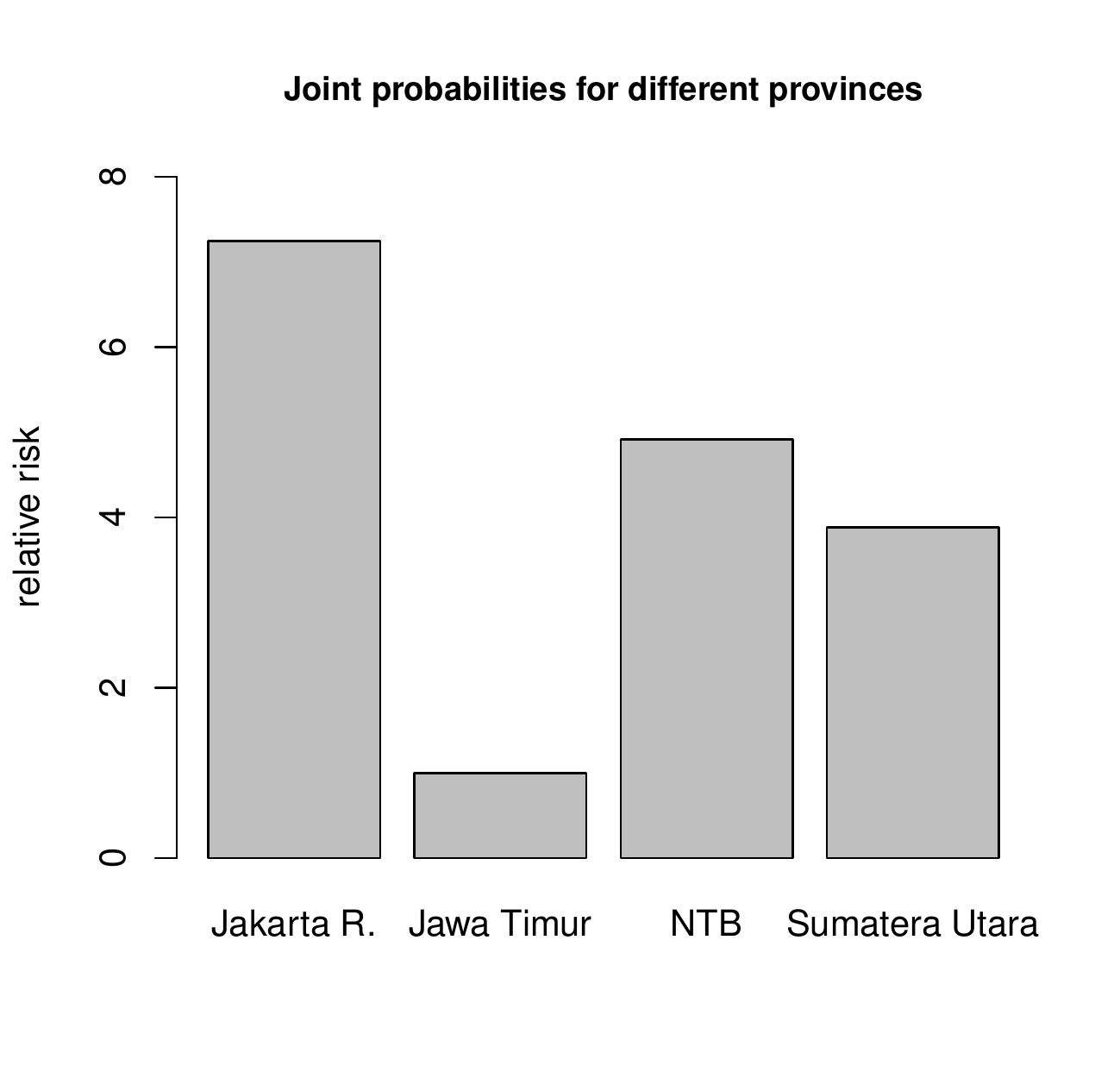} 
\\ %\vskip -20pt 
\includegraphics[scale=0.5]{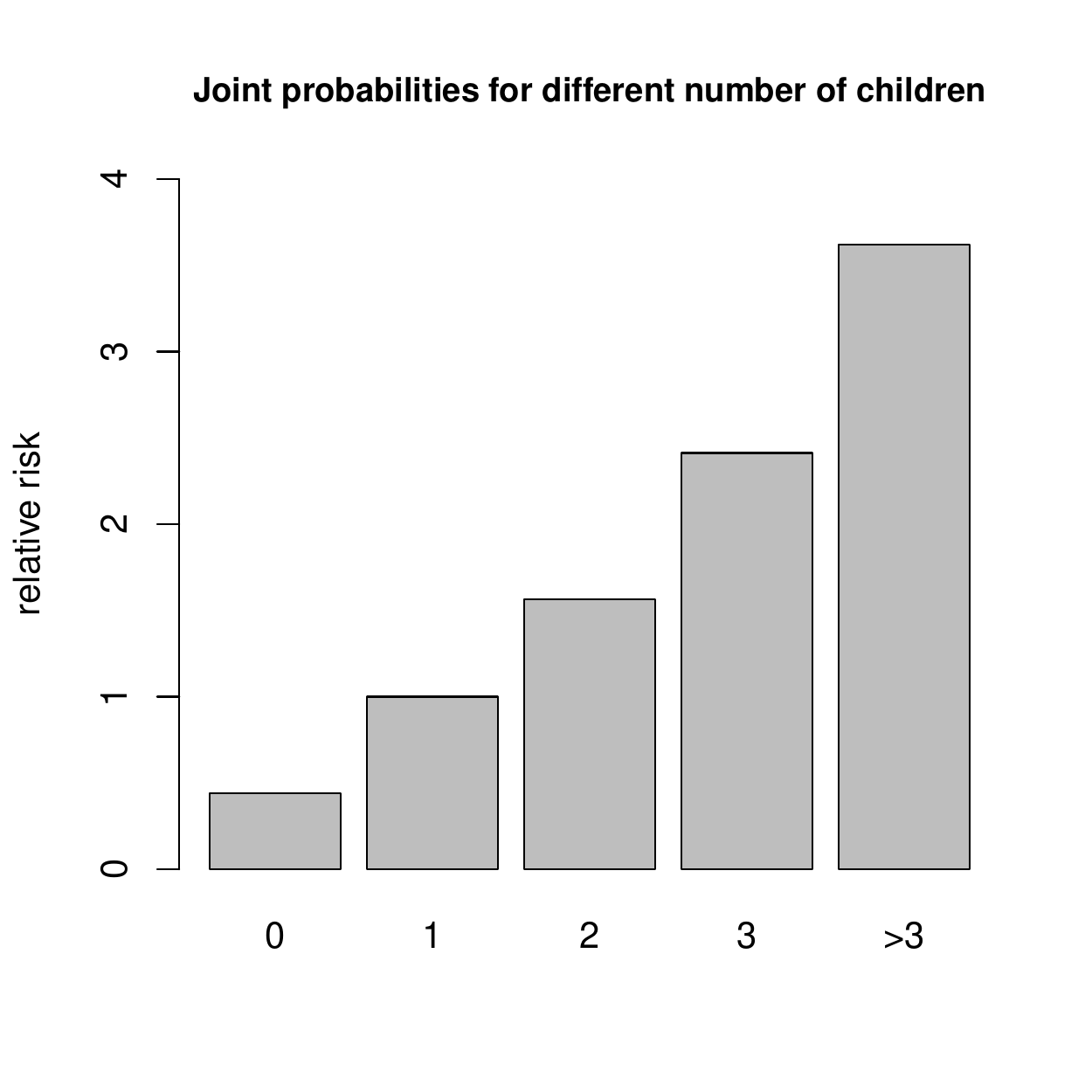}
\includegraphics[scale=0.5]{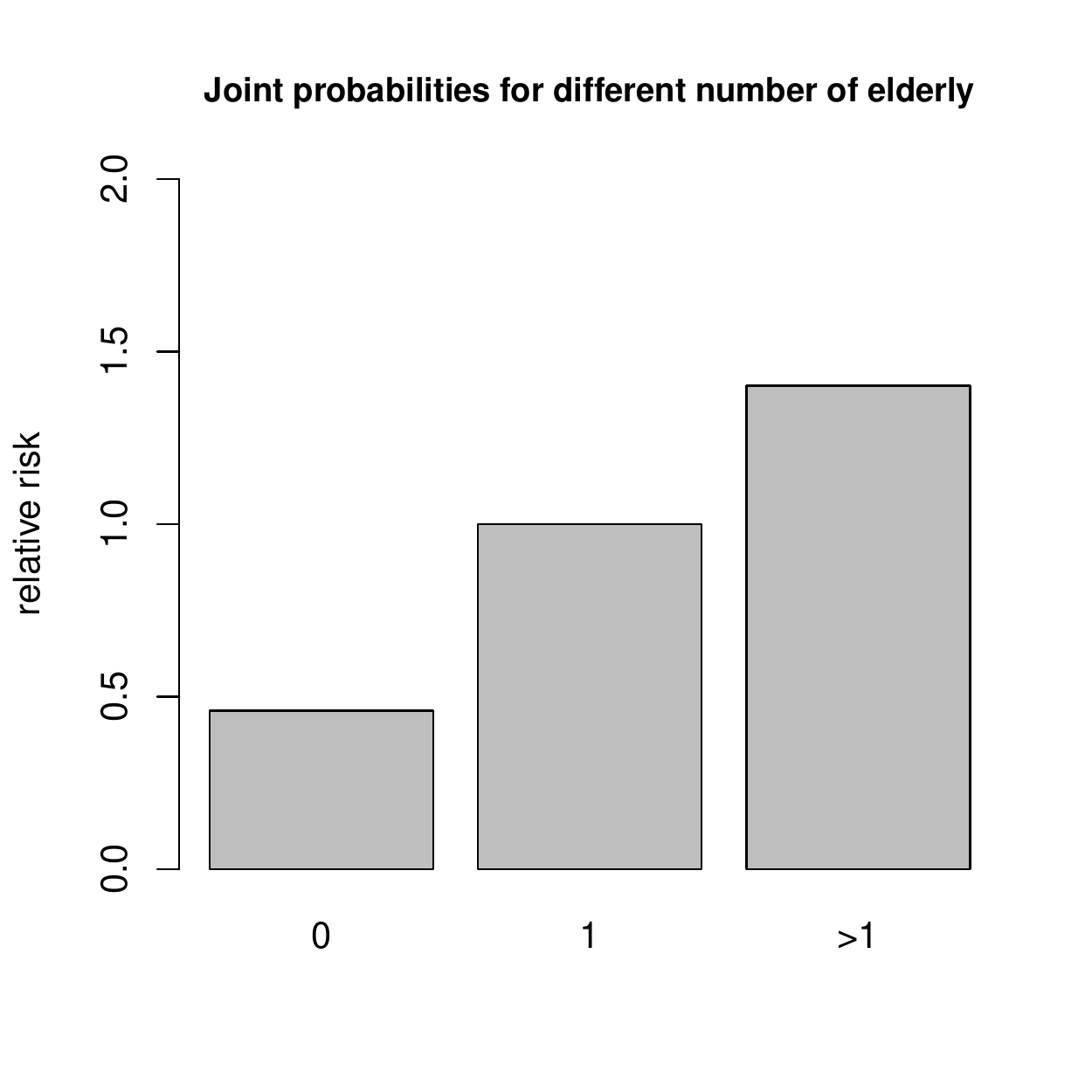}
\caption{Relative joint poverty risks of an example individual differentiated by household location and household composition. Baseline categories are urban, Jawa Timur, 1 child, and 1 elderly, respectively, and are set to one. The abbreviation NTB denotes the province of Nusa Tenggara Barat. \label{fig:joint_probs}}
\end{figure}

\subsubsection*{Vulnerability to poverty}

The higher the dependence between education and income, the higher the chances that we miss some individuals at risk by only looking at the marginal distributions of each poverty dimension. To identify the individuals at risk, we calculate the probabilities of being poor in two dimensions for each individual in the dataset, first for the independence model and then for the copula model. A vulnerability threshold is arbitrarily defined at a probability of 0.1 for simplicity. All individuals with a probability of being poor above this threshold are declared as vulnerable. We then compare the individuals that are identified as vulnerable by the copula and the independence models to their actual poverty status and calculate the specificity or true positive rate. Results are displayed in Table \ref{tab:specificity}. We find that the copula model has better specificity for two-dimensional poverty than the independence model although the difference is fairly small. %\textcolor{red}{Therefore, in this example the copula model does not really serve the purpose of predictions but should be used to analyse the dependence structure of the marginal distributions. XXXX FD: I would not mention that: I don't think it's nice to close the paper with something not so positive. It's good to state the results for full disclosure, but I think we can omit the judgment. Also the results consider only aggregate data. For example how different will be the specificity for the provinces/groups of individuals with the highest dependence? XXXX --> agree, I tried the same for a province with high dependencies, results are not better :(} 
 
\begin{table}[ht]
    \centering
   \caption{Specificity of the copula and independence models. Individuals that are declared vulnerable are compared to their actual poverty status.}
   \begin{tabular}{rrr}
  \hline
 & copula & independence \\ 
  \hline
income dimension & 0.88 & 0.88 \\ 
  both dimensions & 0.66 & 0.63 \\ 
   \hline
\end{tabular}
    \label{tab:specificity}
\end{table}

\subsubsection*{Discussion of Results}

{The main advantage of applying copula GAMLSS in the precedent poverty analysis is that we are able to analyse the dependence between poverty dimensions in more detail than studies that are limited to measuring the dependence. Even though the dependence is not very strong, once we control for covariates in the marginals, we consider this an interesting outcome.} We could further identify heterogeneities in the strength of the dependence between income and education, and poverty risk with respect to a household's location. For example, an example individual in a rural household exhibit lower dependencies compared to the same individual in an urban household. One explanation might be that opportunities in rural areas are restricted and thus the individual's education level has a smaller influence on per capita expenditures. 

 There is also a strong spatial heterogeneity between provinces. High dependencies can be found in the northwest while the values decrease for the more central provinces. The explanation for low dependence due to low opportunities could also apply to the province of Nusa Tenggara Barat which showed little average dependence and lower dependence for a specific example individual compared to other provinces. Nusa Tenggara Bararat is one of the poorest provinces, with a low GDP per capita and with agriculture and fishery being the most important industries. On the other hand, provinces such as Jakarta and Kalimantan Timur, that have the highest per capita GDP out of all provinces in Indonesia, have higher probabilities of being poor in both dimensions compared to most other provinces. These probabilities are calculated for an example individual that have average characteristics and a high school degree. Possible reasons for the discrepancy between rich provinces and high relative poverty risks might be that in Jakarta income and consumption are very unequally distributed and in East Kalimantan the high GDP is a result of high natural resource exploitation which yields little benefit for the population \citep[see for example][who found that growth in the mining sector has had no effect on poverty and inequality in Indonesia]{BHATTACHARYYA2015}.

%%%%%%%%%%%%%%%%%%%%%%%%%%%%%%%%%%%%%%%%%%%%%%%%%%%%%%%%%%%%%%%%%%%%%%
%%%%%%%%%%%%%%%%%%%%%%%%%%%%%%%%%%%%%%%%%%%%%%%%%%%%%%%%%%%%%%%%%%%%%

\section{Conclusion} \label{sec:concl}

Though poverty is conventionally regarded a multidimensional phenomenon, regression analyses in the poverty context either examine dimensions separately or use a scalar index such as the Multidimensional Poverty Index as an outcome. Both approaches neglect the dependence between poverty dimensions. This paper presents an alternative model for an in-depth poverty study that explicitly analyses this dependence and its determinants. It is important to understand what drives the dependence between poverty dimensions since high dependencies can explain persisting poverty.       

For this type of poverty analysis, we propose the use of bivariate copula GAMLSS which relate each distributional parameter of the marginals and of the copula parameter to flexible covariate effects. Since poverty analyses often include one monetary measure such as income or consumption and some ordinal measure such as education level or health status, we extended the class of copula GAMLSS to incorporate ordinal data based on a latent variable approach. This extension has been incorporated in the \texttt{GJRM} \texttt{R}-package.     

We use data from Indonesia to show how copula GAMLSS can be applied to a poverty analysis. The model identified  
the number of elderly and children in the household, the education of the household head, and the household's location as risk factors for low income or poor education or both, and the probability of being poor in both dimensions. The gender of the household head, belonging to a minority religion and being widowed or separated compared to being married, has shown less or the opposite influence than expected. We did not find evidence for strong tail dependencies between education and expenditures after conditioning on the covariates in the marginals and in the copula parameter. 

Focusing more on the household's location, we find that an example individual in a rural household exhibits lower dependencies compared to the same individual in an urban household, potentially due to restricted employment opportunities for highly educated individuals in rural areas. There is a strong spatial heterogeneity regarding poverty risk and the strength of the dependence. High dependencies can be found in the northwest of Indonesia while central provinces have lower ones. For Jakarta and Kalimantan Timur we found a discrepancy between being rich in terms of GDP and exhibiting high relative poverty risks, potentially indicating unequally distributed economic gains. 
 
Thanks to the flexibility of our approach, the analysis of the dependence between poverty dimensions, and the various results that can be derived consistently by only estimating one model, we advocate to include copula GAMLSS into the tool box of poverty researchers. Other applications in the context of poverty may include analysing the drivers and spatial patterns of inter-generational poverty persistence or upward social mobility. On the methodological side, future research can be directed at combining copula GAMLSS with experimental or quasi-experimental methods to evaluate Indonesia's poverty policies on the micro level. {Extensions beyond the bivariate case are dependent on the availability of a suitable copula. The implementation of these models subsequently becomes more numerically and technically demanding and interpretation of the regression results will be challenging.} 

%%%%%%%%%%%%%%%%%%%%%%%%%%%%%%%%%%%%%%%%%%%%%%%%%%%%%%%%%%%%%%%%%%%%%%
%%%%%%%%%%%%%%%%%%%%%%%%%%%%%%%%%%%%%%%%%%%%%%%%%%%%%%%%%%%%%%%%%%%%%

\newpage
\bibliographystyle{kluwer}
\interlinepenalty=10000
\bibliography{copula}

%%%%%%%%%%%%%%%%%%%%%%%%%%%%%%%%%%%%%%%%%%%%%%%%%%%%%%%%%%%%%%%%%%%%%%
%%%%%%%%%%%%%%%%%%%%%%%%%%%%%%%%%%%%%%%%%%%%%%%%%%%%%%%%%%%%%%%%%%%%%
\newpage

%\appendix
\begin{appendices}

%%%%%%%%%%%%%%%%%%%%%%%%%%%%%%%%%%%%%%%%%%%%%%%%%%%%%%%%%%%%%%%%%%%%%%
%%%%%%%%%%%%%%%%%%%%%%%%%%%%%%%%%%%%%%%%%%%%%%%%%%%%%%%%%%%%%%%%%%%%% 

\section{Considered copulas and corresponding (transformed) parameter
}\label{sec:copulas}
%\begin{landscape}
\begin{table}[h] 
\fontfamily{ptm}\selectfont
\centering
\caption{Families of some bivariate copula functions.}\label{tab:mixed.copulae}
%\begin{small}
\begin{tabular}{p{1cm}p{6.5cm}p{2cm}@{\hskip -3mm}p{2cm}p{2.2cm}}
\hline
Name & $\mathcal C(u,v)$ & range of $\gamma$ & $\gamma^*$ & Kendall's $\tau$\\
\hline
Gaussian & $\Phi_2(\Phi^{-1}(u),\Phi^{-1}(v))$ & $[-1,1]$ & $\textup{tanh}^{-1}(\gamma)$ & $\frac{2}{\pi} \arcsin{\gamma}$\\
Clayton & $(u^{-\gamma}+v^{-\gamma}-1)^{-1/\gamma}$ & $(0,\infty)$ & $\log(\gamma-\varepsilon)$ & $\frac{\gamma}{\gamma+2}$\\
Frank & $-\gamma^{-1}\log[1+(e^{-\gamma u}-1)(e^{-\gamma v}-1)/(e^{-\gamma}-1)]$ & $\mathds R\setminus\{0\}$ & $\gamma-\varepsilon$ & $1- \frac{4}{\gamma}(1-D_1(\gamma))$\\
Gumbel & $\exp\left\{-[(-\log u)^\gamma +(-\log v)^\gamma]^{1/\gamma}\right\}$ & $[1,\infty)$ & $\log(\gamma-1)$ & $1-\frac{1}{\gamma}$\\
Joe & $1-[(1-u)^{\gamma}+(1-v)^{\gamma}-(1-u)^{\gamma}(1-v)^{\gamma}]^{1/\gamma}$ & $(1,\infty)$ & $\log(\gamma-1-\varepsilon)$ & $1+ \frac{4}{\gamma^2} D_2(\gamma)$\\
FGM & $uv[1+\gamma(1-u)(1-v)]$ & $[-1,1]$ & $\textup{tanh}^{-1}(\gamma)$ & $\frac{2}{9}\gamma$ \vspace{1pt}\\
AMH & $\frac{uv}{1-\gamma(1-u)(1-v)}$ & $[-1,1)$ & $\textup{tanh}^{-1}(\gamma)$\vspace{4pt} & $1-\frac{2}{3\gamma^2}(\gamma + (1-\gamma)^2 \log(1-\gamma))$\\
Plackett & $\frac{1+(\gamma-1)(u+v)- \sqrt{[1+(\gamma-1)(u+v)]^2 - 4\gamma (\gamma - 1)uv}}{2(\gamma-1)}$ & $(0,\infty)$ & $\log(\gamma-\varepsilon)$ & -\\
\hline \vspace{0.5pt}\\
\multicolumn{5}{p{0.98\linewidth}}{\footnotesize{\textit{Note:} The association parameter is denoted by $\gamma$ and $u$ and $v$ are the marginals $F_1(\eta_{1k})$ and $F_2(y_2)$, respectively. The abbreviations FGM and AMH correspond to the Farlie-Gumbel-Morgenstern and the Ali-Mikhail-Haq copula, respectively. Function $\Phi_2(\dot,\dot;\gamma)$ is the CDF of a bivariate standard normal distribution, while $D_1(\gamma) = \frac{1}{\gamma} \int_0^{\gamma}\frac{t}{\exp(t)-1}dt $ is the Debye function and $D_2(\gamma) = \int_0^1 t \log(t)(1-t)^{\frac{2(1-\gamma)}{\gamma}}dt$ . The quantity $\varepsilon$ denotes the machine smallest floating point multiplied by $10^6$ and is introduced to force the transformed association parameters to lie in their respective supports throughout estimation.}}
\end{tabular}
%\end{small}
\end{table}
%\end{landscape}

\section{Further details on the trust algorithm and smoothing parameter selection}\label{sec:apx_trust}

\subsection{General idea of the trust region algorithm}

The trust region algorithm has the advantage of being more stable and faster compared to in-line search methods \citep{Marra.2017}. Both line search and trust region methods use a quadratic model of the objective function and generate steps from one iterate to the next. While line search methods use the model to find first a search direction and suitable step lengths along this direction, trust region algorithms search the step that  minimizes the objective function within a previously defined region around the current iterate such that both direction and step length are selected at the same time. If a function exhibits long plateaus and the current iterate $\boldsymbol \beta^{[a]}$ is in that region, line search methods may search the next step $\boldsymbol \beta^{[a+1]}$ far away from the current iterate. In this case, it is possible that the evaluation of the log likelihood will not be finite causing the algorithm to fail. Before evaluating the objective function, trust region methods define a maximum distance first based on the trust region. This has two advantages: first, the new iterate will not lie too far away from the current one; second, in case of a non-definite evaluation of $\breve{\ell}_p $, step $\boldsymbol p^{[a+1]}$  will be rejected. If a candidate that minimizes the quadratic approximation of the objective function and also lies in the trust region does not improve the function sufficiently or gives a non definite evaluation of $\breve{\ell}_p $, the trust region shrinks and the algorithm moves back to step 1. On the other hand, if the improvement is large enough, the trust region expands in the next iteration. Details on the trust region algorithm can be found in \citet{Nocedal2006}, Chapter 4. 

\subsection{Details on smoothing parameter selection}

Simultaneous optimization of $\boldsymbol \beta$ and $\boldsymbol \lambda$ causes overfitting. The reason is that the penalized log-likelihood will be highest at $\boldsymbol \lambda = \boldsymbol{0}$. Instead, the smoothing parameters $\boldsymbol \lambda$ should be selected in a way such that the function estimates are close to the true functions. 

%There are several approaches to automatic multiple smoothing parameter estimations for bivariate responses within the penalized likelihood framework in the literature. For example, \cite{Radice2016,marra.wyszynski.2016} marra and Radice 2013 proposes a regression spline approach that uses  
%
%\begin{align*}
%\boldsymbol M = \sqrt(\boldsymbol W) (\boldsymbol W^{-1} %\boldsymbol d + \boldsymbol{Z \delta}),
%\end{align*}
%
%Where is .....
%However, as \cite{marra.baernighausen.2017} pointed out, this %approach might be unstable when the copula parameter depends on %covariates and the additive predictors are highly flexible. 
%The reason is that M requires W to be positive definiite but if %copula parameter dependent, not $\boldsymbol W = \textrm{diag} %(\boldsymbol W_1, \ldots, \boldsymbol W_n)$ all need to be positive %definite. To tackle this requirement,  \citet{marra.baernighausen.2017} derive the result in  (\ref{eq:step2}). Still wanting to use observed information instead of E(W) They use a different parameterization of M that uses H ans g as a whole, not the component, reducing frequency of H not positive definite

Providing a full justification, \citet{marra.baernighausen.2017} propose an approach that yields the result in equation (\ref{eq:step2}). They use the fact that, close to convergence, the trust region algorithm behaves like a classic unconstrained algorithm \citep{Radice2016}. Considering the first order Taylor expansion of $\boldsymbol g_p^{[a+1]}$ and setting this to zero, they obtain
\begin{align} \label{eq:approx}
\boldsymbol{0} = \boldsymbol g_p^{[a+1]} \approx \boldsymbol g_p^{[a]} +\boldsymbol H_p^{[a]}(\boldsymbol \beta ^{[a+1]} - \boldsymbol \beta^{[a]}),
\end{align}
where the quantities as those defined in the paper. The aim is now to find an expression for $\boldsymbol \beta ^{[a+1]}$ that uses $\boldsymbol g^{[a]}$ and $\boldsymbol H ^{[a]}$ as a whole. \citet{marra.baernighausen.2017} argue that this reduces the frequency of situations where $\boldsymbol H$ is not positive definite compared to when working with the $n$ components that make it up. The situations where $\boldsymbol H$ turns out not to be positive definite can then be dealt with by perturbing $\boldsymbol H$ to positive definiteness \citep{wood2015core}. Recalling that $\boldsymbol{\cal I}^{[a]} = - \boldsymbol H^{[a]}$, the expression in (\ref{eq:approx}) can be re-written as  
\begin{align*}
\boldsymbol 0 &= \boldsymbol g_p^{[a]} + (- \boldsymbol{\cal {I}}^{[a]} - \boldsymbol S_{\boldsymbol\lambda} )  (\boldsymbol\beta^{[a+1]} - \boldsymbol \beta ^{[a]}) \\
\boldsymbol g_p^{[a]} &=  (\boldsymbol{\cal{I}}^{[a]} + \boldsymbol S_{\boldsymbol\lambda}) ( \boldsymbol \beta ^{[a+1]} - \boldsymbol \beta ^{[a]}) \\
\boldsymbol g^{[a]} - \boldsymbol S_{\boldsymbol\lambda} \boldsymbol \beta ^{[a]} &= 
(\boldsymbol{\cal{I}}^{[a]} + \boldsymbol S_{\boldsymbol\lambda})( \boldsymbol \beta ^{[a+1]})-(\boldsymbol{\cal{I}}^{[a]} + \boldsymbol S_{\boldsymbol\lambda})(\boldsymbol \beta ^{[a]})\\
%\boldsymbol \beta^{[a+1]} (\boldsymbol{\cal{I}}^{[a]} + \boldsymbol S_{\boldsymbol\lambda}) - \boldsymbol \beta ^{[a]} \boldsymbol{\cal{I}}^{[a]} - \boldsymbol \beta ^{[a]} \boldsymbol S_{\boldsymbol\lambda}\\
%
(\boldsymbol{\cal{I}}^{[a]} + \boldsymbol S_{\boldsymbol\lambda})\boldsymbol \beta ^{[a+1]} 
&= \boldsymbol g^{[a]} - \boldsymbol S_{\boldsymbol\lambda}\boldsymbol \beta ^{[a]} + (\boldsymbol{\cal{I}}^{[a]} + \boldsymbol 
S_{\boldsymbol\lambda})\boldsymbol \beta ^{[a]}\\
&= \boldsymbol g^{[a]} + \boldsymbol{\cal{I}}^{[a]}\boldsymbol \beta ^{[a]} \\
\boldsymbol \beta ^{[a+1]} &= (\boldsymbol{\cal{I}}^{[a]} + \boldsymbol S_{\boldsymbol\lambda})^{-1} \sqrt{\boldsymbol{\cal{I}}^{[a]}} ( \sqrt[]{\boldsymbol{\cal{I}}^{[a]}} \boldsymbol \beta ^{[a]}+\sqrt[]{\boldsymbol{\cal{I}}^{[a]}}^{-1} \boldsymbol g^{[a]} ).
\end{align*}
Thus, at convergence the parameter estimator takes the following form
\begin{align*}
\boldsymbol \beta^{[a+1]} = (\boldsymbol{\cal I}^{[a]} + \boldsymbol S_{\boldsymbol\lambda}) ^{-1}
\sqrt{{\boldsymbol{\cal I}}^{[a]} }\boldsymbol M^{[a]},
\end{align*}
where $\boldsymbol M^{[a]} := \boldsymbol{\mu_M}^{[a]} + \boldsymbol\epsilon^{[a]}$, 
$\boldsymbol{\mu_M}^{[a]} := \sqrt{\boldsymbol{\cal I}^{[a]}} \boldsymbol \beta^{[a]} $
and $\boldsymbol\epsilon^{[a]} := \sqrt{ \boldsymbol{\cal I}^{[a]}}^{-1} \boldsymbol g^{[a]}$.

From likelihood theory we have that $\boldsymbol\epsilon\sim {\cal N}(\boldsymbol 0, \boldsymbol I)$ and $\boldsymbol M  \sim {\cal N}(\boldsymbol \mu_{\boldsymbol M}, \boldsymbol I)$, with $\boldsymbol I$ being the identity matrix and $\boldsymbol \mu_{\boldsymbol M} := \sqrt{\boldsymbol{\cal{I}}} {\boldsymbol \beta^0}$, where $\boldsymbol \beta^0$ denotes the true parameter vector. The predicted value vector for $\boldsymbol M$ is $\hat{\boldsymbol \mu}_{\boldsymbol M} = \sqrt[]{\boldsymbol{\cal I }} \hat{\boldsymbol \beta} = \boldsymbol{AM}$, where $\boldsymbol A = \sqrt{\boldsymbol {\cal I}} (\boldsymbol{\cal I + \boldsymbol S_{\boldsymbol\lambda}}) ^{-1} \sqrt{\boldsymbol {\cal I}}$. To obtain an estimate of $\boldsymbol \lambda$ that suppresses the complexity of smooth terms not supported by the observed data, $\hat{\boldsymbol \mu}_{\boldsymbol M}$ should be close to $\boldsymbol\mu_{\boldsymbol M}$. Hence, employing the expected squared error, yields
\begin{align}
\mathbb{E}(\| \boldsymbol{\mu_M} - \hat{\boldsymbol \mu}_ {\boldsymbol M} \|^2) 
&= \mathbb{E}(\| (\boldsymbol M - \boldsymbol \epsilon) - \boldsymbol{AM} \|^2) \nonumber\\
&= \mathbb{E}(\| \boldsymbol M - \boldsymbol{AM} - \boldsymbol \epsilon \|^2) \nonumber\\
&= \mathbb{E}(\| \boldsymbol M - \boldsymbol{AM} \|^2) + \mathbb{E}(-\boldsymbol \epsilon^{\prime} \boldsymbol\epsilon - 2 \boldsymbol \epsilon ^{\prime} \boldsymbol {\mu_M} + 2 \boldsymbol \epsilon ^{\prime} \boldsymbol {A \mu_M} + 2 \boldsymbol \epsilon ^{\prime} \boldsymbol{A \epsilon}) \nonumber\\ 
&= \mathbb{E}(\| \boldsymbol M - \boldsymbol{AM} \|^2) - Kn + 2\textrm{tr}(\boldsymbol A). \label{eq:expectation}
\end{align}
The smoothing parameter is found by minimizing an estimate of the expectation in equation (\ref{eq:expectation}). For a given $\boldsymbol \beta ^{[a+1]}$ we arrive then at equation (\ref{eq:step2}).

\section{Gradient and Hessian} \label{sec:apx_gradient}
Throughout the derivations of the gradient vector and the Hessian matrix, we denote by $\boldsymbol\beta_1$ the regression coefficient associated to the fist ordinal equation, and by $\boldsymbol\beta_2:=(\boldsymbol\beta^{\mu_2}, \boldsymbol\beta^{\sigma_2})^\prime$ the coefficients for the location and scale parameters of the second continuous equation. Similarly $\boldsymbol\beta^{\gamma}$ refers to the coefficients of the copula association parameter.

\subsection{Gradient vector}
A general expression for the contribution of the $i$-th individual to the gradient vector is the following
\begin{align*}
\ell_{\boldsymbol\beta i}'(\boldsymbol\beta):=\frac{\partial\ell_i(\boldsymbol\beta)}{\partial\boldsymbol\beta}=\sum_{r\in\mathcal R}\mathds 1_{\{y_{1i}=r\}}\frac{1}{\nabla_r F_{12.2}(\eta_{1ri})}\frac{\partial\nabla_r F_{12.2}(\eta_{1ri})}{\partial\boldsymbol\beta}+f_2(y_{2i})^{-1}\frac{\partial f_2( y_{2i})}{\partial\boldsymbol\beta},\nonumber
\end{align*}
where $\nabla_r$ is the backward difference operator applied to $r$, in particular $\nabla_r F_{12.2}(\eta_{1ri}):=F_{12.2}(\eta_{1ri})-F_{12.2}(\eta_{1r-1i})$. We next define the quantities
\begin{align*}
F_{12.12}(\eta_{1ri}):=\frac{\partial^2\mathcal C(F_1(\eta_{1ri}),F_2(y_{2i}))}{\partial F_1(\eta_{1ri})\partial F_2(y_{2i})} \qquad \textrm{with} \qquad F_{12.22}(\eta_{1ri}):=\frac{\partial^2\mathcal C(F_1(\eta_{1ri}),F_2(y_{2i}))}{\partial F_2(y_{2i})^2}.\nonumber
\end{align*}

\subsubsection{Derivatives with respect to the transformed cut points}
\begin{itemize}
\item $h=1$:
\begin{align*}
\ell_{\theta_1^*i}' 
&=\frac{1}{\nabla_r F_{12.2}(\eta_{1ri})} \frac{\partial\theta_1}{\partial\theta_1^*}\left(F_{12.12}(\eta_{1ri})\frac{\partial F_1(\eta_{1ri})}{\partial\eta_{1ri}}\frac{\partial\eta_{1ri}}{\partial\theta_1}-F_{12.12}(\eta_{1r-1i})\frac{\partial F_1(\eta_{1r-1i})}{\partial\eta_{1r-1i}}\frac{\partial\eta_{1r-1i}}{\partial\theta_1}\right)\nonumber\\
&=\frac{1}{\nabla_r F_{12.2}(\eta_{1ri})}\left(F_{12.12}(\eta_{1ri})f_1(\eta_{1ri})-\mathds 1_{\{2\preceq r\}}F_{12.12}(\eta_{1r-1i})f_1(\eta_{1r-1i})\right)\nonumber
\end{align*}
\item $h=2,\ldots,R$:
\begin{align*}
\ell_{\theta_h^*i}'
&=\frac{1}{\nabla_r F_{12.2}(\eta_{1ri})}\frac{\partial\theta_h}{\partial\theta_h^*}\left(F_{12.12}(\eta_{1ri})
\frac{\partial F_1(\eta_{1ri})}{\partial\eta_{1ri}}\frac{\partial\eta_{1ri}}{\partial\theta_h}-F_{12.12}(\eta_{1r-1i})\frac{\partial F_1(\eta_{1r-1i})}{\partial\eta_{1r-1i}}\frac{\partial\eta_{1r-1i}}{\partial\theta_h}\right)\\ 
&=  \frac{2}{\nabla_r F_{12.2}(\eta_{1ri})} \left(\mathds 1_{\{h\preceq r\}}F_{12.12}(\eta_{1ri})f_1(\eta_{1ri})-\mathds 1_{\{h+1\preceq r\}}F_{12.12}(\eta_{1r-1i})f_1(\eta_{1r-1i})\right)\theta_h^*\nonumber
\end{align*}
\end{itemize}

\subsubsection{Derivatives with respect to \texorpdfstring{$\boldsymbol\beta_1$}{b1} and \texorpdfstring{$\boldsymbol\beta_2$}{b2}}
\begin{itemize}
\item $\boldsymbol\beta_1$:
\begin{align*}
\ell_{\boldsymbol\beta_1i}'
%&= \frac{1}{\nabla_r F_{12.2}(\eta_{1ri})} 
%\frac{\partial \nabla_r F_{12.2}(\eta_{1ri})}{ \partial \boldsymbol\beta_1}
%\nonumber\\ 
&= 
\left[\frac{1}{\nabla_r F_{12.2}(\eta_{1ri})}
\nabla_r\left(
\frac{\partial F_1(\eta_{1ri})}{\partial\eta_{1ri}}
\frac{\partial F_{12.2}(\eta_{1ri})}{\partial F_1(\eta_{1ri})}\right)\right]\frac{\partial\eta_{1ri}}{\partial\boldsymbol\beta_1}\\
\\
\quad  \textrm{with} &\quad \frac{\partial \eta_{1ri}}{\partial \boldsymbol \beta_1} = 
\frac{\partial \eta_{1r-1i}}{\partial \boldsymbol \beta_1}
%&=& \frac{1}{\nabla_k F_{1|2}(\eta_{1,k,i}|\textup y_{2,i})}\frac{\partial\eta_{1,k,i}}{\partial\boldsymbol\beta_1}\left(F_{1|12}(\eta_{1,k,i}|\textup y_{2,i})\frac{\partial F_1(\eta_{1,k,i})}{\partial\eta_{1,k,i}}-F_{1|12}(\eta_{1,k-1,i}|\textup y_{2,i})\frac{\partial F_1(\eta_{1,k-1,i})}{\partial\eta_{1,k-1,i}}\right)\nonumber\\
%&=&-\frac{1}{\nabla_k F_{1|2}(\eta_{1,k,i}|\textup y_{2,i})}\left(F_{1|12}(\eta_{1,k,i}|\textup y_{2,i})f_{1,i}(\eta_{1,k,i})-F_{1|12}(\eta_{1,k-1,i}|\textup y_{2,i})f_{1,i}(\eta_{1,k-1,i})\right)\mathbf x_{1,i}\nonumber
\end{align*}
\item $\boldsymbol\beta_2$; $\vartheta_k = \{\mu_2,\sigma_2\}$:
\begin{align*}
\ell_{\boldsymbol\beta_2i}'
\, %&=
%\frac{1}{\nabla_r F_{12.2}(\eta_{1ri})}
%\frac{\partial \nabla_r F_{12.2}(\eta_{1ri})} {\partial \boldsymbol \beta _2}\nonumber
%+ f_2(y_{2i})^{-1}
%\frac{\partial f_2(y_{2i})}{\partial \boldsymbol\beta_2}\\
&=
\left[
\frac{1}{\nabla_r F_{12.2}(\eta_{1ri})}
\nabla_r\left(
\frac{\partial F_{12.2}(\eta_{1ri})}{\partial F_2(y_{2i})}
 \right) \frac{\partial F_2(y_{2i})}{\partial \vartheta_k}+ f_2(y_{2i})^{-1}
\left.
\frac{\partial f_2(y_{2i})}{\partial \vartheta_k}
\right]\frac{\partial \vartheta_k}{\partial \eta^{\vartheta_k}_i}
 \frac{\partial\eta^{\vartheta_k}_i}{\partial\boldsymbol\beta^{\vartheta_k}}\right. \notag %\\
% &\quad + f_2(y_{2i})^{-1}
%\left.
%\frac{\partial f_2(y_{2i})}{\partial \vartheta_k}
%\right]\frac{\partial \vartheta_k}{\partial \eta^{\vartheta_k}_i}
% \frac{\partial\eta^{\vartheta_k}_i}{\partial\boldsymbol\beta^{\vartheta_k}} 
\end{align*}
\end{itemize}

\subsubsection{Derivatives with respect to the copula association parameter}
\begin{align*}
\ell_{\boldsymbol\beta^{\gamma}i}'
&=\left[
\frac{1}{\nabla_r F_{12.2}(\eta_{1ri})}
\nabla_r\left(
\frac{F_{12.2}(\eta_{1ri})}{\partial \gamma}
\right)\right]\frac{\partial\gamma}{\partial\eta^{\gamma}_i}\frac{\partial\eta^{\gamma}_i}{\partial\boldsymbol\beta^{\gamma}}
\nonumber
\end{align*}

\subsection{Hessian matrix}
The general expression for the contribution of the $i$-th individual to the Hessian matrix is given by

\begin{align*}
\ell_{\boldsymbol\beta\boldsymbol\beta^\prime}''(\boldsymbol\beta)&:=\frac{\partial}{\partial\boldsymbol\beta}\left(\frac{\partial\ell_i(\boldsymbol\beta)}{\partial\boldsymbol\beta}\right)^\prime\\
&=\sum_{r\in\mathcal R}\mathds1_{\{y_{1i}=r\}} 
\frac{1}{\nabla_r F_{12.2}(\eta_{1ri})}\\
&\left[
\frac{\partial^2\nabla_r F_{12.2}(\eta_{1ri})}{\partial \boldsymbol\beta\partial\boldsymbol\beta^\prime}-\frac{1}{\nabla_r F_{12.2}(\eta_{1ri})}
\left(\frac{\partial\nabla_r F_{12.2}(\eta_{1ri})}{\partial\boldsymbol\beta}\right)\left(\frac{\partial\nabla_r F_{12.2}(\eta_{1ri})}{\partial\boldsymbol\beta}\right)^\prime\right] \\
&+f_2(y_{2i})^{-1} \frac{\partial^2 f_2(y_{2i})}{\partial\boldsymbol\beta\partial\boldsymbol\beta^\prime}-f_2( y_{2i})^{-2}\left(\frac{\partial f_2( y_{2i})}{\partial\boldsymbol\beta}\right)\left(\frac{\partial f_2( y_{2i})}{\partial\boldsymbol\beta}\right)^\prime.
\end{align*}

\subsubsection{Hessian components for the transformed cut points}
\begin{itemize}
\item $(\theta_1^*)^2$:
\begin{align*}
\ell_{(\theta_1^*)^2i}'' &= \left(\frac{\partial\theta_1}{\partial\theta_1^*}\right)^2\left\{\frac{1}{\nabla_r F_{12.2}(\eta_{1ri})}\nabla_r\left(\left(\frac{\partial\eta_{1ri}}{\partial\theta_1}\right)^2 \left[\frac{\partial^2 F_{12.2}(\eta_{1ri})}{\partial F_1(\eta_{1ri})^2}\left(\frac{\partial F_1(\eta_{1ri})}{\partial\eta_{1ri}}\right)^2 \right. \right. \right. + \\
&\left. \left. \left. \frac{\partial F_{12.2}(\eta_{1ri})}{\partial F_1(\eta_{1ri})}\frac{\partial^2 F_1(\eta_{1ri})}{\partial\eta_{1ri}^2}\right]\right)
%&\nabla_r\left(\left(\frac{\partial\eta_{1ri}}{\partial\theta_1}\right)^2 \left[\frac{\partial^2 F_{12.2}(\eta_{1ri})}{\partial F_1(\eta_{1ri})^2}\left(\frac{\partial F_1(\eta_{1ri})}{\partial\eta_{1ri}}\right)^2+\frac{\partial F_{12.2}(\eta_{1ri})}{\partial F_1(\eta_{1ri})}\frac{\partial^2 F_1(\eta_{1ri})}{\partial\eta_{1ri}^2}\right]\right)\\
%&-\left.\left(\frac{\partial\theta_1}{\partial\theta_1^*}\frac{\partial\eta_{1r-1i}}{\partial\theta_1}\right)^2 \left(\frac{\partial^2 F_{12.2}(\eta_{1r-1i})}{\partial F_1(\eta_{1r-1i})^2} \left(\frac{\partial F_1(\eta_{1r-1i})}{\partial\eta_{1r-1i}}\right)^2+   \frac{\partial F_{12.2}(\eta_{1r-1i})}{\partial F_1(\eta_{1r-1i})}\frac{\partial^2 F_1(\eta_{1r-1i})}{\partial\eta_{1r-1i}^2}\right)\right]\nonumber \\
%
%&\left(\frac{\partial^2 F_{12.2}(\eta_{1r-1i})}{\partial F_1(\eta_{1r-1i})^2}\left. \left(\frac{\partial F_1(\eta_{1r-1i})}{\partial\eta_{1r-1i}}\right)^2+   \frac{\partial F_{12.2}(\eta_{1r-1i})}{\partial F_1(\eta_{1r-1i})}\frac{\partial^2 F_1(\eta_{1r-1i})}{\partial\eta_{1r-1i}^2}\right)\right]\nonumber\\
\left(\frac{1}{\nabla_r F_{12.2}(\eta_{1ri})}\right)^2\left[\nabla_r\left(\frac{\partial\eta_{1ri}}{\partial\theta_1}\frac{\partial F_{12.2}(\eta_{1ri})}{\partial F_1(\eta_{1ri})}\frac{\partial F_1(\eta_{1ri})}{\partial\eta_{1ri}}\right)\right]^2\right\}
%\\
%&\left.\left[\nabla_r\left(\frac{\partial\eta_{1ri}}{\partial\theta_1}\frac{\partial F_{12.2}(\eta_{1ri})}{\partial F_1(\eta_{1ri})}\frac{\partial F_1(\eta_{1ri})}{\partial\eta_{1ri}}\right)\right]^2\right\}\nonumber
\end{align*}
\item $\theta_1^*(\boldsymbol\theta_h^*)^\prime$; $h = 2,\ldots,R$:
\begin{align*}
\ell_{\theta_1^*(\boldsymbol\theta_h^*)^\prime i}''& =\frac{\partial\theta_1}{\partial\theta_1^*}\left\{\frac{1}{\nabla_r F_{12.2}(\eta_{1ri})} \right. \\
& \left. \nabla_r\left(\frac{\partial\eta_{1ri}}{\partial\theta_1}\left[\frac{\partial^2 F_{12.2}(\eta_{1ri})}{\partial F_1(\eta_{1ri})^2}\left(\frac{\partial F_1(\eta_{1ri})}{\partial\eta_{1ri}}\right)^2+\frac{\partial F_{12.2}(\eta_{1ri})}{\partial F_1(\eta_{1ri})}\frac{\partial^2 F_1(\eta_{1ri})}{\partial\eta_{1ri}^2}\right]\left(\frac{\partial\eta_{1ri}}{\partial\boldsymbol\theta_h}\right)^\prime\right)\right. \\
%
%&\left[\left(\frac{\partial\eta_{1ri}}{\partial\theta_1}\right)\left(\frac{\partial^2 F_{12.2}(\eta_{1ri})}{\partial F_1(\eta_{1ri})^2}\left(\frac{\partial F_1(\eta_{1ri})}{\partial\eta_{1ri}}\right)^2+\frac{\partial F_{12.2}(\eta_{1ri})}{\partial F_1(\eta_{1ri})}\frac{\partial^2 F_1(\eta_{1ri})}{\partial\eta_{1ri}^2}\right)\right.\\
%
%&\left(\frac{\partial\eta_{1ri}}{\partial\boldsymbol\theta_h}\right)^\prime- \left(\frac{\partial\theta_1}{\partial\theta_1^*}\frac{\partial\eta_{1r-1i}}{\partial\theta_1}\right)\left(\frac{\partial^2 F_{12.2}(\eta_{1r-1i})}{\partial F_1(\eta_{1r-1i})^2}\left(\frac{\partial F_1(\eta_{1r-1i})}{\partial\eta_{1r-1i}}\right)^2 \right. \notag \\
%
%&+\frac{\partial F_{12.2}(\eta_{1r-1i})}{\partial F_1(\eta_{1r-1i})}\left. \frac{\partial^2 F_1(\eta_{1r-1i})}{\partial\eta_{1r-1i}^2}\right) \left.\left(\frac{\partial\boldsymbol\theta_h^\top}{\partial\boldsymbol\theta_h^*}\frac{\partial\eta_{1r-1i}}{\partial\boldsymbol\theta_h}\right)^\prime\right] \\
&\left.-\left(\frac{1}{\nabla_r F_{12.2}(\eta_{1ri})}\right)^2 \right. \\
& \left. \nabla_r\left(\frac{\partial\eta_{1ri}}{\partial\theta_1}\frac{\partial F_{12.2}(\eta_{1ri})}{\partial F_1(\eta_{1ri})}\frac{\partial F_1(\eta_{1ri})}{\partial\eta_{1ri}}\right)\nabla_r\left(\frac{\partial F_{12.2}(\eta_{1ri})}{\partial F_1(\eta_{1ri})}\frac{\partial F_1(\eta_{1ri})}{\partial\eta_{1ri}}\left(\frac{\partial\eta_{1ri}}{\partial\boldsymbol\theta_h}\right)^\prime\right) \right\}
\left(\frac{\partial\boldsymbol\theta_h^\prime}{\partial\boldsymbol\theta_h^*}\right)^\prime\nonumber
%&\nabla_r\left(\frac{\partial\eta_{1ri}}{\partial\theta_1}\frac{\partial F_{12.2}(\eta_{1ri})}{\partial F_1(\eta_{1ri})}\frac{\partial F_1(\eta_{1ri})}{\partial\eta_{1ri}}\right)\nabla_r\left(\frac{\partial F_{12.2}(\eta_{1ri})}{\partial F_1(\eta_{1ri})}\frac{\partial F_1(\eta_{1ri})}{\partial\eta_{1ri}}\left(\frac{\partial\eta_{1ri}}{\partial\boldsymbol\theta_h}\right)^\prime\right)\notag \\
%
%&\left(\frac{\partial\theta_1}{\partial\theta_1^*} \left. \frac{\partial\eta_{1r-1i}}{\partial\theta_1}\right)\frac{\partial F_{12.2}(\eta_{1r-1i})}{\partial F_1(\eta_{1r-1i})}\frac{\partial F_1(\eta_{1r-1i})}{\partial\eta_{1ri}}\right]\\
%&\left.\left[\nabla_r\left(\frac{\partial F_{12.2}(\eta_{1ri})}{\partial F_1(\eta_{1ri})}\frac{\partial F_1(\eta_{1ri})}{\partial\eta_{1ri}}\left(\frac{\partial\eta_{1ri}}{\partial\boldsymbol\theta_h}\right)^\prime\right)\right]\right\}\left(\frac{\partial\boldsymbol\theta_h^\prime}{\partial\boldsymbol\theta_h^*}\right)^\prime\nonumber
\end{align*}

\item $\theta_1^*\boldsymbol\beta_1^\prime$:
\begin{align*}
\ell_{\theta_1^*\boldsymbol\beta_1^\prime i}'' &=\frac{\partial\theta_1}{\partial\theta_1^*}\left\{\frac{1}{\nabla_r F_{12.2}(\eta_{1ri})} \right.\\
& \left. \nabla_r\left(\frac{\partial\eta_{1ri}}{\partial\theta_1}\left[\frac{\partial^2 F_{12.2}(\eta_{1ri})}{\partial F_1(\eta_{1ri})^2}\left(\frac{\partial F_1(\eta_{1ri})}{\partial\eta_{1ri}}\right)^2+\frac{\partial F_{12.2}(\eta_{1ri})}{\partial F_1(\eta_{1ri})}\frac{\partial^2 F_1(\eta_{1ri})}{\partial\eta_{1ri}^2}\right]\right)\right.\nonumber\\
%&\nabla_r\left(\frac{\partial\eta_{1ri}}{\partial\theta_1}\left[\frac{\partial^2 F_{12.2}(\eta_{1ri})}{\partial F_1(\eta_{1ri})^2}\left(\frac{\partial F_1(\eta_{1ri})}{\partial\eta_{1ri}}\right)^2+\frac{\partial F_{12.2}(\eta_{1ri})}{\partial F_1(\eta_{1ri})}\frac{\partial^2 F_1(\eta_{1ri})}{\partial\eta_{1ri}^2}\right]\right)\nonumber\\
%&-\left.\left(\frac{\partial\eta_{1r-1i}}{\partial\theta_1}\right) \left(\frac{\partial^2 F_{12.2}(\eta_{1r-1i})}{\partial F_1(\eta_{1r-1i})^2} \left(\frac{\partial F_1(\eta_{1r-1i})}{\partial\eta_{1r-1i}}\right)^2+\frac{\partial F_{12.2}(\eta_{1r-1i})}{\partial F_1(\eta_{1r-1i})}\frac{\partial^2 F_1(\eta_{1r-1i})}{\partial\eta_{1r-1i}^2}\right)\right]\\
%& \left(\frac{\partial^2 F_{12.2}(\eta_{1r-1i})}{\partial F_1(\eta_{1r-1i})^2} \left. \left(\frac{\partial F_1(\eta_{1r-1i})}{\partial\eta_{1r-1i}}\right)^2+\frac{\partial F_{12.2}(\eta_{1r-1i})}{\partial F_1(\eta_{1r-1i})}\frac{\partial^2 F_1(\eta_{1r-1i})}{\partial\eta_{1r-1i}^2}\right)\right]\\
&\left.-\left(\frac{1}{\nabla_r F_{12.2}(\eta_{1ri})}\right)^2 \right. \\
& \left.\nabla_r\left(\frac{\partial\eta_{1ri}}{\partial\theta_1}\frac{\partial F_{12.2}(\eta_{1ri})}{\partial F_1(\eta_{1ri})}\frac{\partial F_1(\eta_{1ri})}{\partial\eta_{1ri}}\right)\nabla_r\left(\frac{\partial F_{12.2}(\eta_{1ri})}{\partial F_1(\eta_{1ri})}\frac{\partial F_1(\eta_{1ri})}{\partial\eta_{1ri}}\right)\right\}\left(\frac{\partial\eta_{1ri}}{\partial\boldsymbol\beta_1}\right)^\prime\nonumber
%&\left.\nabla_r\left(\frac{\partial\eta_{1ri}}{\partial\theta_1}\frac{\partial F_{12.2}(\eta_{1ri})}{\partial F_1(\eta_{1ri})}\frac{\partial F_1(\eta_{1ri})}{\partial\eta_{1ri}}\right)\nabla_r\left(\frac{\partial F_{12.2}(\eta_{1ri})}{\partial F_1(\eta_{1ri})}\frac{\partial F_1(\eta_{1ri})}{\partial\eta_{1ri}}\right)\right\}\left(\frac{\partial\eta_{1ri}}{\partial\boldsymbol\beta_1}\right)^\prime\nonumber\\
%
%&\left(\frac{\partial\theta_1}{\partial\theta_1^*}\left.\frac{\partial\eta_{1r-1i}}{\partial\theta_1}\right)\frac{\partial F_{12.2}(\eta_{1r-1i})}{\partial F_1(\eta_{1r-1i})}\frac{\partial F_1(\eta_{1r-1i})}{\partial\eta_{1ri}}\right]\nonumber\\
%&\left.\nabla_r\left(\frac{\partial F_{12.2}(\eta_{1ri})}{\partial F_1(\eta_{1ri})}\frac{\partial F_1(\eta_{1ri})}{\partial\eta_{1ri}}\right)\right\}\left(\frac{\partial\eta_{1ri}}{\partial\boldsymbol\beta_1}\right)^\prime\nonumber
\end{align*}

\item $\theta_1^*\boldsymbol\beta_2^\prime$; $\vartheta_k = \{\mu_2,\sigma_2\}$:
\begin{align*}
\ell_{\theta_1^*\boldsymbol\beta_2^\prime i}'' &=\frac{\partial\theta_1}{\partial\theta_1^*}\left\{\frac{1}{\nabla_r F_{12.2}(\eta_{1ri})}\nabla_r\left(\frac{\partial\eta_{1ri}}{\partial\theta_1}\frac{\partial^2 F_{12.2}(\eta_{1ri})}{\partial F_1(\eta_{1ri})\partial F_2(y_{2,i})}\frac{\partial F_1(\eta_{1ri})}{\partial\eta_{1ri}}\right)\frac{\partial F_2(y_{2i})}{\partial\eta^{\vartheta_k}_i} \right.\\
&-\left.\left(\frac{1}{\nabla_r F_{12.2}(\eta_{1ri})}\right)^2\nabla_r\left(\frac{\partial\eta_{1ri}}{\partial\theta_1}\frac{\partial F_{12.2}(\eta_{1ri})}{\partial F_1(\eta_{1ri})}\frac{\partial F_1(\eta_{1ri})}{\partial\eta_{1ri}}\right)\frac{\partial\nabla_k F_{12.2}(\eta_{1ri})}{\partial F_2(y_{2i})}\frac{\partial F_2(y_{2i})}{\partial\eta^{\vartheta_k}_i}\right\} \\
&\left(\frac{\partial\eta^{\vartheta_k}_i}{\partial\boldsymbol\beta^{\vartheta_k}}\right)^\prime\nonumber
%&\left.\nabla_r\left(\frac{\partial\eta_{1ri}}{\partial\theta_1}\frac{\partial F_{12.2}(\eta_{1ri})}{\partial F_1(\eta_{1ri})}\frac{\partial F_1(\eta_{1ri})}{\partial\eta_{1ri}}\right)\left(\frac{\partial\nabla_k F_{12.2}(\eta_{1ri})}{\partial F_2(y_{2i})}\frac{\partial F_2(y_{2i})}{\partial\eta^{\vartheta_k}_i}\right)\right\}\left(\frac{\partial\eta^{\vartheta_k}_i}{\partial\boldsymbol\beta^{\vartheta_k}}\right)^\prime\nonumber\\
%
%&-\left(\frac{\partial\theta_1}{\partial\theta_1^*}\left. \frac{\partial\eta_{1r-1i}}{\partial\theta_1}\right)\frac{\partial F_{12.2}(\eta_{1r-1i})}{\partial F_1(\eta_{1r-1i})}\frac{\partial F_1(\eta_{1r-1i})}{\partial\eta_{1ri}}\right]\nonumber\\
%&\left.\left(\frac{\partial\nabla_k F_{12.2}(\eta_{1ri})}{\partial F_2(y_{2i})}\frac{\partial F_2(y_{2i})}{\partial\eta^{\vartheta_k}_i}\right)\right\}\left(\frac{\partial\eta^{\vartheta_k}_i}{\partial\boldsymbol\beta^{\vartheta_k}}\right)^\prime
\end{align*}

\item $\theta_1^*(\boldsymbol\beta^{\gamma})^\prime$:
\begin{align*}
\ell_{\theta_1^*(\boldsymbol\beta^{\gamma})^\prime i}''
&= \frac{\partial\theta_1}{\partial\theta_1^*}\left\{\frac{1}{\nabla_r F_{12.2}(\eta_{1ri})}\nabla_r\left(\frac{\partial\eta_{1ri}}{\partial\theta_1}\frac{\partial^2 F_{12.2}(\eta_{1ri})}{\partial F_1(\eta_{1ri})\partial\gamma}\frac{\partial F_1(\eta_{1ri})}{\partial\eta_{1ri}}\right)\right.\nonumber\\
%&\nabla_r\left(\frac{\partial\eta_{1ri}}{\partial\theta_1}\frac{\partial^2 F_{12.2}(\eta_{1ri})}{\partial F_1(\eta_{1ri})\partial\gamma}\frac{\partial F_1(\eta_{1ri})}{\partial\eta_{1ri}}\right) \frac{\partial\gamma}{\partial\eta^{\gamma}_i} \\
%
%& \left(\frac{\partial\theta_1}{\partial\theta_1^*} \left. \frac{\partial\eta_{1r-1i}}{\partial\theta_1}\right)\frac{\partial^2 F_{12.2}(\eta_{1r-1i})}{\partial F_1(\eta_{1r-1i})\partial\gamma}\frac{\partial F_1(\eta_{1r-1i})}{\partial\eta_{1r-1i}}\right] \frac{\partial\gamma}{\partial\eta^{\gamma}_i} \\
&-\left.\left(\frac{1}{\nabla_r F_{12.2}(\eta_{1ri})}\right)^2\nabla_r\left(\frac{\partial\eta_{1ri}}{\partial\theta_1}\frac{\partial F_{12.2}(\eta_{1ri})}{\partial F_1(\eta_{1ri})}\frac{\partial F_1(\eta_{1ri})}{\partial\eta_{1ri}}\right)\frac{\partial\nabla_r F_{12.2}(\eta_{1ri})}{\partial\gamma}\right\}\frac{\partial\gamma}{\partial\eta^{\gamma}_i}\left(\frac{\partial\eta^{\gamma}_i}{\partial\boldsymbol\beta^{\gamma}}\right)^\prime\nonumber
%
%&\nabla_r\left(\frac{\partial\eta_{1ri}}{\partial\theta_1}\frac{\partial F_{12.2}(\eta_{1ri})}{\partial F_1(\eta_{1ri})}\frac{\partial F_1(\eta_{1ri})}{\partial\eta_{1ri}}\right)\\
%
%& \left(\frac{\partial\theta_1}{\partial\theta_1^*} \left. \frac{\partial\eta_{1r-1i}}{\partial\theta_1}\right)\frac{\partial F_{12.2}(\eta_{1r-1i})}{\partial F_1(\eta_{1r-1i})}\frac{\partial F_1(\eta_{1r-1i})}{\partial\eta_{1r-1i}}\right]\nonumber\\
%&\left.\frac{\partial\nabla_r F_{12.2}(\eta_{1ri})}{\partial\gamma}\right\}\frac{\partial\gamma}{\partial\eta^{\gamma}_i}\left(\frac{\partial\eta^{\gamma}_i}{\partial\boldsymbol\beta^{\gamma}}\right)^\prime\nonumber
\end{align*}

\item $\boldsymbol\theta_{\bar h}^*(\boldsymbol\theta_h^*)^\prime$; $\bar h, h = 2,\ldots,R$:
\begin{align*}
\ell_{\boldsymbol\theta_{\bar h}^*(\boldsymbol\theta_h^*)^\prime i}''
&=\frac{1}{\nabla_r F_{12.2}(\eta_{1ri})}\nonumber\\
&\left\{\frac{\partial\boldsymbol\theta_h^\prime}{\partial\boldsymbol\theta_h^*}\nabla_r\left(\left[\frac{\partial^2F_{12.2}(\eta_{1ri})}{\partial F_1(\eta_{1ri})^2}\left(\frac{\partial F_1(\eta_{1ri})}{\partial\eta_{1ri}}\right)^2+\frac{\partial F_{12.2}(\eta_{1ri})}{\partial F_1(\eta_{1ri})}\frac{\partial^2 F_1(\eta_{1ri})}{\partial\eta_{1ri}^2}\right]\frac{\partial\eta_{1ri}}{\partial\boldsymbol\theta_{\bar h}}\left(\frac{\partial\eta_{1ri}}{\partial\boldsymbol\theta_h}\right)^\prime\right) \right. \\ 
& \left.  \left(\frac{\partial\boldsymbol\theta_h^\prime}{\partial\boldsymbol\theta_h^*}\right)^\prime
%&\left(\frac{\partial\boldsymbol\theta_{\bar h}}{\partial\boldsymbol\theta_{\bar h}^*}\frac{\partial\eta_{1ri}}{\partial\boldsymbol\theta_{\bar h}}\right)\left(\frac{\partial\boldsymbol\theta_h}{\partial\boldsymbol\theta_h^*}\frac{\partial\eta_{1ri}}{\partial\boldsymbol\theta_h}\right)^\prime-\nonumber\\
%&-\left[\frac{\partial^2F_{12.2}(\eta_{1r-1i})}{\partial F_1(\eta_{1r-1i})^2}\left(\frac{\partial F_1(\eta_{1r-1i})}{\partial\eta_{1r-1i}}\right)^2+\frac{\partial F_{12.2}(\eta_{1r-1i})}{\partial F_1(\eta_{1r-1i})}\frac{\partial^2 F_1(\eta_{1r-1i})}{\partial\eta_{1r-1i}^2}\right]\left(\frac{\partial\boldsymbol\theta_{\bar h}}{\partial\boldsymbol\theta_{\bar h}^*}\frac{\partial\eta_{1r-1i}}{\partial\boldsymbol\theta_{\bar h}}\right)\left(\frac{\partial\boldsymbol\theta_h}{\partial\boldsymbol\theta_h^*}\frac{\partial\eta_{1r-1i}}{\partial\boldsymbol\theta_h}\right)^\prime\nonumber\\
%&\left(\frac{\partial\boldsymbol\theta_{\bar h}}{\partial\boldsymbol\theta_{\bar h}^*}\frac{\partial\eta_{1r-1i}}{\partial\boldsymbol\theta_{\bar h}}\right)\left(\frac{\partial\boldsymbol\theta_h}{\partial\boldsymbol\theta_h^*}\frac{\partial\eta_{1r-1i}}{\partial\boldsymbol\theta_h}\right)^\prime+\nonumber\\
+\nabla_r\left(\frac{\partial F_{12.2}(\eta_{1ri})}{\partial F_1(\eta_{1ri})}\frac{\partial F_1(\eta_{1ri})}{\partial\eta_{1ri}}\left(\frac{\partial\eta_{1ri}}{\partial\boldsymbol\theta_h}\right)^\prime\right)\frac{\partial^2\boldsymbol\theta_h}{\partial\boldsymbol\theta_{\bar h}^*\partial(\boldsymbol\theta_h^*)^\prime}\right\}\nonumber\\
%&\left.\left.\frac{\partial F_{12.2}(\eta_{1r-1i})}{\partial F_1(\eta_{1r-1i})}\frac{\partial F_1(\eta_{1r-1i})}{\partial\eta_{1r-1i}}\left(\frac{\partial\eta_{1r-1i}}{\partial\boldsymbol\theta_h}\right)^\prime\right]\frac{\partial^2\boldsymbol\theta_h}{\partial\boldsymbol\theta_{\bar h}^*\partial(\boldsymbol\theta_h^*)^\prime}\right\}- \\
%
& -\left(\frac{1}{\nabla_r F_{12.2}(\eta_{1ri})}\right)^2\nonumber\\
&\nabla_r\left(\frac{\partial F_{12.2}(\eta_{1ri})}{\partial F_1(\eta_{1ri})}\frac{\partial F_1(\eta_{1ri})}{\partial\eta_{1ri}}\left(\frac{\partial\eta_{1ri}}{\partial\boldsymbol\theta_h}\right)^\prime\right)\left(\frac{\partial\boldsymbol\theta_h^\prime}{\partial\boldsymbol\theta_h^*}\right)^\prime\frac{\partial\boldsymbol\theta_h^\prime}{\partial\boldsymbol\theta_h^*}\nabla_r\left(\frac{\partial F_{12.2}(\eta_{1ri})}{\partial F_1(\eta_{1ri})}\frac{\partial F_1(\eta_{1ri})}{\partial\eta_{1ri}}\frac{\partial\eta_{1ri}}{\partial\boldsymbol\theta_h}\right)\\
%& \frac{\partial F_{12.2}(\eta_{1r-1i})}{\partial F_1(\eta_{1r-1i})}\left. \frac{\partial F_1(\eta_{1r-1i})}{\partial\eta_{1ri}}\left(\frac{\partial\boldsymbol\theta_h}{\partial\boldsymbol\theta_h^*}\frac{\partial\eta_{1r-1i}}{\partial\boldsymbol\theta_h}\right)^\prime\right]\nonumber\\
%
%&\left[\frac{\partial F_{12.2}(\eta_{1ri})}{\partial F_1(\eta_{1ri})}\frac{\partial F_1(\eta_{1ri})}{\partial\eta_{1ri}}\left(\frac{\partial\boldsymbol\theta_h}{\partial\boldsymbol\theta_h^*}\frac{\partial\eta_{1ri}}{\partial\boldsymbol\theta_h}\right)^\prime-\frac{\partial F_{12.2}(\eta_{1r-1i})}{\partial F_1(\eta_{1r-1i})}\frac{\partial F_1(\eta_{1r-1i})}{\partial\eta_{1ri}}\left(\frac{\partial\boldsymbol\theta_h}{\partial\boldsymbol\theta_h^*}\frac{\partial\eta_{1r-1i}}{\partial\boldsymbol\theta_h}\right)^\prime\right]^\prime\nonumber\\
%&\left.\frac{\partial F_{1|2}(\eta_{1,k-1,i}|\textup y_{2,i})}{\partial F_1(\eta_{1,k-1,i})}\frac{\partial F_1(\eta_{1,k-1,i})}{\partial\eta_{1,k,i}}\left(\frac{\partial\boldsymbol\theta_h}{\partial\boldsymbol\theta_h^*}\frac{\partial\eta_{1,k-1,i}}{\partial\boldsymbol\theta_h}\right)^\top\right]^\top
\\
\quad  \textrm{where} &\quad \frac{\partial^2\boldsymbol\theta_h}{\partial\boldsymbol\theta_{\bar h}^*\partial(\boldsymbol\theta_h^*)^\prime} = \left[\frac{\partial^2\theta_h}{\partial\theta_{\bar h}^*\partial(\theta_h^*)^\prime}\right]_{\bar h = 2,\ldots, R; h = 1,\ldots,R} \\
\quad  \textrm{and} &\quad \frac{\partial^2\theta_h}{\partial\theta_{\bar h}^*\partial\theta_h^*}=\left\{\begin{array}{ll}2 & \textup{if } \bar h=h \\ 0 & \textup{o/w}\end{array}\right.
\end{align*}

\item $\boldsymbol\theta_h^*\boldsymbol\beta_1^\prime$; $h=2,\ldots,R$:
\begin{align*}
\ell_{\boldsymbol\theta_h^*
\boldsymbol\beta_1^\prime i}'' 
&=\frac{\partial\boldsymbol\theta_h^\prime}{\partial\boldsymbol\theta_h^*} \left\{\frac{1}{\nabla_r F_{12.2}(\eta_{1ri})} \right. \\
& \left.\nabla_r\left(\frac{\partial\eta_{1ri}}{\partial\boldsymbol\theta_h}\left[\frac{\partial^2F_{12.2}(\eta_{1ri})}{\partial F_1(\eta_{1ri})^2}\left(\frac{\partial F_1(\eta_{1ri})}{\partial\eta_{1ri}}\right)^2+\frac{\partial F_{12.2}(\eta_{1ri})}{\partial F_1(\eta_{1ri})}\frac{\partial^2 F_1(\eta_{1ri})}{\partial\eta_{1ri}^2}\right]\right)\right.\nonumber\\
%&\nabla_r\left(\frac{\partial\eta_{1ri}}{\partial\boldsymbol\theta_h}\left[\frac{\partial^2F_{12.2}(\eta_{1ri})}{\partial F_1(\eta_{1ri})^2}\left(\frac{\partial F_1(\eta_{1ri})}{\partial\eta_{1ri}}\right)^2+\frac{\partial F_{12.2}(\eta_{1ri})}{\partial F_1(\eta_{1ri})}\frac{\partial^2 F_1(\eta_{1ri})}{\partial\eta_{1ri}^2}\right]\right)\nonumber\\
%&-\left.\left(\frac{\partial\boldsymbol\theta_h}{\partial\boldsymbol\theta_h^*}\frac{\partial\eta_{1r-1i}}{\partial\boldsymbol\theta_h}\right) \left[\frac{\partial^2F_{12.2}(\eta_{1r-1i})}{\partial F_1(\eta_{1r-1i})^2} \left(\frac{\partial F_1(\eta_{1r-1i})}{\partial\eta_{1r-1i}}\right)^2+ \frac{\partial F_{12.2}(\eta_{1r-1i})}{\partial F_1(\eta_{1r-1i})}\frac{\partial^2 F_1(\eta_{1r-1i})}{\partial\eta_{1r-1i}^2}\right]\right\} \notag \\
%&\left[\frac{\partial^2F_{12.2}(\eta_{1r-1i})}{\partial F_1(\eta_{1r-1i})^2}\left. \left(\frac{\partial F_1(\eta_{1r-1i})}{\partial\eta_{1r-1i}}\right)^2+ \frac{\partial F_{12.2}(\eta_{1r-1i})}{\partial F_1(\eta_{1r-1i})}\frac{\partial^2 F_1(\eta_{1r-1i})}{\partial\eta_{1r-1i}^2}\right]\right\}-\nonumber\\
&-\left.\left(\frac{1}{\nabla_r F_{12.2}(\eta_{1ri})}\right)^2\nabla_r\left(\frac{\partial\eta_{1ri}}{\partial\boldsymbol\theta_h}\frac{\partial F_{12.2}(\eta_{1ri})}{\partial F_1(\eta_{1ri})}\frac{\partial F_1(\eta_{1ri})}{\partial\eta_{1ri}}\right)\nabla_r\left(\frac{\partial F_{12.2}(\eta_{1ri})}{\partial F_1(\eta_{1ri})}\frac{\partial F_1(\eta_{1ri})}{\partial\eta_{1ri}}\right)\right\} \\
& \left(\frac{\partial\eta_{1ri}}{\partial\boldsymbol\beta_1}\right)^\prime\nonumber%\\
%&\left.\nabla_r\left(\frac{\partial\eta_{1ri}}{\partial\boldsymbol\theta_h}\frac{\partial F_{12.2}(\eta_{1ri})}{\partial F_1(\eta_{1ri})}\frac{\partial F_1(\eta_{1ri})}{\partial\eta_{1ri}}\right)\nabla_r\left(\frac{\partial F_{12.2}(\eta_{1ri})}{\partial F_1(\eta_{1ri})}\frac{\partial F_1(\eta_{1ri})}{\partial\eta_{1ri}}\right)\right\}\left(\frac{\partial\eta_{1ri}}{\partial\boldsymbol\beta_1}\right)^\prime\nonumber\\
%
%&\left(\frac{\partial\boldsymbol\theta_h}{\partial\boldsymbol\theta_h^*}\left. \frac{\partial\eta_{1r-1i}}{\partial\boldsymbol\theta_h}\right)\frac{\partial F_{12.2}(\eta_{1r-1i})}{\partial F_1(\eta_{1r-1i})}\frac{\partial F_1(\eta_{1r-1i})}{\partial\eta_{1ri}}\right]\nonumber\\
%&\left.\nabla_r\left(\frac{\partial F_{12.2}(\eta_{1ri})}{\partial F_1(\eta_{1ri})}\frac{\partial F_1(\eta_{1ri})}{\partial\eta_{1ri}}\right)\right\}\left(\frac{\partial\eta_{1ri}}{\partial\boldsymbol\beta_1}\right)^\prime\nonumber
\end{align*}

\item $\boldsymbol\theta_h^*\boldsymbol\beta_2^\prime$; $h=2,\ldots,R$; $\vartheta_k = \{\mu_2,\sigma_2\}$:
\begin{align*}
\ell_{\boldsymbol\theta_h^*
\boldsymbol\beta_2^\prime i}'' 
&=\frac{\partial\boldsymbol\theta_h^\prime}{\partial\boldsymbol\theta_h^*} \left\{\frac{1}{\nabla_r F_{12.2}(\eta_{1ri})}\nabla_r\left(\frac{\partial\eta_{1ri}}{\partial\boldsymbol\theta_h}\frac{\partial^2F_{12.2}(\eta_{1ri})}{\partial F_1(\eta_{1ri})\partial F_2(y_{2i})}\frac{\partial F_1(\eta_{1ri})}{\partial\eta_{1ri}}\right)\right.\nonumber\\
&-\left.\left(\frac{1}{\nabla_r F_{12.2}(\eta_{1ri})}\right)^2\nabla_r\left(\frac{\partial\eta_{1ri}}{\partial\boldsymbol\theta_h}\frac{\partial F_{12.2}(\eta_{1ri})}{\partial F_1(\eta_{1ri})}\frac{\partial F_1(\eta_{1ri})}{\partial\eta_{1ri}}\right)\nabla_r\left(\frac{\partial F_{12.2}(\eta_{1ri})}{\partial F_2(y_{2i})}\right)\right\} \\
&\frac{\partial F_2(y_{2i})}{\partial\vartheta_k}\frac{\partial\vartheta_k}{\partial\eta^{\vartheta_k}_i}\left(\frac{\partial\eta^{\vartheta_k}_i}{\partial\boldsymbol\beta^{\vartheta_k}}\right)^\prime\nonumber%\\
\end{align*}

\item $\boldsymbol\theta_h^*(\boldsymbol\beta^{\gamma})^\prime$; $h=2,\ldots,R$:
\end{itemize}
\begin{align*}
\ell_{\boldsymbol\theta_h^*
\boldsymbol(\boldsymbol\beta^\gamma)^\prime i}'' 
&=\frac{\partial\boldsymbol\theta_h^\prime}{\partial\boldsymbol\theta_h^*} \left\{\frac{1}{\nabla_r F_{12.2}(\eta_{1ri})}\nabla_r\left(\frac{\partial\eta_{1ri}}{\partial\boldsymbol\theta_h}\frac{\partial^2F_{12.2}(\eta_{1ri})}{\partial F_1(\eta_{1ri})\partial\gamma}\frac{\partial F_1(\eta_{1ri})}{\partial\eta_{1ri}}\right)\right.\nonumber\\
&-\left.\left(\frac{1}{\nabla_r F_{12.2}(\eta_{1ri})}\right)^2\nabla_r\left(\frac{\partial\eta_{1ri}}{\partial\boldsymbol\theta_h}\frac{\partial F_{12.2}(\eta_{1ri})}{\partial F_1(\eta_{1ri})}\frac{\partial F_1(\eta_{1ri})}{\partial\eta_{1ri}}\right)\nabla_r\left(\frac{\partial F_{12.2}(\eta_{1ri})}{\partial\gamma}\right)\right\} \\
&\frac{\partial F_2(y_{2i})}{\partial\gamma}\frac{\partial\gamma}{\partial\eta^{\gamma}_i}\left(\frac{\partial\eta^{\gamma}_i}{\partial\boldsymbol\beta^{\gamma}}\right)^\prime\nonumber%\\
\end{align*}

\subsubsection{Hessian components for \texorpdfstring{$\boldsymbol\beta_1$}{b1}}
\begin{itemize}
\item $\boldsymbol\beta_1\boldsymbol\beta_1^\prime$:
\begin{align*}
\ell_{\boldsymbol\beta_1\boldsymbol\beta_1^\prime i}''
&=\frac{\partial\eta_{1ri}}{\partial\boldsymbol\beta_1}\left\{\frac{1}{\nabla_r F_{12.2}(\eta_{ri})}\nabla_r\left(\frac{\partial^2 F_{12.2}(\eta_{1ri})}{\partial F_1(\eta_{1ri})^2}\left(\frac{\partial F_1(\eta_{1ri})}{\partial\eta_{1ri}}\right)^2+\frac{\partial F_{12.2}(\eta_{1ri})}{\partial F_1(\eta_{1ri})}\frac{\partial^2 F_1(\eta_{1ri})}{\partial\eta_{1ri}^2}\right)\right. \\
%&\nabla_r\left(\frac{\partial^2 F_{12.2}(\eta_{1ri})}{\partial F_1(\eta_{1ri})^2}\left(\frac{\partial F_1(\eta_{1ri})}{\partial\eta_{1ri}}\right)^2+\frac{\partial F_{12.2}(\eta_{1ri})}{\partial F_1(\eta_{1ri})}\frac{\partial^2 F_1(\eta_{1ri})}{\partial\eta_{1ri}^2}\right) \\
%&\left.\left(\frac{\partial^2 F_{12.2}(\eta_{1r-1i})}{\partial F_1(\eta_{1r-1i})^2}\left(\frac{\partial F_1(\eta_{1r-1i})}{\partial\eta_{1r-1i}}\right)^2+\frac{\partial F_{12.2}(\eta_{1ri})}{\partial F_1(\eta_{1r-1i})}\frac{\partial^2 F_1(\eta_{1r-1i})}{\partial\eta_{1r-1i}^2}\right)\right]- \\
&-\left.\left(\frac{1}{\nabla_r F_{12.2}(\eta_{1ri})}\right)^2 \left[\nabla_r\left(\frac{\partial F_{12.2}(\eta_{1ri})}{\partial F_1(\eta_{1ri})} \left. \frac{\partial F_1(\eta_{1ri})}{\partial\eta_{1ri}}\right)\right]^2\right\}
\left(\frac{\partial\eta_{1ri}}{\partial\boldsymbol\beta_1}\right)^\prime\right. \notag
%&\left[\nabla_r\left(\frac{\partial F_{12.2}(\eta_{1ri})}{\partial F_1(\eta_{1ri})} \left. \frac{\partial F_1(\eta_{1ri})}{\partial\eta_{1ri}}\right)\right]^2\right\}\left(\frac{\partial\eta_{1ri}}{\partial\boldsymbol\beta_1}\right)^\prime
\end{align*}

\item $\boldsymbol\beta_1\boldsymbol\beta_2^\prime$; $h=2,\ldots,R$; $\vartheta_k = \{\mu_2,\sigma_2\}$:
\begin{align*}
\ell_{\boldsymbol\beta_1\boldsymbol\beta_2^\prime i}''
&=\frac{\partial\eta_{1ri}}{\partial\boldsymbol\beta_1}\left\{\frac{1}{\nabla_r F_{12.2}(\eta_{1ri})}\nabla_r\left(\frac{\partial^2 F_{12.2}(\eta_{1ri})}{\partial F_1(\eta_{1ri})\partial F_2(y_{2i})}\frac{\partial F_1(\eta_{1ri})}{\partial\eta_{1ri}}\right)\right.\nonumber\\
%&\nabla_r\left(\frac{\partial^2 F_{12.2}(\eta_{1ri})}{\partial F_1(\eta_{1ri})\partial F_2(y_{2i})}\frac{\partial F_1(\eta_{1ri})}{\partial\eta_{1ri}}\right)\nonumber\\
&\left.-\left(\frac{1}{\nabla_r F_{12.2}(\eta_{1ri})}\right)^2\nabla_r\left(\frac{\partial F_{12.2}(\eta_{1ri})}{\partial F_1(\eta_{1ri})}\frac{\partial F_1(\eta_{1ri})}{\partial\eta_{1ri}}\right)\nabla_r\left(\frac{\partial F_{12.2}(\eta_{1ri})}{\partial F_2(y_{2i})}\right)\right\}\frac{\partial F_2( y_{2i})}{\partial\eta^{\vartheta_k}_i}\left(\frac{\partial\eta^{\vartheta_k}_i}{\partial\boldsymbol\beta^{\vartheta_k}}\right)^\prime\nonumber
%&\left.\nabla_r\left(\frac{\partial F_{12.2}(\eta_{1ri})}{\partial F_2(y_{2i})}\right)\right\}\frac{\partial F_2( y_{2i})}{\partial\eta^{\vartheta_k}_i}\left(\frac{\partial\eta^{\vartheta_k}_i}{\partial\boldsymbol\beta^{\vartheta_k}}\right)^\prime\nonumber
\end{align*}

\item $\boldsymbol\beta_1(\boldsymbol\beta^{\gamma})^\prime$:
\begin{align*}
\ell_{\boldsymbol\beta_1\boldsymbol\beta_{\gamma}^\prime i}'' & =\frac{\partial\eta_{1ri}}{\partial\boldsymbol\beta_1}\left\{\frac{1}{\nabla_r F_{12.2}(\eta_{1ri})}\nabla_r\left(\frac{\partial^2 F_{12.2}(\eta_{1ri})}{\partial F_1(\eta_{1ri})\partial\gamma}\frac{\partial F_1(\eta_{1ri})}{\partial\eta_{1ri}}\right)\right.\\
%&\nabla_r\left(\frac{\partial^2 F_{12.2}(\eta_{1ri})}{\partial F_1(\eta_{1ri})\partial\gamma}\frac{\partial F_1(\eta_{1ri})}{\partial\eta_{1ri}}\right)\nonumber\\
&-\left.\left(\frac{1}{\nabla_r F_{12.2}(\eta_{1ri})}\right)^2\nabla_r\left(\frac{\partial F_{12.2}(\eta_{1ri})}{\partial F_1(\eta_{1ri})}\frac{\partial F_1(\eta_{1ri})}{\partial\eta_{1ri}}\right)\nabla_r\left(\frac{\partial F_{12.2}(\eta_{1ri})}{\partial\gamma}\right)\right\}\frac{\partial\gamma}{\partial\eta^{\gamma}_i}\left(\frac{\partial\eta^{\gamma}_i}{\partial\boldsymbol\beta^{\gamma}}\right)^\prime\nonumber
%&\left.\nabla_r\left(\frac{\partial F_{12.2}(\eta_{1ri})}{\partial\gamma}\right)\right\}\frac{\partial\gamma}{\partial\eta^{\gamma}_i}\left(\frac{\partial\eta^{\gamma}_i}{\partial\boldsymbol\beta^{\gamma}}\right)^\prime\nonumber
\end{align*}
\end{itemize}

\subsubsection{Hessian components for \texorpdfstring{$\boldsymbol\beta_2$}{b2}}
\begin{itemize}
\item $\boldsymbol\beta_2\boldsymbol\beta_2^\prime$; $\vartheta_{\bar k}, \vartheta_k = \{\mu_2,\sigma_2\}$:
\begin{align*}
\ell_{\boldsymbol\beta_2\boldsymbol\beta_2^\prime i}'' & = \frac{\partial\eta^{\vartheta_{\bar k}}_i}{\partial\boldsymbol\beta^{\vartheta_{\bar k}}}\left\{\frac{1}{\nabla_r F_{12.2}(\eta_{1ri})}\left(\frac{\partial^2\nabla_r F_{12.2}(\eta_{1ri})}{\partial F_2( y_{2i})^2}\frac{\partial F_2( y_{2i})}{\partial\eta^{\vartheta_{\bar k}}_i}\frac{\partial F_2( y_{2i})}{\partial\eta^{\vartheta_k}_i}+\frac{\partial\nabla_r F_{12.2}(\eta_{1ri})}{\partial F_2(y_{2i})}\frac{\partial^2 F_2( y_{2i})}{\partial\eta^{\vartheta_{\bar k}}_i\partial\eta^{\vartheta_k}_i}\right)\right.\\
%& \left(\frac{\partial^2\nabla_r F_{12.2}(\eta_{1ri})}{\partial F_2( y_{2i})^2}\frac{\partial F_2( y_{2i})}{\partial\eta^{\vartheta_{\bar k}}_i}\frac{\partial F_2( y_{2i})}{\partial\eta^{\vartheta_k}_i}+\frac{\partial\nabla_r F_{12.2}(\eta_{1ri})}{\partial F_2(y_{2i})}\frac{\partial^2 F_2( y_{2i})}{\partial\eta^{\vartheta_{\bar k}}_i\partial\eta^{\vartheta_k}_i}\right)\\
&-\left(\frac{1}{\nabla_2 F_{12.2}(\eta_{1ri}}\right)^2\left(\frac{\partial\nabla_r F_{12.2}(\eta_{1ri})}{\partial F_2(y_{2i})}\frac{\partial F_2( y_{2i})}{\partial\eta^{\vartheta_{\bar k}}_i}\right)\left(\frac{\partial\nabla_r F_{12.2}(\eta_{1ri})}{\partial F_2(y_{2i})}\frac{\partial F_2( y_{2i})}{\partial\eta^{\vartheta_k}_i}\right)\\
%& \left(\frac{\partial\nabla_k F_{1|2}(\eta_{1,k,i}|\textup y_{2,i})}{\partial F_2(\textup y_{2,i})}\frac{\partial F_2(\textup y_{2,i})}{\partial\eta_{\vartheta_1,i}}\right)\left(\frac{\partial\nabla_k F_{1|2}(\eta_{1,k,i}|\textup y_{2,i})}{\partial F_2(\textup y_{2,i})}\frac{\partial F_2(\textup y_{2,i})}{\partial\eta_{\vartheta_2,i}}\right)+\\
&+\left.f_2(y_{2i})^{-1}\frac{\partial^2 f_2( y_{2i})}{\partial\eta^{\vartheta_{\bar k}}_i\partial\eta^{\vartheta_k}_i}-f_2(y_{2i})^{-2}\frac{\partial f_2( y_{2i})}{\partial\eta^{\vartheta_{\bar k}}_i}\frac{\partial f_2( y_{2i})}{\partial\eta^{\vartheta_k}_i}\right\}\left(\frac{\partial\eta^{\vartheta_{\bar k}}_i}{\partial\boldsymbol\beta^{\vartheta_k}}\right)^\prime
\end{align*}
\item $\boldsymbol\beta_2(\boldsymbol\beta^{\gamma})^\prime$; $\vartheta_k = \{\mu_2,\sigma_2\}$:
\begin{align*}
\ell_{\boldsymbol\beta_2(\boldsymbol\beta^{\gamma})^\prime i}''
&=\frac{\partial\eta^{\vartheta_k}_i}{\partial\boldsymbol\beta^{\vartheta_k}}\left\{\frac{1}{\nabla_r F_{12.2}(\eta_{1ri})}\left(\frac{\partial^2\nabla_r F_{12.2}(\eta_{1ri})}{\partial F_2(y_{2,i})\partial\gamma}\frac{\partial F_2( y_{2i})}{\partial\eta^{\vartheta_k}_i}\frac{\partial\gamma}{\partial\eta^{\gamma}_i}\right)\right.\nonumber\\
& -\left.\left(\frac{1}{\nabla_r F_{12.2}(\eta_{1ri})}\right)^2 \right. \left(\frac{\partial\nabla_r F_{12.2}(\eta_{1ri})}{\partial F_2(y_{2i})}\left. \frac{\partial F_2(y_{2i})}{\partial\eta^{\vartheta_k}_i}\right)\left(\frac{\partial\nabla_r F_{12.2}(\eta_{1ri})}{\partial\gamma}\frac{\partial\gamma}{\partial\eta^{\gamma}_i}\right)\right\}\left(\frac{\partial\eta^{\gamma}_i}{\partial\boldsymbol\beta^{\gamma}}\right)^\prime\nonumber
%& \left(\frac{\partial\nabla_k F_{1|2}(\eta_{1,k,i}|\textup y_{2,i})}{\partial F_2(\textup y_{2,i})}\left. \frac{\partial F_2(\textup y_{2,i})}{\partial\eta_{\vartheta_1,i}}\right)\left(\frac{\partial\nabla_k F_{1|2}(\eta_{1,k,i}|\textup y_{2,i})}{\partial\gamma}\frac{\partial\gamma}{\partial\eta_{\gamma,i}}\right)\right\}\left(\frac{\partial\eta_{\gamma,i}}{\partial\boldsymbol\beta_{\gamma}}\right)^\top\nonumber
\end{align*}
\end{itemize}
\subsubsection{Hessian components for the copula association parameter}
\begin{itemize}
\item $\boldsymbol\beta^\gamma(\boldsymbol\beta^\gamma)^\prime$:
\begin{align*}
\ell_{\boldsymbol\beta^\gamma(\boldsymbol\beta^\gamma)^\prime i}''
&= \frac{\partial\eta^\gamma_i}{\partial\boldsymbol\beta^{\gamma}}\left\{\frac{1}{\nabla_r F_{12.2}(\eta_{1ri})}\left(\frac{\partial^2\nabla_r F_{12.2}(\eta_{1ri})}{\partial\gamma^2} \left(\frac{\partial\gamma}{\partial\eta^\gamma_i}\right)^2+\frac{\partial\nabla_r F_{12.2}(\eta_{1ri})}{\partial\gamma}\frac{\partial^2\gamma}{\partial(\eta^\gamma_i)^2}\right)\right. \notag \\
%&- \left[\frac{\partial^2\nabla_k F_{1|2}(\eta_{1,k,i}|\textup y_{2,i})}{\partial\gamma^2}\left. \left(\frac{\partial\gamma}{\partial\eta_{\gamma,i}}\right)^2+\frac{\partial\nabla_k F_{1|2}(\eta_{1,k,i}|\textup y_{2,i})}{\partial\gamma}\frac{\partial^2\gamma}{\partial\eta_{\gamma,i}^2}\right]-\right.\nonumber\\
&-\left.\left(\frac{1}{\nabla_r F_{12.2}(\eta_{1ri})}\right)^2\left(\frac{\partial\nabla_r F_{12.2}(\eta_{1ri})}{\partial\gamma}\frac{\partial\gamma}{\partial\eta^\gamma_i}\right)^2\right\}\left(\frac{\partial\eta^\gamma_i}{\partial\boldsymbol\beta^\gamma}\right)^{\prime}\nonumber
\end{align*}
\end{itemize}

%%%%%%%%%%%%%%%%%%%%%%%%%%%%%%%%%%%%%%%%%%%%%%%%%%%%%%%%%%%%%%%%%%%%%%
%%%%%%%%%%%%%%%%%%%%%%%%%%%%%%%%%%%%%%%%%%%%%%%%%%%%%%%%%%%%%%%%%%%%% 

\section{Remarks on asymptotic properties \label{sec:apx_asymp}}

Asymptotic results for the proposed estimator can be derived along the lines of \citet{Marra.2017} and \cite{Donat2017}, for instance. Consider the Maximum Penalized Likelihood estimator
\begin{align*}
    \hat{\boldsymbol \beta} = \arg \underset{\boldsymbol{\beta}}{\max}   \, \ell_p(\boldsymbol \beta),
\end{align*}
where the penalized log-likelihood $\ell_p$ is given in equation (\ref{lik_pen}) and the parameter vector $\boldsymbol{\beta}$ comprises coefficients for the transformed cut points, all distributional parameters, and the copula parameter, i.e. $\boldsymbol{\beta} = (\theta^*_1, \dots, \theta^*_R, \boldsymbol{\beta}_1^{\prime}, \boldsymbol{\beta}_2^{\prime}, \boldsymbol\beta^{\boldsymbol\gamma \prime})^{\prime}$. The situation under consideration has a fixed number of spline bases such that the unknown smooth functions may not be exactly represented as linear combinations of given basis functions. However, as for example noted in \citet{Kauermann.2005}, using a large number of basis functions the approximation bias plays a minor role compared to estimation variability.   

If $L^t$ is the likelihood of the true model, its Kullback-Leibler distance to likelihood $L(\boldsymbol \beta)$ is 
\begin{align*}
    KL(L^t || L(\boldsymbol \beta)) = E (\ell^t - \ell(\boldsymbol \beta)).
\end{align*}

Defining the minimizer of this distance as $\boldsymbol{\beta}^0 = (\theta^{*0}_1, \dots, \theta^{*0}_{c-1}, \boldsymbol{\beta}_1^{0\prime}, \boldsymbol{\beta}_2^{0\prime}, \boldsymbol \beta^{\gamma0\prime})^{\prime}$ leads to
\begin{align*}
    \boldsymbol{\beta}^0 = \arg \underset{\boldsymbol{\beta}}{\min}  KL(L^t || L(\boldsymbol \beta)). 
\end{align*}
Thus, $\boldsymbol{\beta}^0$ minimizes the unpenalized log-likelihood $\ell(\cdot)$, that is $\mathbb E(\boldsymbol g(\boldsymbol \beta^0)) = 0,$ with $\boldsymbol g$ being the gradient of $\ell(\cdot)$. The Hessian is denoted by $\boldsymbol{H}(\boldsymbol{\beta})$. Both gradient and Hessian have penalized versions:
\begin{align*}
    \boldsymbol g_p(\boldsymbol{\beta}) = \boldsymbol{g}(\boldsymbol{\beta}) - \boldsymbol{S}_{\boldsymbol\lambda} \boldsymbol{\beta}, \\
    \boldsymbol H_p(\boldsymbol{\beta}) = \boldsymbol{H}(\boldsymbol{\beta}) - \boldsymbol{S}_{\boldsymbol\lambda} . 
\end{align*}
We assume the following set of conditions related to gradient and Hessian: % which are also used by e.g. Kauermann (2005) in a similar way for survival models: 
\begin{itemize}
    \item [(A1)] $\boldsymbol g(\boldsymbol\beta^0) = O_P(n^{1/2})$
    \item[(A2)] $\mathbb E(\boldsymbol H(\boldsymbol \beta^0)) = O(n)$
    \item[(A3)] $\boldsymbol H(\boldsymbol \beta^0) - \mathbb E(\boldsymbol H(\boldsymbol \beta^0) ) =  O_P(n^{1/2}) $
    \item[(A4)] $\boldsymbol S_{\boldsymbol \lambda} = o(n^{1/2})$
\end{itemize}
The first three assumptions are standard conditions when showing consistency of the un-penalized MLE and (A1) and (A3) mean that $\frac{1}{n}\boldsymbol g(\boldsymbol \beta^0)$ and $\frac{1}{n}\boldsymbol H(\boldsymbol \beta^0)$ converge in probability to their expected values with rate $n^{1/2}$. \citet{Kauermann.2005} gives an alternative formulation of (A4): $\lambda^{\vartheta_k} = o(n^{1/2})$, which means that the penalty sub-matrices of $\boldsymbol S_{\boldsymbol\lambda}$ corresponding to $\vartheta_k$ are asymptotically bounded and the penalty term becomes less and less important for the fitting procedure as $n\rightarrow \infty$.  

Regarding consistency, one has to show that 
\begin{align}\label{eq:consistency}
 \hat{\boldsymbol\beta} - \boldsymbol\beta^0  =O_P(n^{-1/2}) \qquad \textrm{with} \, n\rightarrow \infty.   
\end{align}
\begin{prop}\label{prop:asym}

Let $\boldsymbol \beta^0$ be the "true" parameter vector as defined above, under conditions (A1) - (A4) the penalized ML estimator $\hat{\boldsymbol \beta}$ satisfies 
\begin{align*}
    \hat{\boldsymbol{\beta}} - \boldsymbol \beta^0 = (-\mathbb E(\boldsymbol H (\boldsymbol \beta^0)) + \boldsymbol S_{\boldsymbol \lambda})^{-1} (\boldsymbol g(\boldsymbol \beta^0) - \boldsymbol S_{\boldsymbol \lambda} \boldsymbol  \beta^0) (\boldsymbol I + o_P(1))
\end{align*}
implying convergence in probability at rate $n^{-1/2}$ and hence consistency of $\hat{\boldsymbol\beta}$.
\end{prop}

\begin{proof}
Using the Taylor expansion of $\boldsymbol g_p(\hat{\boldsymbol \beta})$ at point $\boldsymbol \beta^0$, yields:
\begin{align*}
    \boldsymbol g_p(\hat{\boldsymbol \beta}) 
   & = \boldsymbol g_p({\boldsymbol \beta^0})  
    + \boldsymbol H_p(\boldsymbol \beta^0) (\hat{\boldsymbol \beta} - \boldsymbol \beta^0) 
    + \dots \\
    \textrm{since} \quad \boldsymbol g_p(\hat{\boldsymbol \beta}) = 0 \Longleftrightarrow
    \hat{\boldsymbol \beta} - \boldsymbol \beta^0 &= - \boldsymbol H_p(\boldsymbol \beta^0)^{-1} \boldsymbol g_p(\boldsymbol \beta^0) \\
    &= - \boldsymbol H_p(\boldsymbol \beta^0)^{-1} (\boldsymbol g (\boldsymbol \beta^0) - \boldsymbol S_{\boldsymbol \lambda} \boldsymbol \beta^0)
\end{align*}
We decompose  $\boldsymbol H_p(\boldsymbol \beta^0)$ and obtain:
\begin{align*}
    \boldsymbol H_p(\boldsymbol \beta^0) 
    = \boldsymbol H (\boldsymbol \beta^0) - \mathbb E(\boldsymbol H(\boldsymbol \beta^0)) - (-\mathbb E(\boldsymbol H(\boldsymbol \beta^0)) + \boldsymbol S_{\boldsymbol \lambda}).
\end{align*}
Upon defining $\boldsymbol R= \boldsymbol H (\boldsymbol \beta^0) - \mathbb E(\boldsymbol H(\boldsymbol \beta^0))$ as a stochastic error term and the penalized Fisher information matrix $\boldsymbol F = -\mathbb E(\boldsymbol H(\boldsymbol \beta^0)) + \boldsymbol S_{\boldsymbol \lambda})$, we calculate the Taylor expansion of $f(\boldsymbol{R}) = \boldsymbol H_p^{-1}(\boldsymbol \beta^0)$ at $f(\boldsymbol{0})$. The auxiliary function $f(\cdot) = (\cdot - \boldsymbol F(\lambda))^{-1}$ takes a matrix as input. We obtain:  
\begin{align*}
     \boldsymbol H_p^{-1}(\boldsymbol \beta^0) &= - \boldsymbol F(\boldsymbol\lambda)^{-1} - \boldsymbol F(\boldsymbol \lambda)^{-1} \boldsymbol R (\boldsymbol F(\boldsymbol \lambda))^{-1})^{\prime} + \dots
     \\
     & = - \boldsymbol F(\boldsymbol \lambda)^{-1} (\boldsymbol I + \boldsymbol{RF}(\boldsymbol \lambda)^{-1} + \dots)\\
     & = - \boldsymbol F(\boldsymbol \lambda)^{-1}(\boldsymbol I + O_P(n^{-1/2})) \qquad \textrm{with (A2)-(A4)}\\
      \Longleftrightarrow \hat{\boldsymbol{\beta}} - \boldsymbol \beta^0 &= (-\mathbb E(\boldsymbol H (\boldsymbol \beta^0)) + \boldsymbol S_{\boldsymbol \lambda})^{-1} (\boldsymbol g(\boldsymbol \beta^0) - \boldsymbol S_{\boldsymbol \lambda} \boldsymbol  \beta^0) (\boldsymbol I + o_P(1)),
\end{align*}
which proves the stated proposition and hence consistency as in equation (\ref{eq:consistency}). 
\end{proof}

The argumentation above is in line with maximum likelihood theory and is also adopted by \citet{Kauermann.2005} and \citet{Kauermann.2009} to derive asymptotic results on penalized spline smoothing. With Proposition \ref{prop:asym}, we can also derive the bias and covariance matrix of $\hat{\boldsymbol \beta}$, i.e. 
\begin{align*}
    \textrm{bias:} \quad &\mathbb E(\hat{\boldsymbol \beta}) - \boldsymbol \beta^0 = - \boldsymbol F(\boldsymbol \lambda)^{-1}\boldsymbol S_{\boldsymbol \lambda} \boldsymbol\beta^0(\boldsymbol I + o_P(1)) \quad \textrm{with} \quad \mathbb E(\boldsymbol g(\boldsymbol \beta^0)) = 0 \\
    \textrm{covariance:} \quad & Cov(\hat{\boldsymbol \beta})  = - \boldsymbol F(\boldsymbol \lambda)^{-1}
    \mathbb E(\boldsymbol H{\boldsymbol\beta^0} F(\boldsymbol \lambda)^{-1}(\boldsymbol I + o_P(1)) \quad \textrm{with} \quad Cov(\boldsymbol g(\boldsymbol \beta^0)) = - \mathbb E(\boldsymbol H(\boldsymbol \beta^0)). 
\end{align*}
Taking these considerations together with (A2) and (A4), we can characterise the asymptotic order of bias and covariance matrix as $o(n^{-1/2})$ and $O(n^{-1})$, respectively.

To guarantee an asymptotically normal behaviour of the score, we need an additional assumption: 
\begin{itemize}
    \item[(A5)] $\forall \beta^s \in \boldsymbol \beta: {\partial^3 \ell(\boldsymbol \beta)}/{\partial \beta^{s3}}$ exists and is bounded in the neighbourhood of $\beta_0^s$, that is \newline
    $|{\partial^3 \ell(\boldsymbol \beta)}/{\partial \beta^{s3}}| \leq M(\nu)$,  with  $E(M(\nu)|\beta^{0s})<\infty$, for all $\nu \in \mathds{R}$. Furthermore, $0\leq \boldsymbol{I}(\beta^{0s}) < \infty$. 
\end{itemize}

As we assume the observations to be independent, $\boldsymbol g(\boldsymbol \beta^0)$ and $\boldsymbol H(\boldsymbol \beta^0 )$  consist of sums of i.i.d. random variables. Therefore, 
\begin{align*}
    (-\mathbb E(\boldsymbol H(\boldsymbol \beta^0)))^{-1/2} \boldsymbol g(\boldsymbol \beta^0) \xrightarrow{d} \mathcal{N}(\boldsymbol{0}, \boldsymbol{I}),
\end{align*}
which, together with Proposition (\ref{prop:asym}), implies asymptotic normality of $\hat{\boldsymbol \beta}$. As in \citet{Radice2016} note, the normal approximation is not accurate in case the copula parameter is bounded and the sample size is small. 

%%%%%%%%%%%%%%%%%%%%%%%%%%%%%%%%%%%%%%%%%%%%%%%%%%%%%%%%%%%%%%%%%%%%%%
%%%%%%%%%%%%%%%%%%%%%%%%%%%%%%%%%%%%%%%%%%%%%%%%%%%%%%%%%%%%%%%%%%%%% 

\section{Simulations} \label{apx:simulation}

To assess the method's effectiveness, we report the results of four simulation exercises. The first scenario corresponds to our application setting in terms of marginals and type of copula, i.e. using a normal copula and the lognormal distribution as the continuous marginal. We then change either the second marginal to a gamma distribution (scenario 2) or the copula to a Joe copula (scenario 3), or both (scenario 4). 

In all scenarios, the response vector consists of an ordinal outcome with 3 levels and a continuous variable which follows either a lognormal (scenario 1 \& 3) or a gamma (scenario 2 \& 4) distribution. The five covariates, $x_1, x_2, x_3, \nu_1, \nu_2$ are uniformly distributed on the $[-2,2]$ interval and two of them, $\nu_1$ and $\nu_2$, enter the model in a non-linear fashion. Each parameter's additive predictor is constructed in a way that they are roughly symmetrical around zero with most values in the $[-3,3]$ interval. The copula is Gaussian in scenario 1 and Joe copula in scenario 2.  

The $i$-th predictors are constructed as follows:
\begin{align*}
    \eta_i^{\mu_1} & = \theta_{ri} -                          \left(\beta_1^{\mu_1}x_{1i} +
                                   s_1^{\mu_1}(\nu_{1i}) +
                                   s_2^{\mu_1}(\nu_{2i}) \right)\\
    \eta_i^{\mu_2} & = \beta_0^{\mu_2} +
                                  \beta_1^{\mu_2}x_{1i} +
                                  \beta_2^{\mu_2}x_{2i} +
                                   s_3^{\mu_2}(\nu_{1i})\\
    \eta_i^{\sigma_2} & = \beta_0^{\sigma_2} +                \beta_1^{\sigma_2}x_{3i}\\
    \eta_i^{\gamma} & = \beta_0^{\gamma} + s_3^{\gamma}(\nu_{2i}),\\  
\end{align*} %	formula.1 <- r1 ~ D1 + s(Z1) + s(Z2)
%	formula.2 <- r2 ~ D1 + D2 + s(Z1)
%	eq.sigma2 <- ~ D3
%	eq.theta  <- ~ s(Z2)
where $s_1, s_2, s_3$ are the three different smooth functions below:
\begin{align*}
 s_1(\nu) &= \nu * \sin(3 * \nu) \\
  s_2(\nu) & = \sin(2*\nu) + 0.5 * \nu \\
  s_3(\nu) &=  3*\nu* \cos(\nu). 
\end{align*} 
For each scenario we run three different cases that differ in number of observations. We consider case 1 with $n= 1,000$, case 2 with $n=3,000$ and case 3 with $n = 10,000$. For each iteration, we store the coefficients of the linear effects and estimated smooth functions. 
There were also iterations that gave warning messages about the convergence of the model. As these warnings often indicates that the model is too complicated for the number of observation, convergence fails more frequently for the cases with $n=1,000$. Sometimes, even if these warnings occur, the fit might be reasonable. Users are nonetheless advised to check the model carefully whenever warning messages are returned. We decided to drop iterations with warnings and move on to the next ones until 100 simulation runs without warnings were obtained. The number of iterations displaying warning messages are reported in Table \ref{tab:sim_warnings}. The number of those warnings decreases drastically as the sample size increases. For scenario 1, no iterations had to be excluded due to non-convergence, while for scenarios 3 and 4 more iterations returned warning messages. Although the number of parameters is equal in all four scenarios, based on the number of repetitions, it seems that scenarios 1 and 2 using a Gaussian copula are relatively easier to be fitted than scenarios 3 and 4 which employ a Joe copula.    

Note that the aim of this simulation is to check the implementation and the ability to estimate reliably the model's coefficients. How well \texttt{GJRM()} selects the correct copula specification via the AIC and BIC is considered elsewhere \citep[e.g.,][]{Marra.2017,Radice2016}.   

Figure \ref{fig:sim_lin_scen1} shows the boxplots of all estimated coefficients for the linear effects in scenario 1. The estimator captures the effect fairly well and the performance improves with the sample size.

\begin{figure} 
    %\centering%
  \begin{subfigure}{0.5\textwidth}
    % \hspace{-10mm}
    \includegraphics[scale = 0.4]{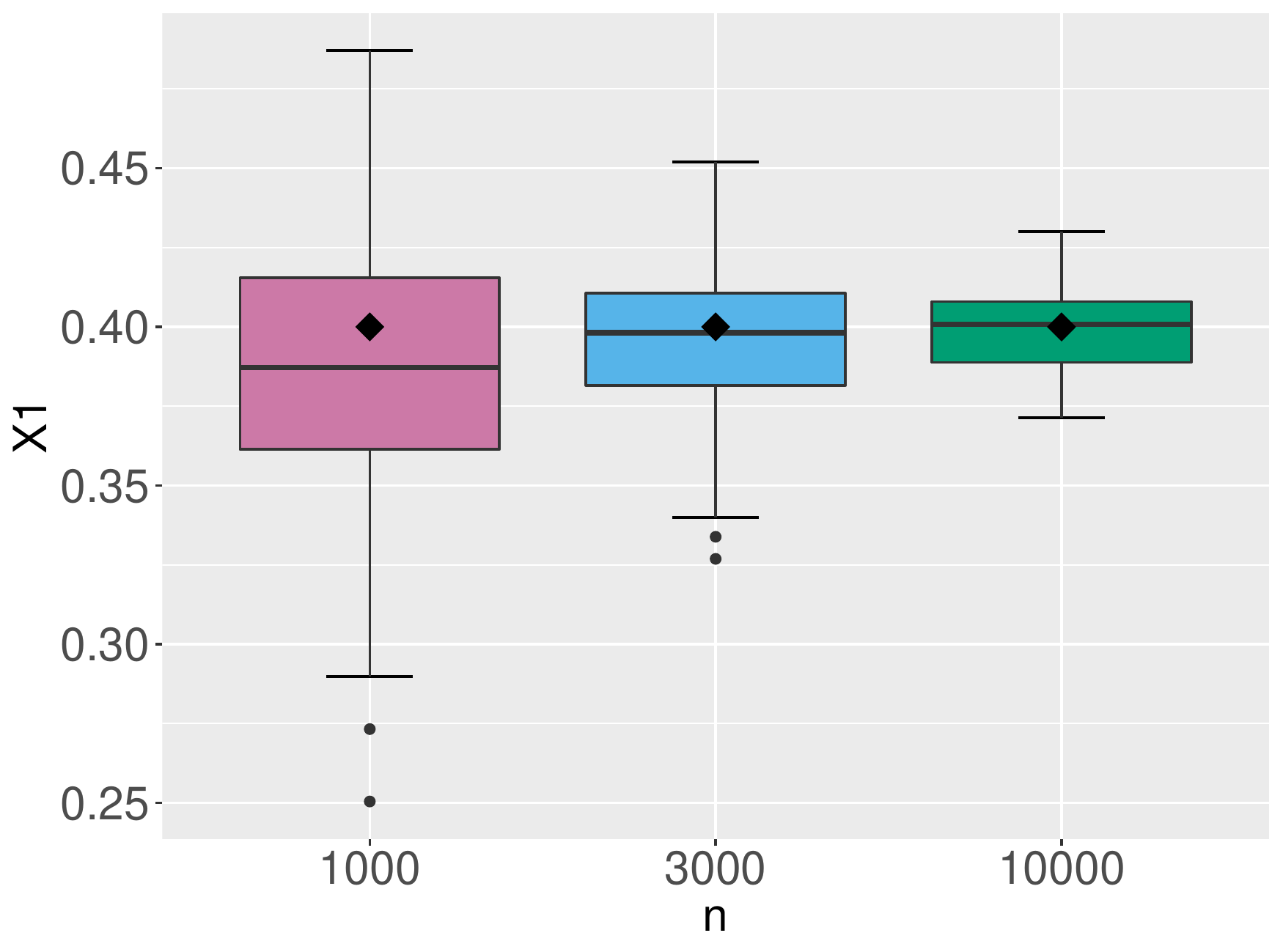}
    \caption{coefficients for $x_1$ of $\mu_1$}
  \end{subfigure}%
  ~ 
  \begin{subfigure}{0.5\textwidth}
    %\hspace{-5mm}%
    \includegraphics[scale = 0.4]{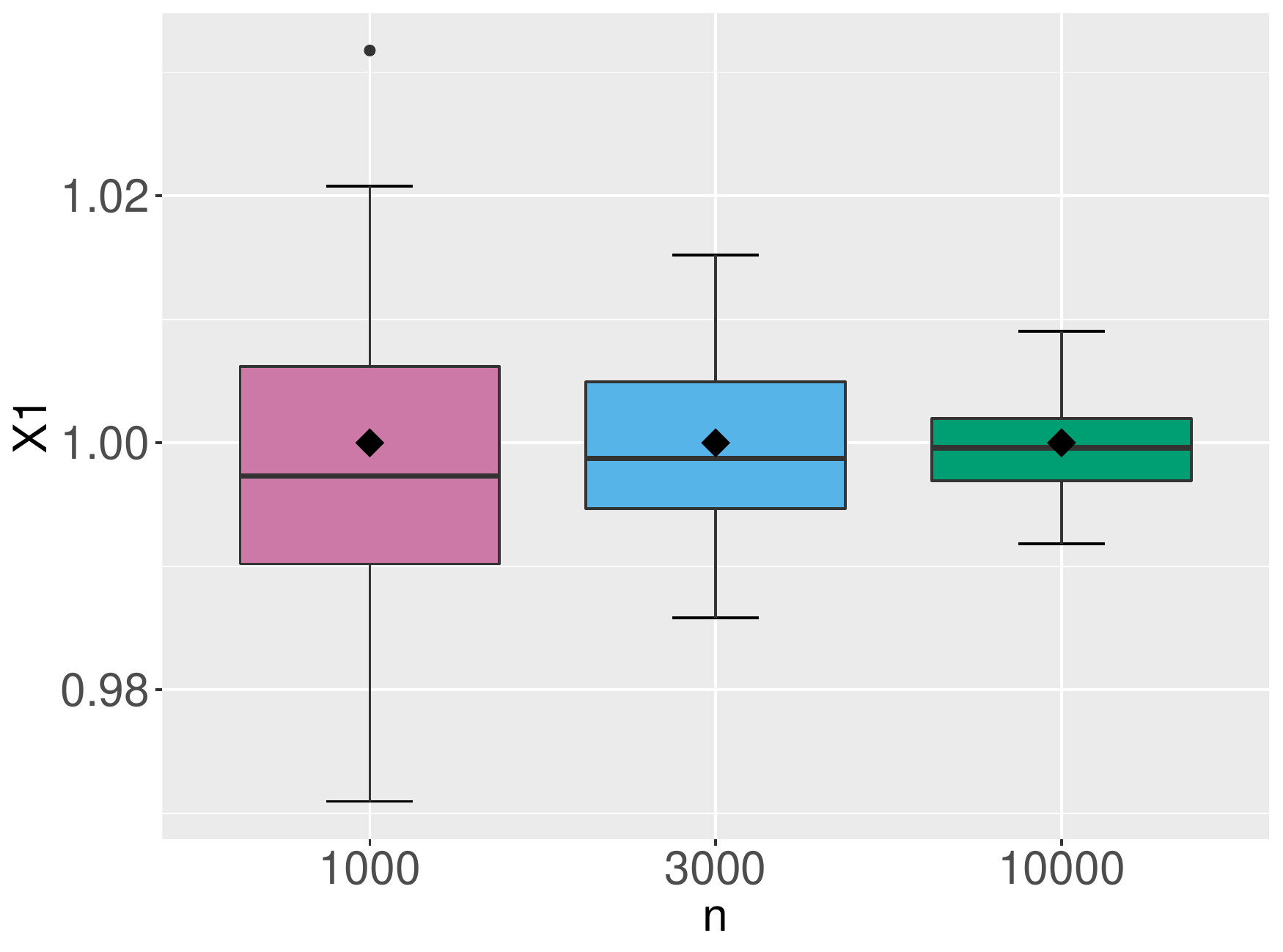}
    \caption{coefficients for $x_1$ of $\mu_2$}
  \end{subfigure}
  \begin{subfigure}{0.5\textwidth}
 %   \hspace{-10mm}
    \includegraphics[scale = 0.4]{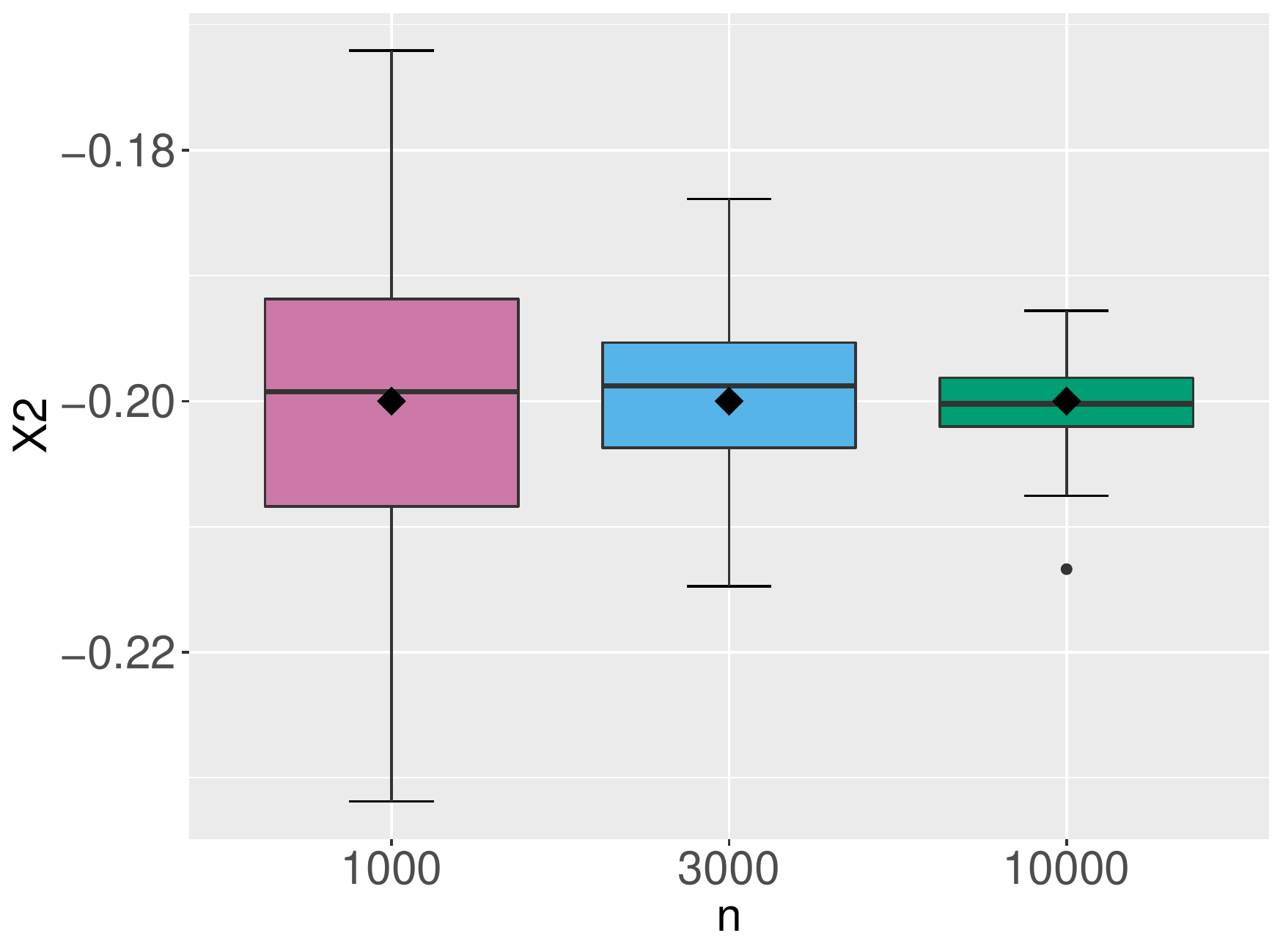}
    \caption{coefficients for $x_2$ of $\mu_2$}
  \end{subfigure}%
  ~
  \begin{subfigure}{0.5\textwidth}
    %  \hspace{-5mm}%
    \includegraphics[scale = 0.4]{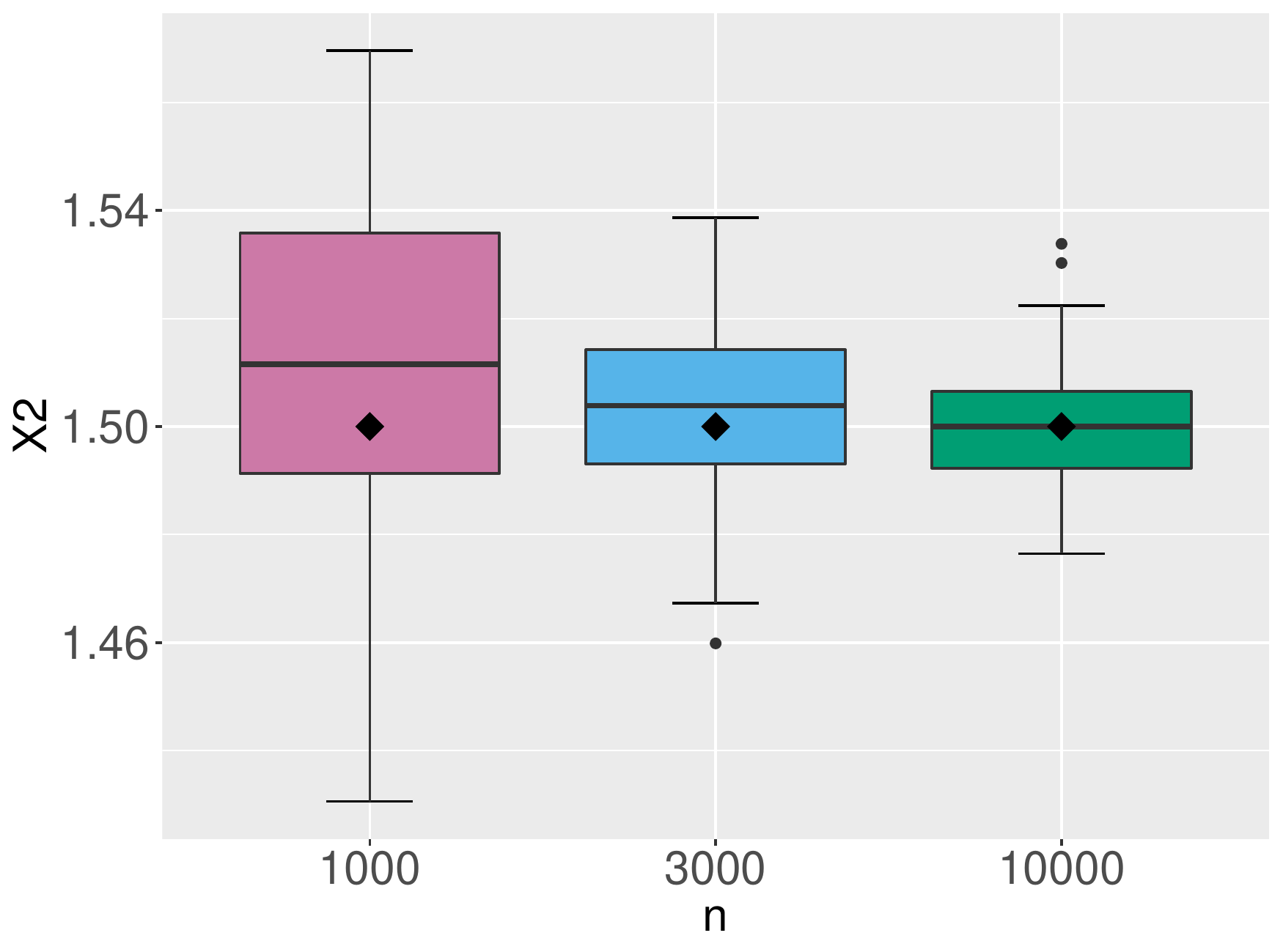}
    \caption{coefficients for $x_3$ of $\sigma_2$}
  \end{subfigure}
\caption{Simulation results for linear effects in scenario 1. The boxplots represent the estimated linear coefficients in $N = 100$ iterations. The true values of the coefficients are denoted by the black diamond symbols.} \label{fig:sim_lin_scen1}
\end{figure}

%However, the coefficient for the scale parameter is estimated with a small upward bias.
Figure \ref{fig:sim_smooth_scen1} exemplarily shows the estimated smooth functions against the true ones for scenario 1 and $n= 1,000$. The procedure is able to recover the smooth functions satisfactorily although some of the curves are wigglier or too smooth compared to the true functions. This effect vanishes in the other cases as the sample size grows. Note that for the location parameter of the continuous marginal, the smooth effects seem to be easier to estimate as all curves are very close to the true functions. If we increase the sample size, the fit of the smooth functions improves significantly especially at the local minima and maxima (plots are available upon request). 

\begin{figure}
    %\centering%
  \begin{subfigure}{0.5\textwidth}
   %  \hspace{-10mm}
    \includegraphics[scale = 0.4]{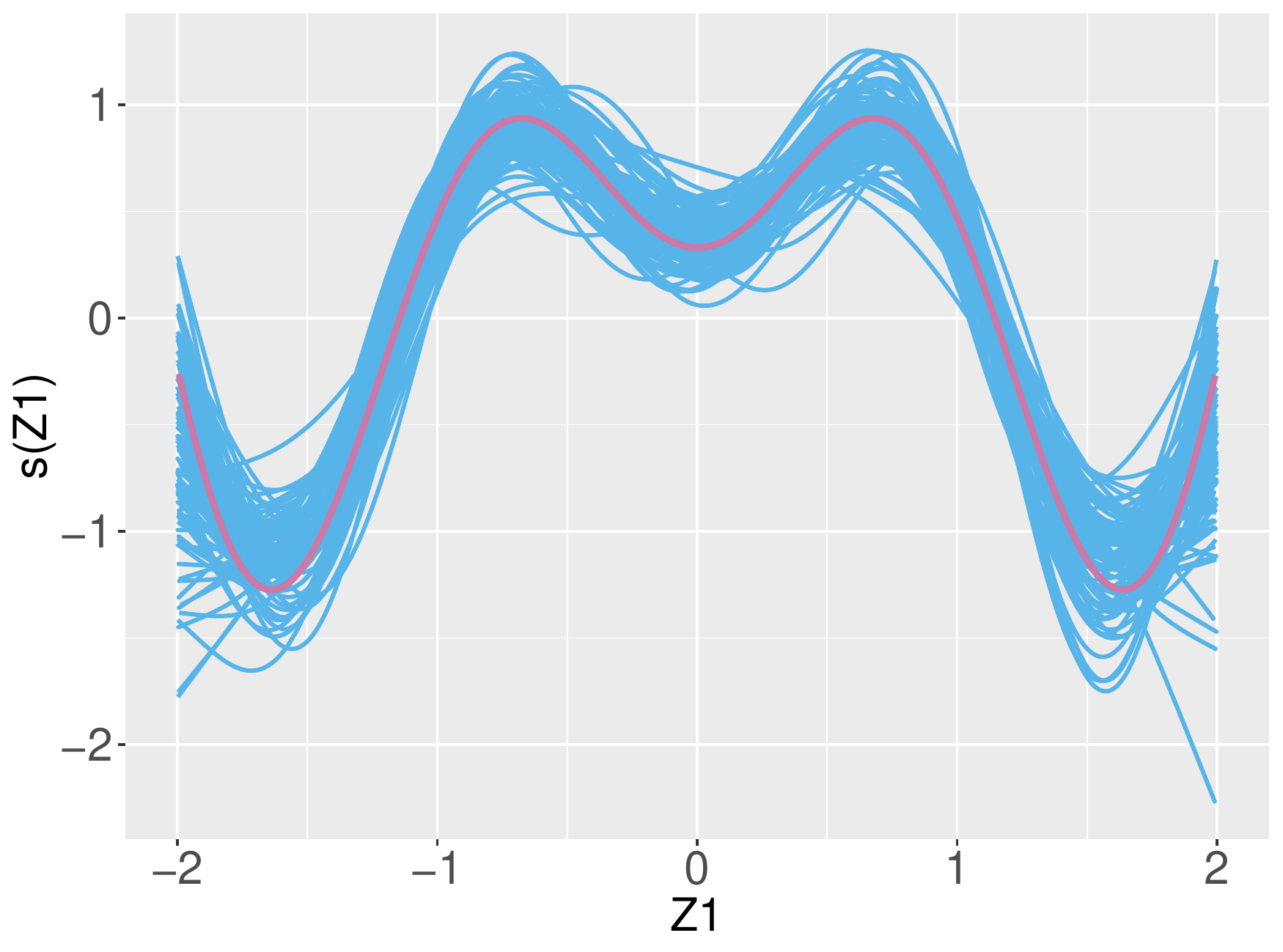}
    \caption{$s_1^{\mu_1}(\nu_1)$}
  \end{subfigure}%
  ~ 
  \begin{subfigure}{0.5\textwidth}
 %   \hspace{-5mm}%
    \includegraphics[scale = 0.4]{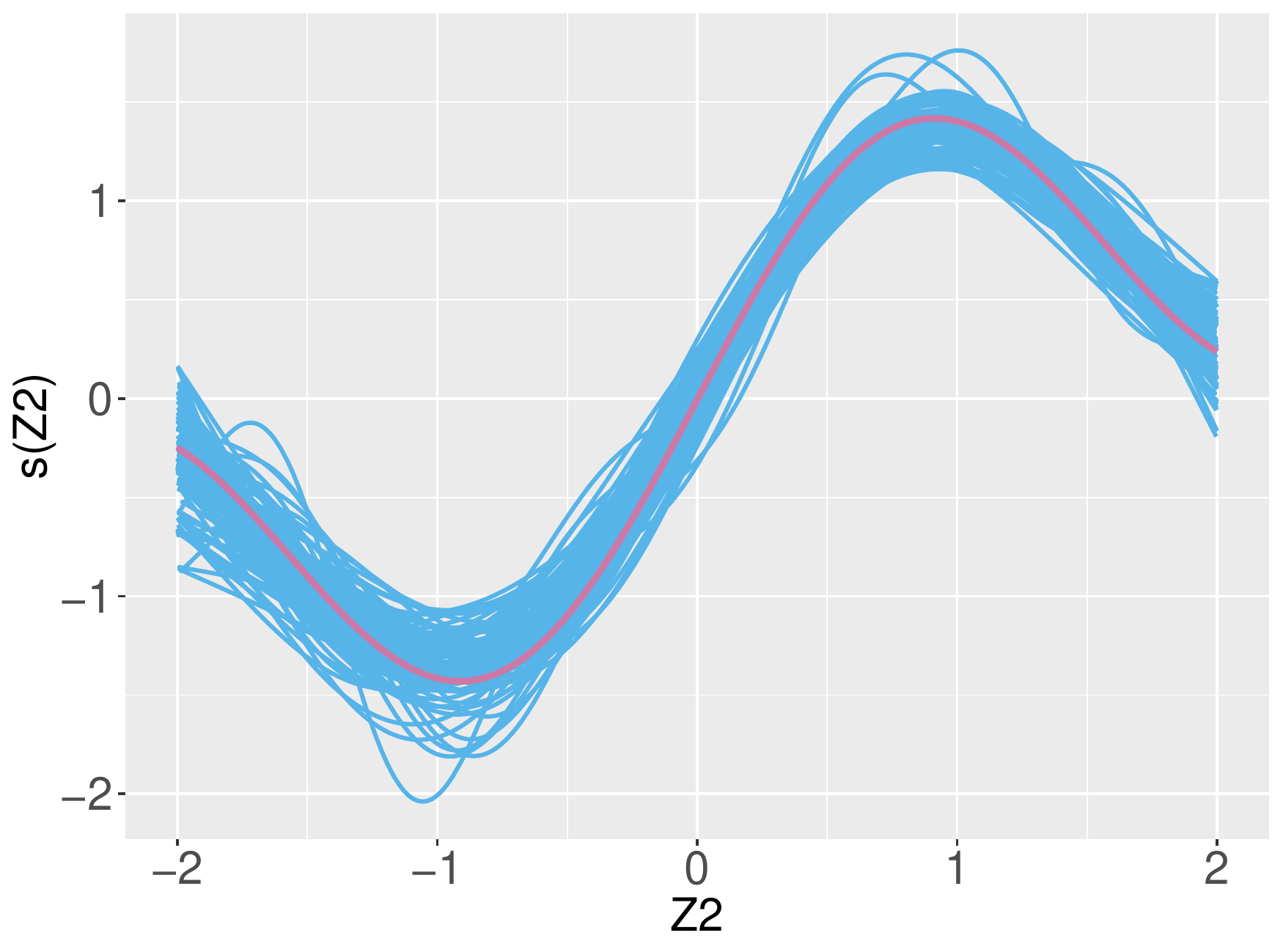}
    \caption{$s_2^{\mu_1}(\nu_2)$}
  \end{subfigure}
  \begin{subfigure}{0.5\textwidth}
 %   \hspace{-10mm}
    \includegraphics[scale = 0.4]{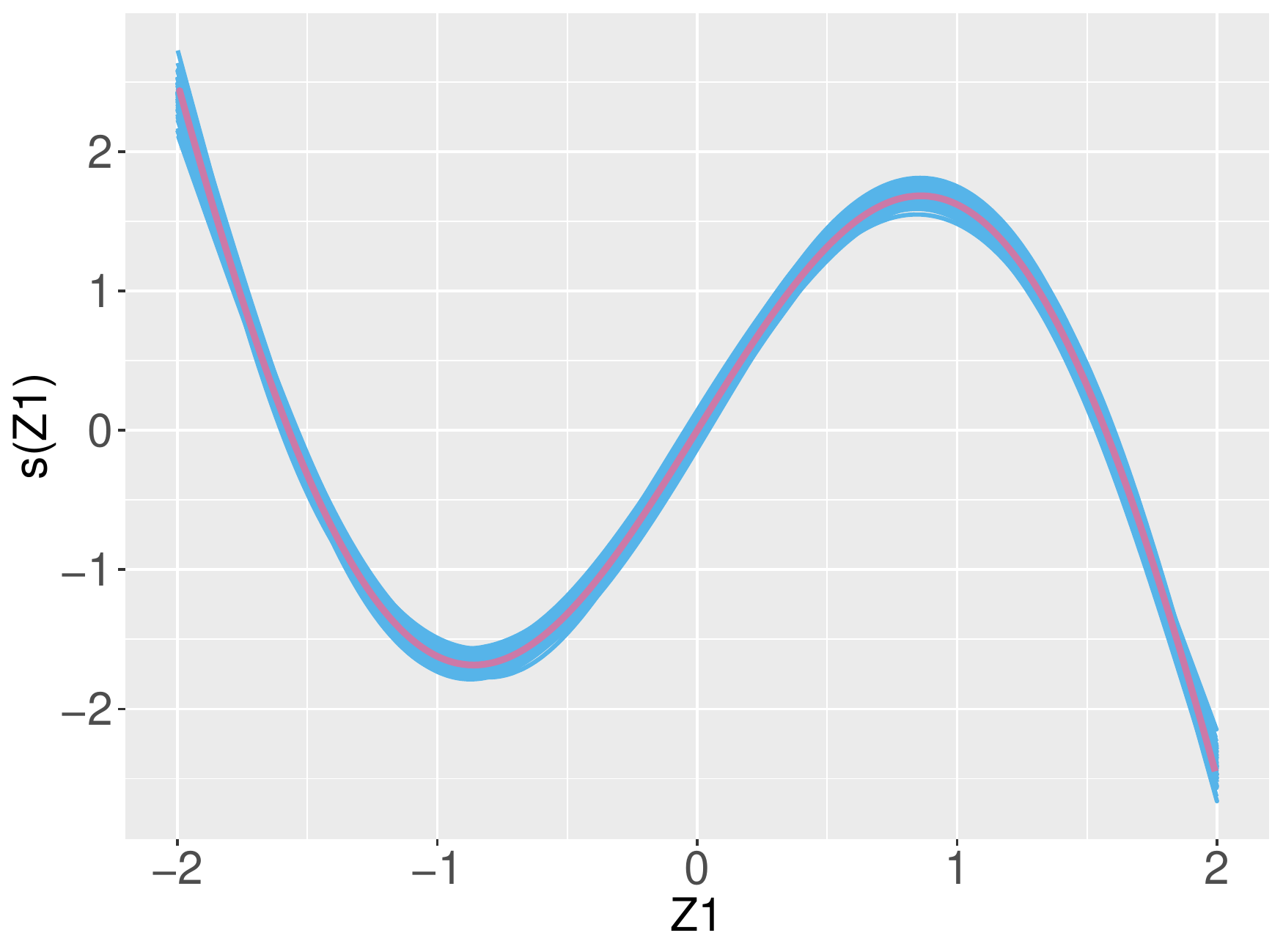}
    \caption{$s_3^{\mu_2}(\nu_1)$}
  \end{subfigure}%
  ~
  \begin{subfigure}{0.5\textwidth}
 %     \hspace{-5mm}%
    \includegraphics[scale = 0.4]{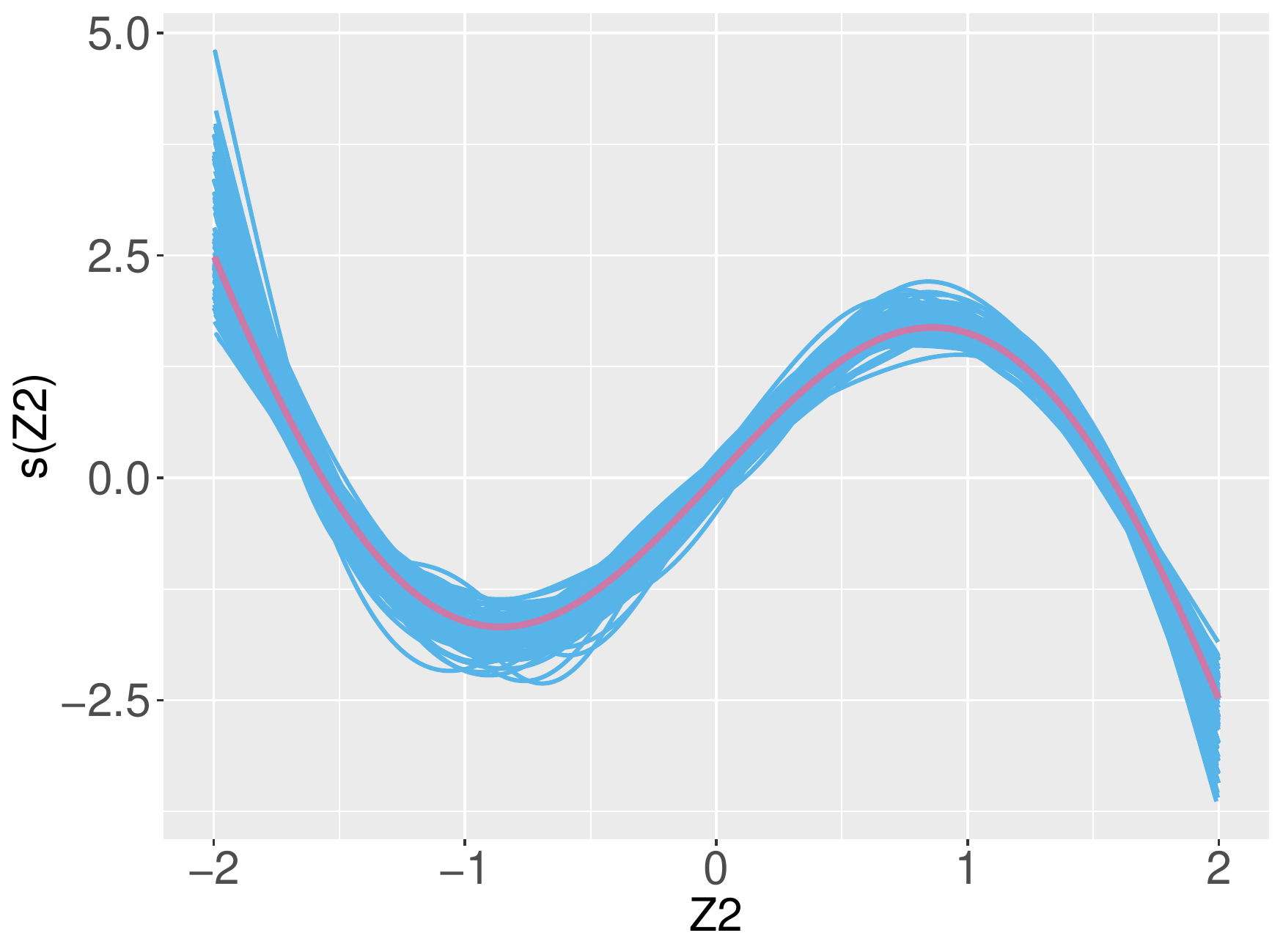}
    \caption{$s_3^{\gamma}(\nu_2)$}
  \end{subfigure}
\caption{Simulation results for non-linear effects in scenario 1. The pink solid lines represent the true functions.\label{fig:sim_smooth_scen1}} 
\end{figure}
%%%%%%%%%%%%%%%%%%%%%%%%%%

Note that some simulation iterations for the case of $n=1,000$ were problematic in that the smooth functions were not estimated adequately. Their number was rather small and we excluded them. Table \ref{tab:sim_warnings} gives an overview of how many iterations were excluded.                 

The results for scenarios 2, 3 and 4 are similar to the first one and shown in Figures \ref{fig:sim_smooth_scen2}, \ref{fig:sim_smooth_scen3} and \ref{fig:sim_smooth_scen4}, respectively. Again, the estimation of the smooth terms improves considerably as the sample size increases but we only show the more difficult case of $n=1,000$ here to demonstrate that results are still acceptable also at relatively small sample sizes.  

\begin{table}
    \centering
    \begin{tabular}{llcc}
        \hline
        scenario  & case & rejected by algorithm & manually deleted \\ \hline 
         scenario 1 & n = 1,000 & 0 & 3 \\ 
         & n = 3,000 & 0 & 0\\
         & n = 10,000 & 0 & 0\\
          scenario 2 & n = 1,000 & 1 & 1 \\ 
         & n = 3,000 & 1 & 0\\
          & n = 10,000 & 1 & 0\\
          scenario 3 (J0) & n = 1,000 & 62 & 2\\ 
          & n = 3,000 & 3 & 0\\
          & n = 10,000 & 0 & 0\\
          scenario 4 & n = 1,000 & 77 & 1 \\ 
          & n = 3,000 & 3 & 0\\
          & n = 10,000 & 2 & 0\\ \hline          
    \end{tabular}
    \caption{Number of repeated iterations due to warning messages and number of manually deleted iterations that were extreme outliers.}
    \label{tab:sim_warnings}
\end{table}

%%%%%%%%%%

\begin{figure} 
    \centering
  \begin{subfigure}{0.5\textwidth}
   %  \hspace{-10mm}
    \includegraphics[scale = 0.35]{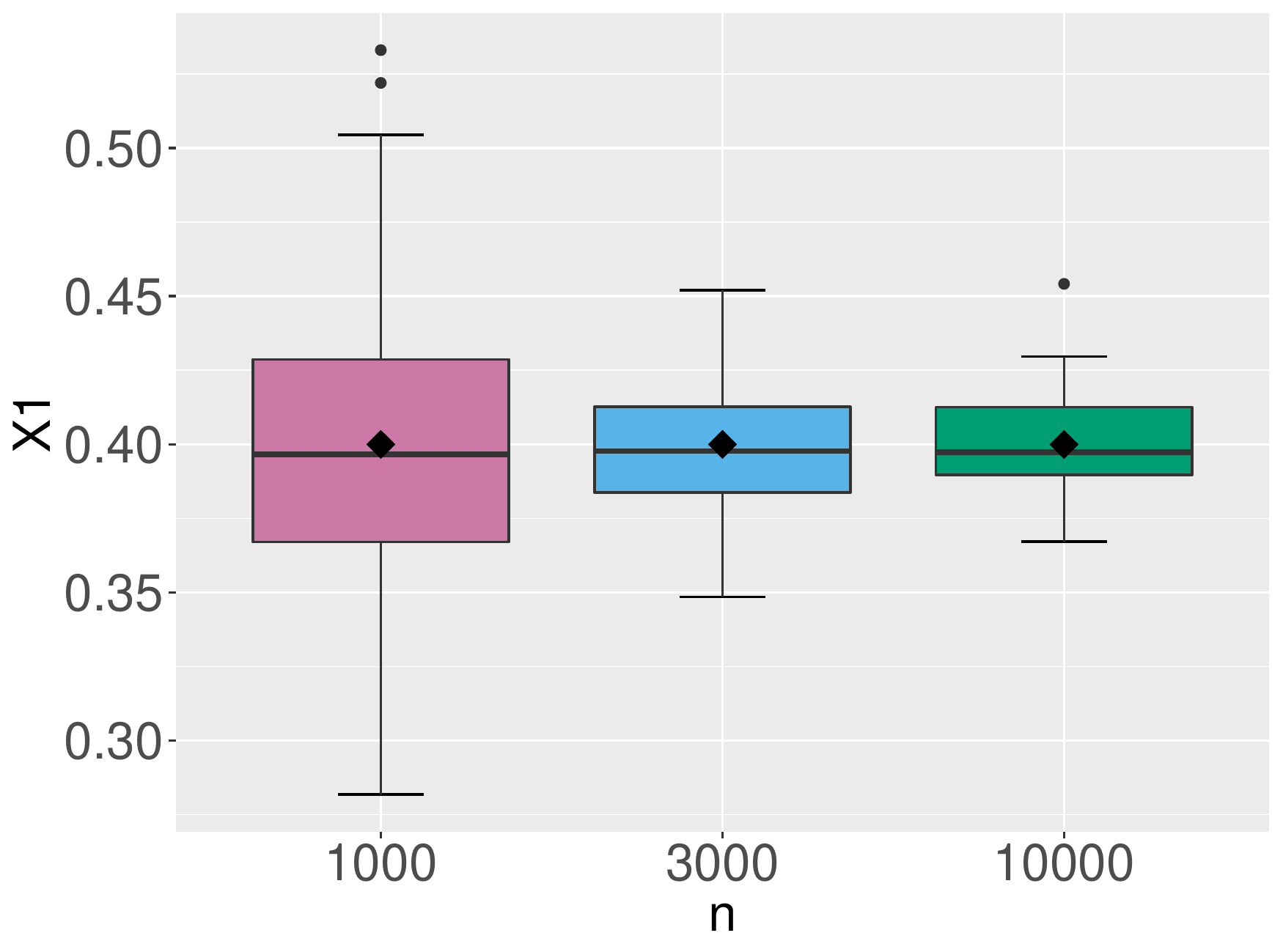}
    \caption{coefficients for $x_1$ of $\mu_1$}
  \end{subfigure}%
  ~ 
  \begin{subfigure}{0.5\textwidth}
    %\hspace{-5mm}%
    \includegraphics[scale = 0.35]{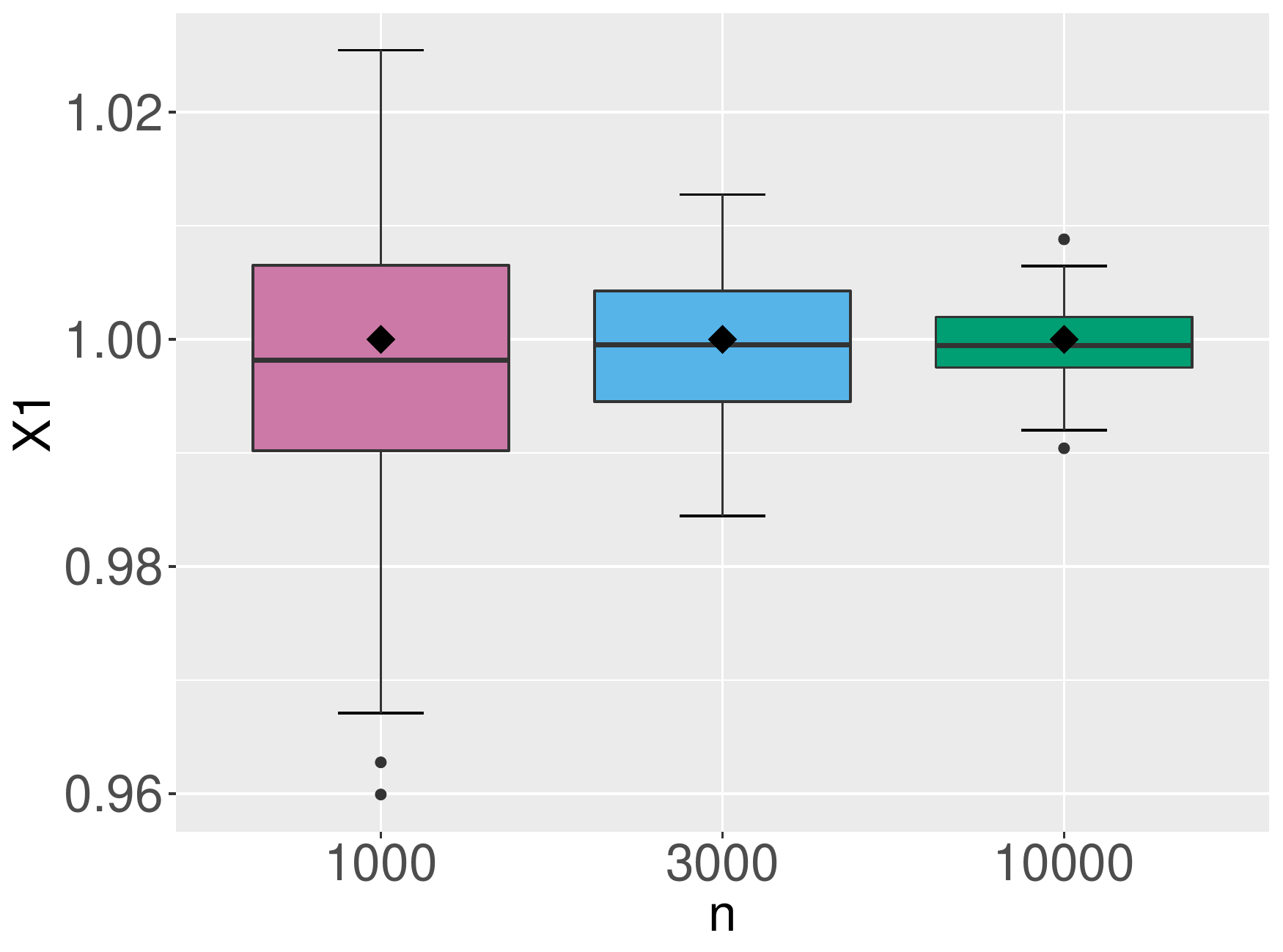}
    \caption{coefficients for $x_1$ of $\mu_2$}
  \end{subfigure}
  \begin{subfigure}{0.5\textwidth}
  %  \hspace{-10mm}
    \includegraphics[scale = 0.35]{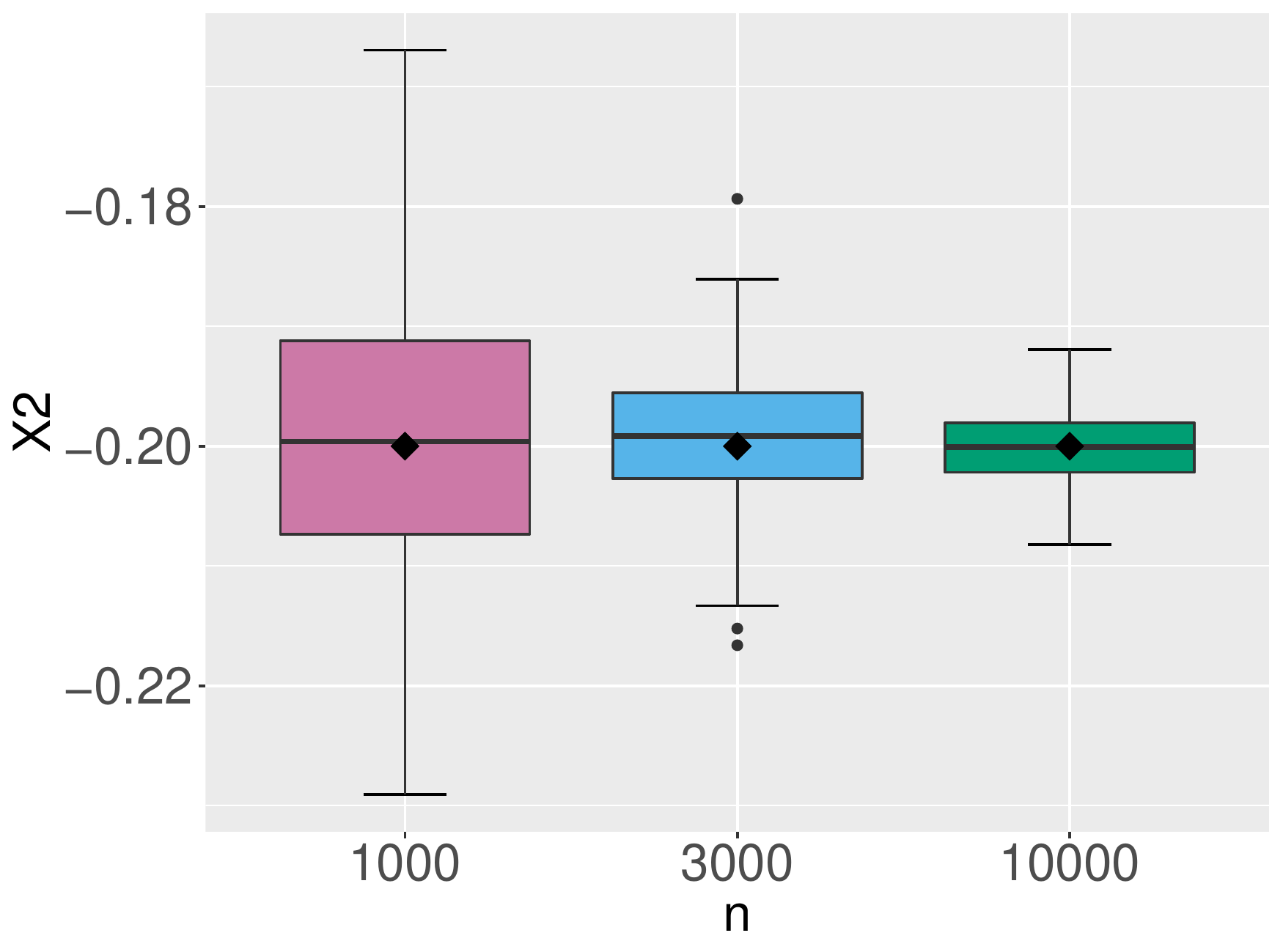}
    \caption{coefficients for $x_2$ of $\mu_2$}
  \end{subfigure}%
  ~
  \begin{subfigure}{0.5\textwidth}
  %    \hspace{-5mm}%
    \includegraphics[scale = 0.35]{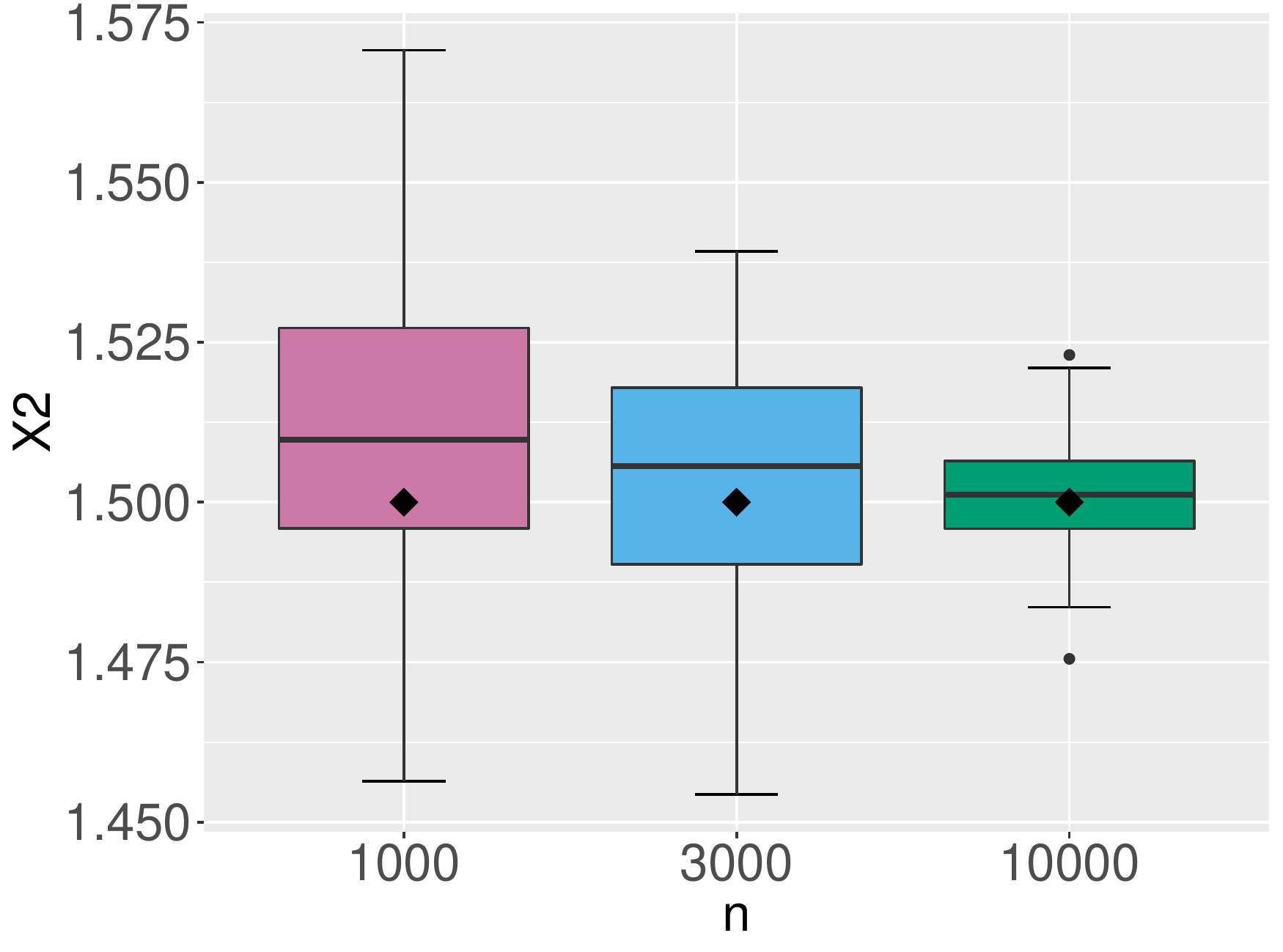}
    \caption{coefficients for $x_3$ of $\sigma_2$}
  \end{subfigure}
  \begin{subfigure}{0.5\textwidth}
   %  \hspace{-10mm}
    \includegraphics[scale = 0.35]{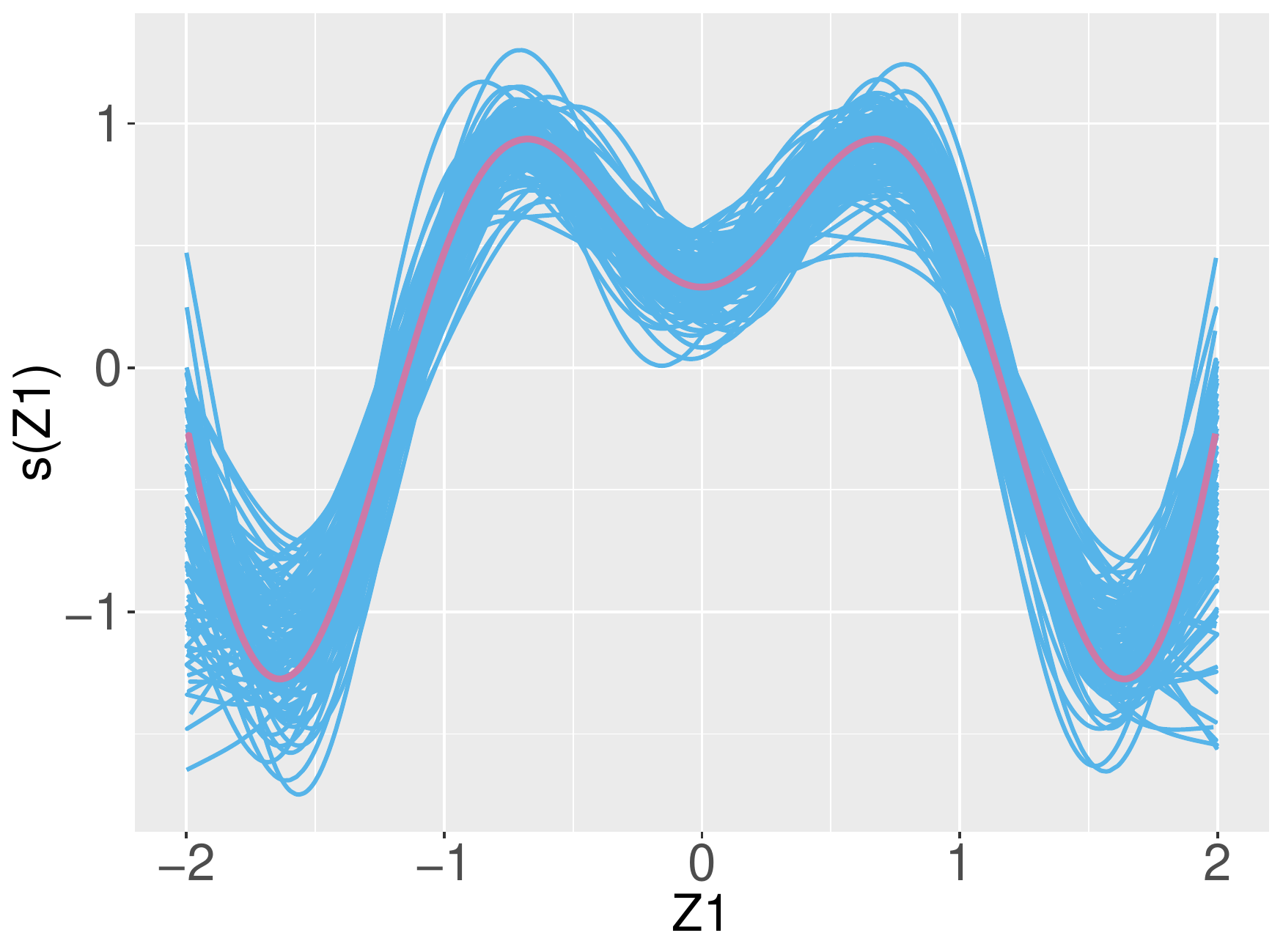}
    \caption{$s_1^{\mu_1}(\nu_1)$}
  \end{subfigure}%
  ~ 
  \begin{subfigure}{0.5\textwidth}
    %\hspace{-5mm}%
    \includegraphics[scale = 0.35]{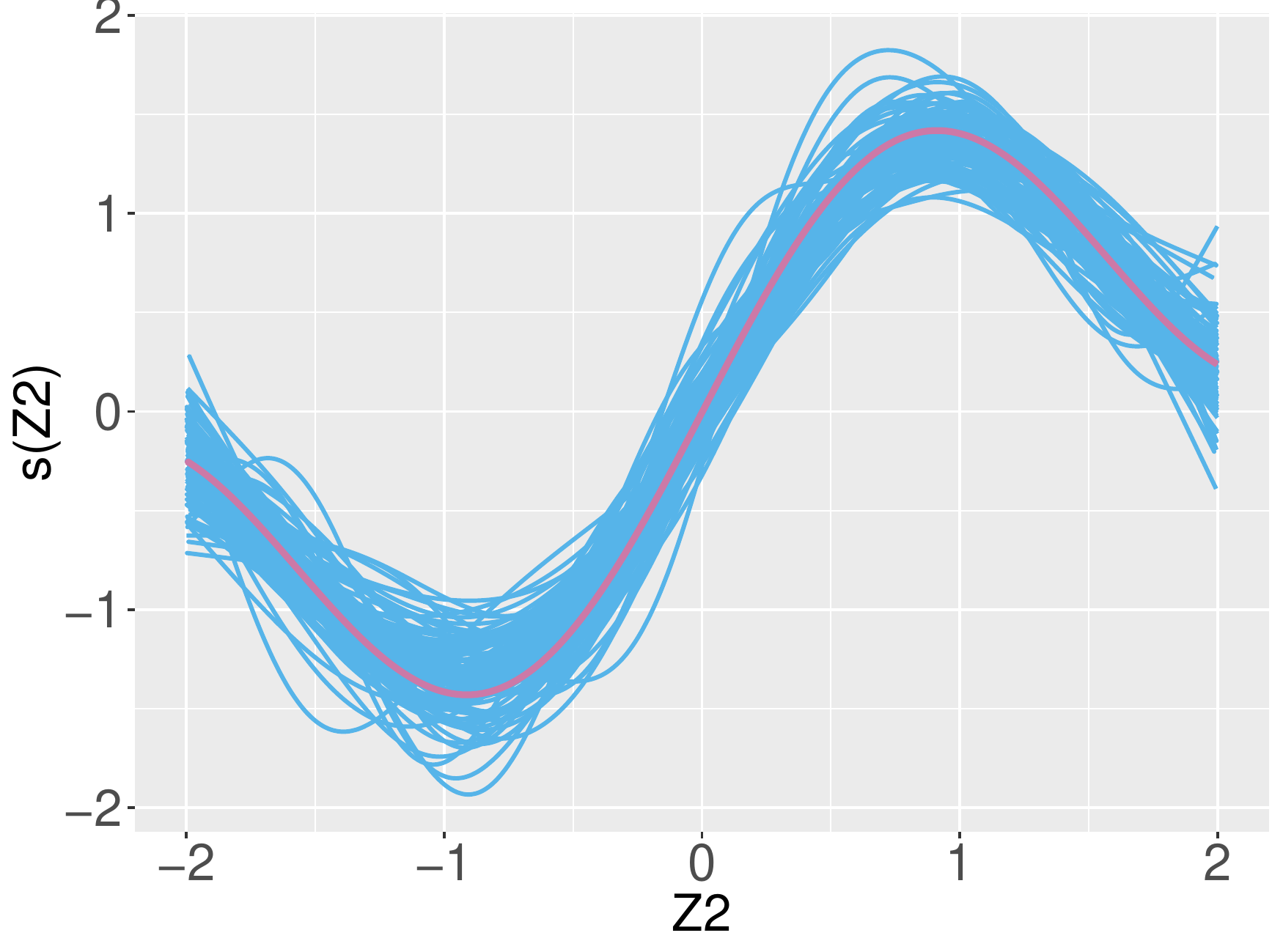}
    \caption{$s_2^{\mu_1}(\nu_2)$}
  \end{subfigure}
  \begin{subfigure}{0.5\textwidth}
   % \hspace{-10mm}
    \includegraphics[scale = 0.35]{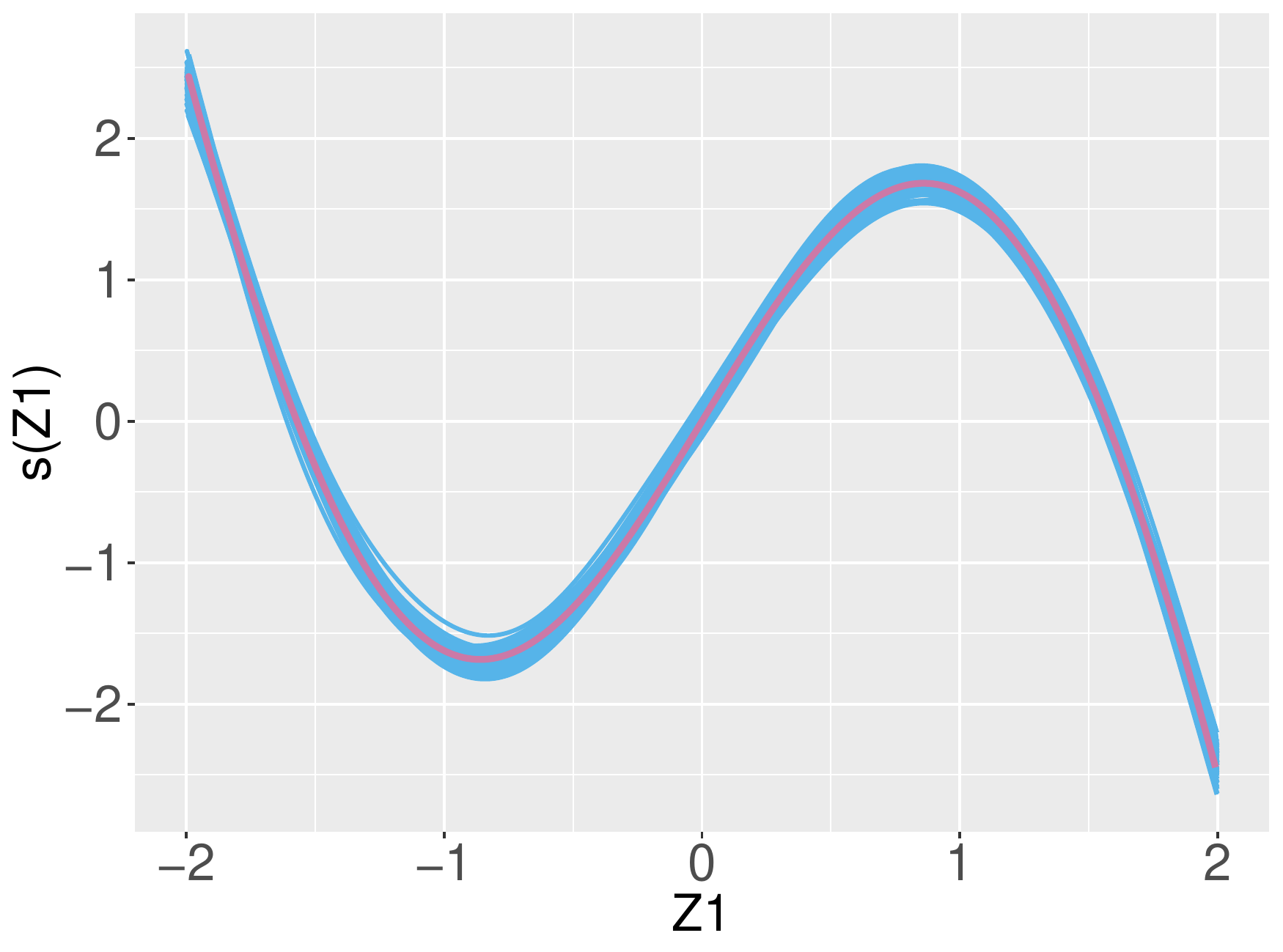}
    \caption{$s_3^{\mu_2}(\nu_1)$}
  \end{subfigure}%
  ~
  \begin{subfigure}{0.5\textwidth}
     % \hspace{-5mm}%
    \includegraphics[scale = 0.35]{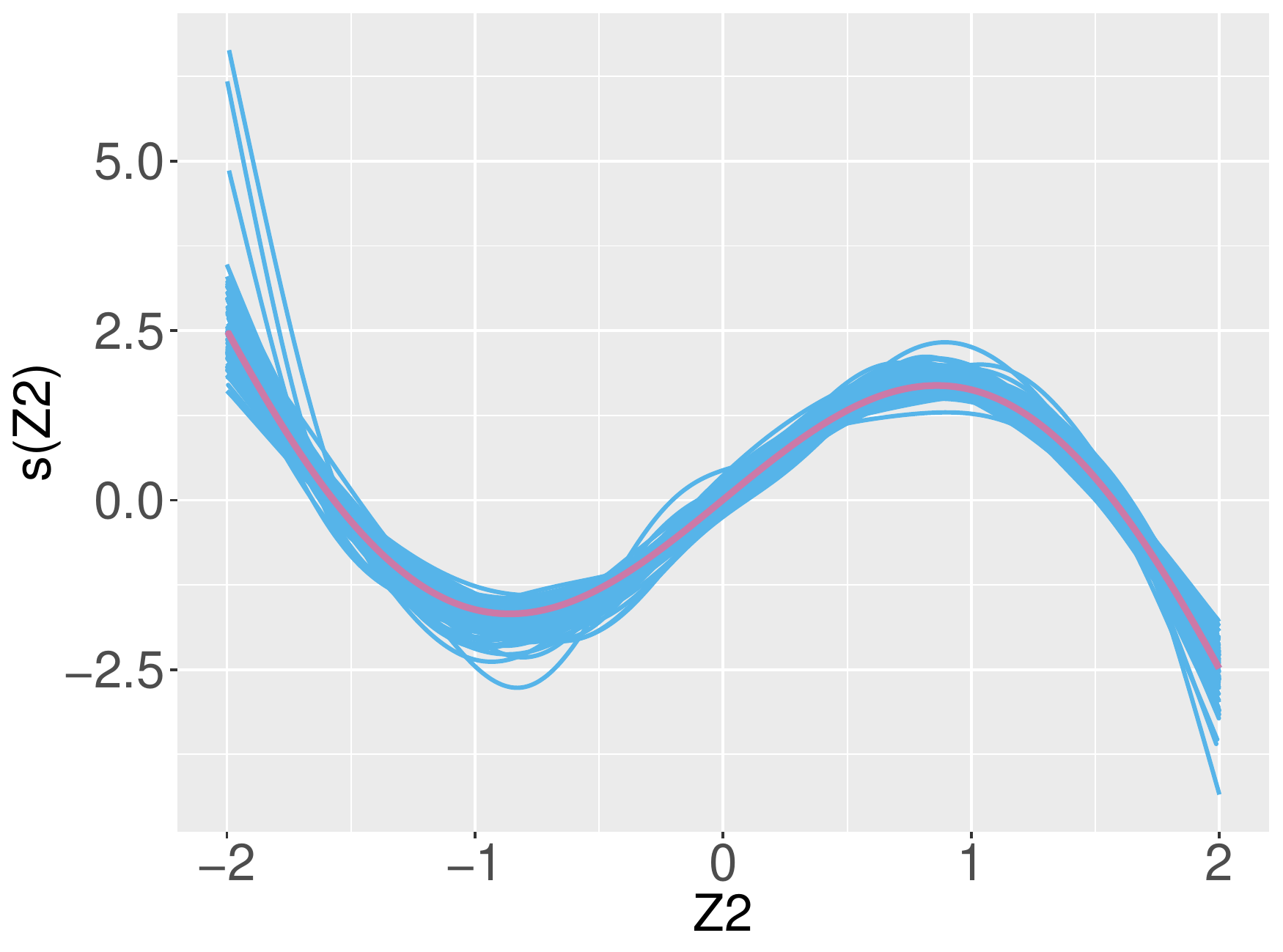}
    \caption{$s_3^{\gamma}(\nu_2)$}
  \end{subfigure}
\caption{Simulation results for linear and non-linear effects (case $n=1000$) in scenario 2. 
The boxplots represent the estimated linear coefficients in $N = 100$ iterations. The true values of the coefficients are denoted by the black diamond symbols. The pink solid lines represent the true functions of the non-linear effects.} \label{fig:sim_smooth_scen2}
\end{figure}

%%%%%%

\begin{figure} 
    %\centering%
  \begin{subfigure}{0.5\textwidth}
     %\hspace{-10mm}
    \includegraphics[scale = 0.35]{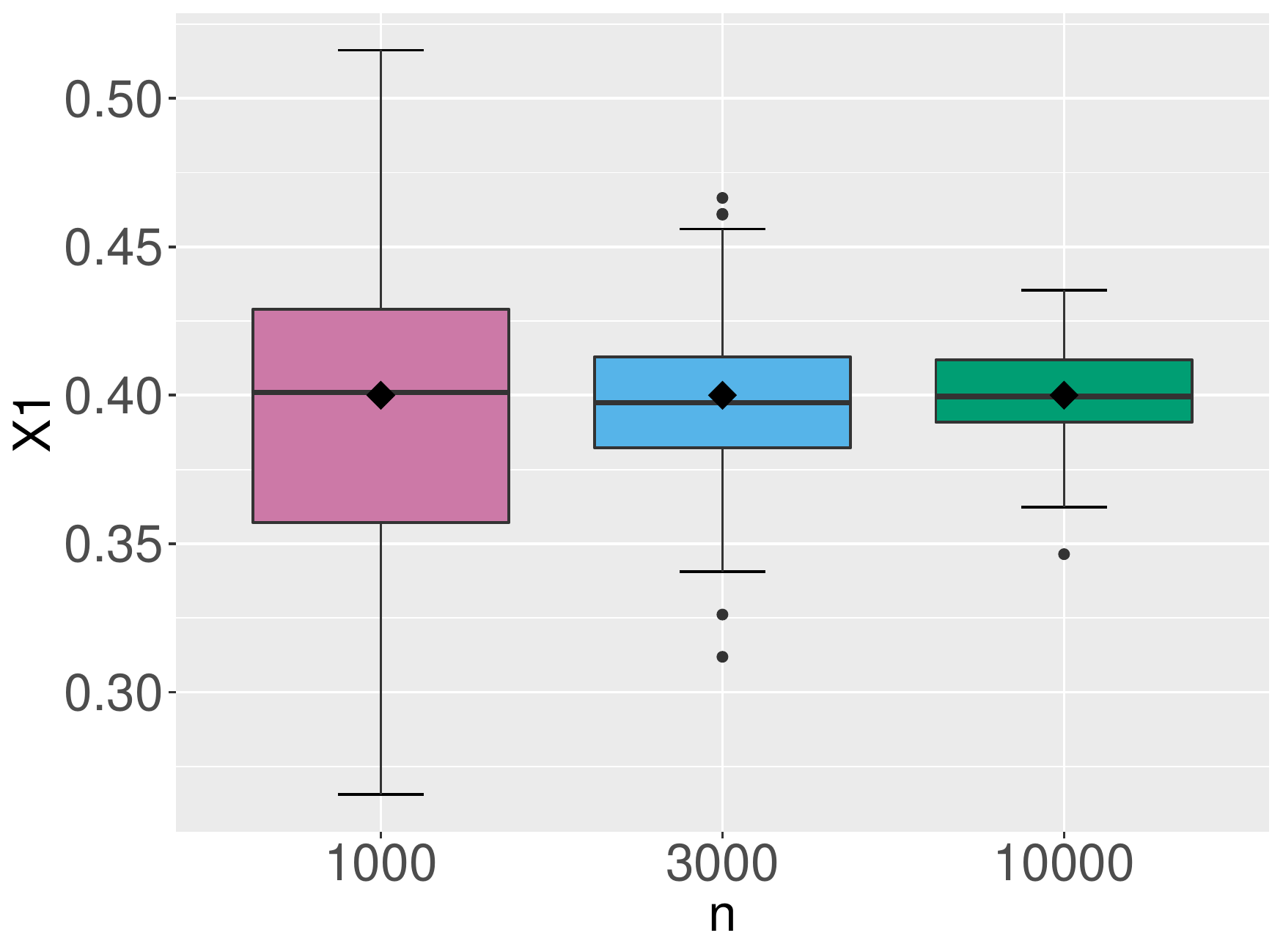}
    \caption{coefficients for $x_1$ of $\mu_1$}
  \end{subfigure}%
  ~ 
  \begin{subfigure}{0.5\textwidth}
   % \hspace{-5mm}%
    \includegraphics[scale = 0.35]{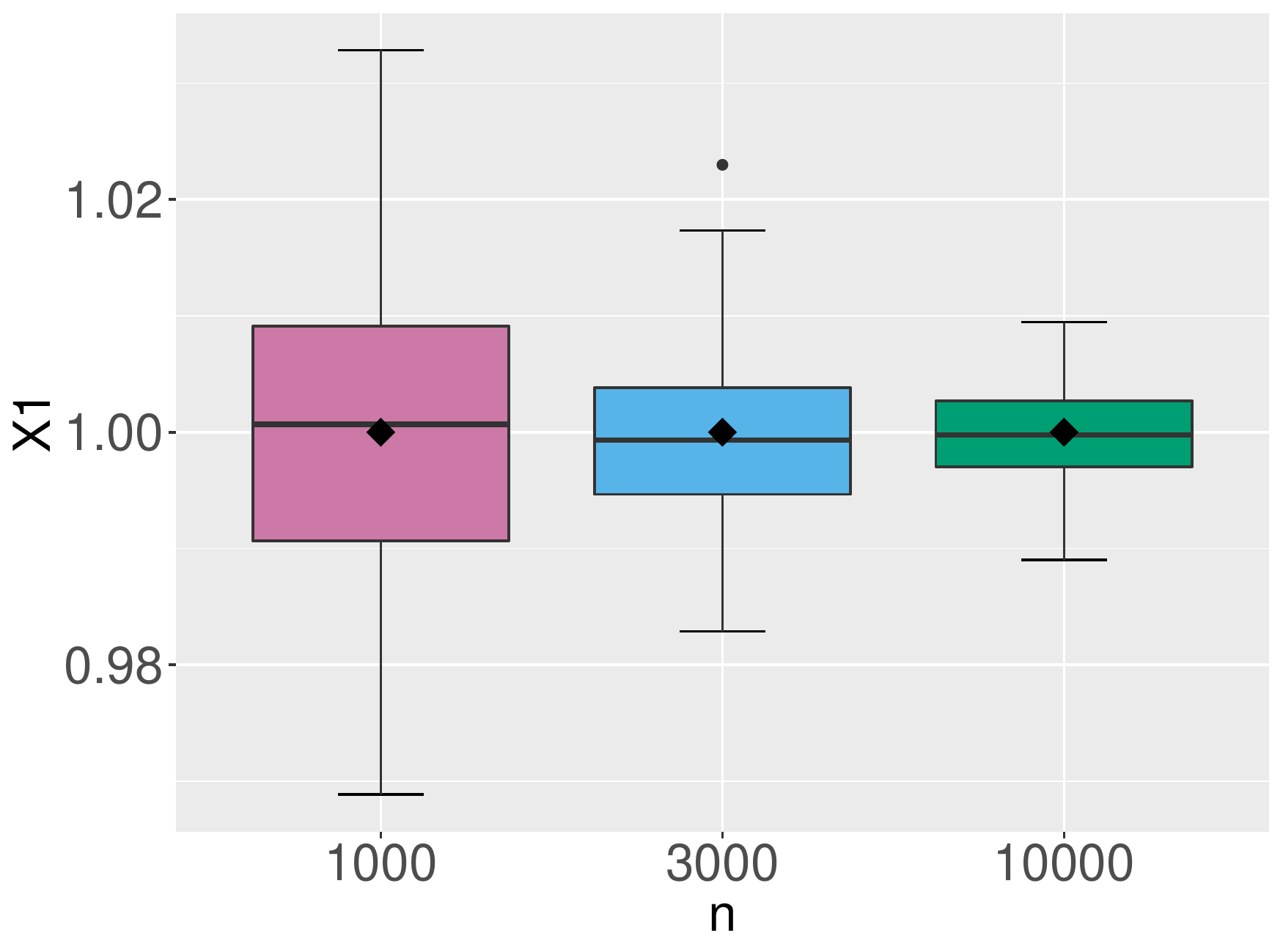}
    \caption{coefficients for $x_1$ of $\mu_2$}
  \end{subfigure}
  \begin{subfigure}{0.5\textwidth}
   % \hspace{-10mm}
    \includegraphics[scale = 0.35]{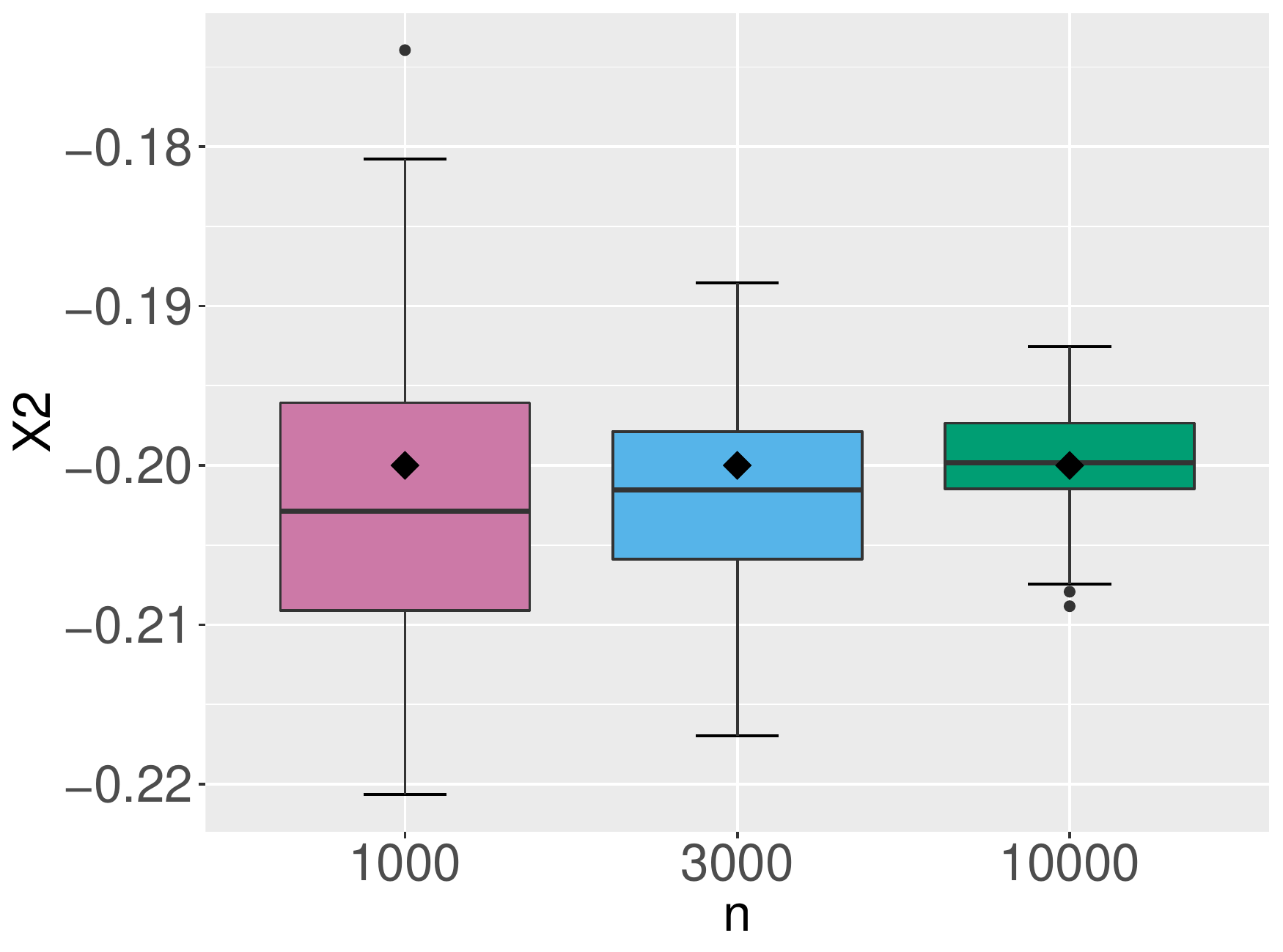}
    \caption{coefficients for $x_2$ of $\mu_2$}
  \end{subfigure}%
  ~
  \begin{subfigure}{0.5\textwidth}
     % \hspace{-5mm}%
    \includegraphics[scale = 0.35]{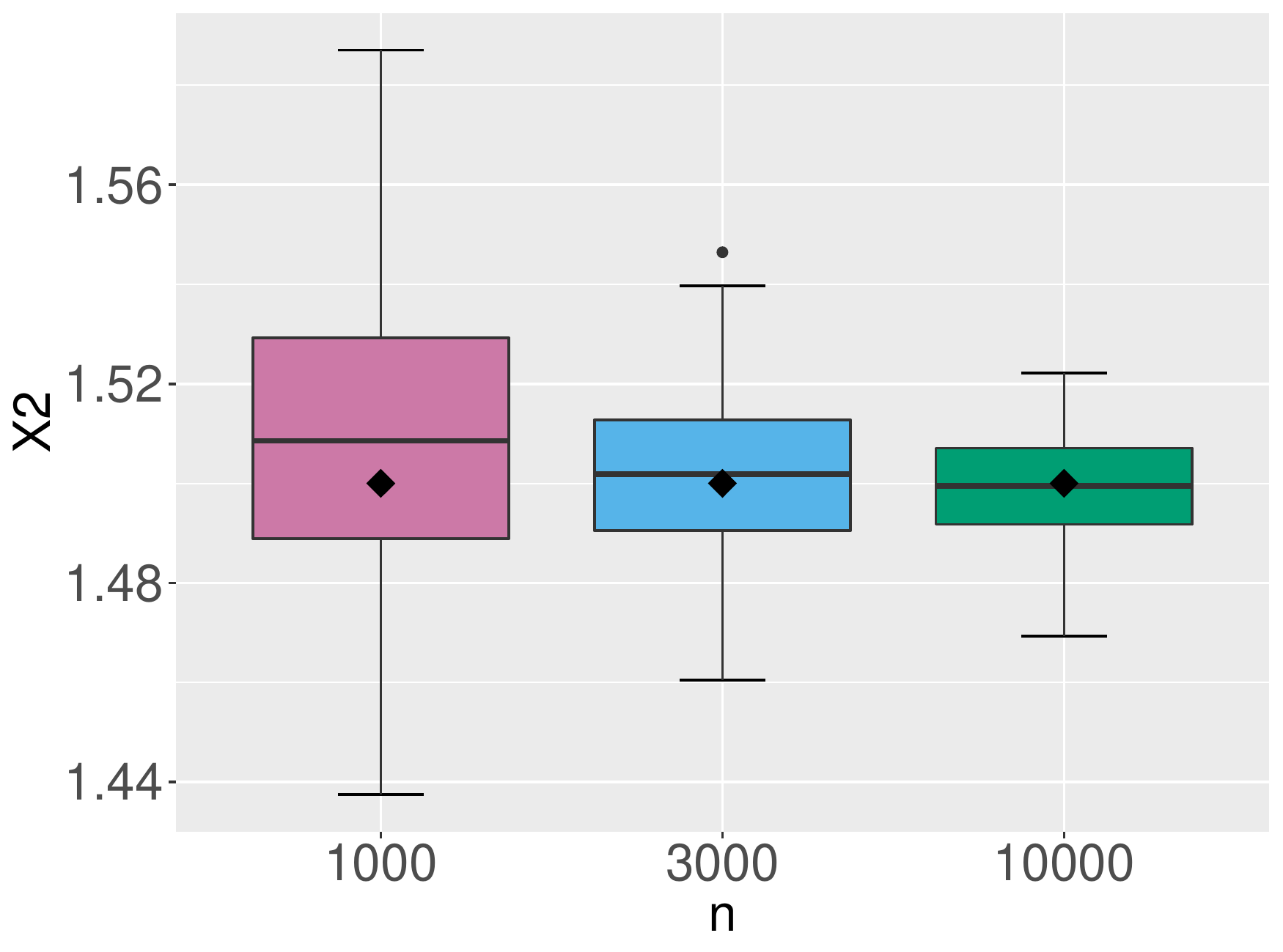}
    \caption{coefficients for $x_3$ of $\sigma_2$}
  \end{subfigure}
  \begin{subfigure}{0.5\textwidth}
   %  \hspace{-10mm}
    \includegraphics[scale = 0.35]{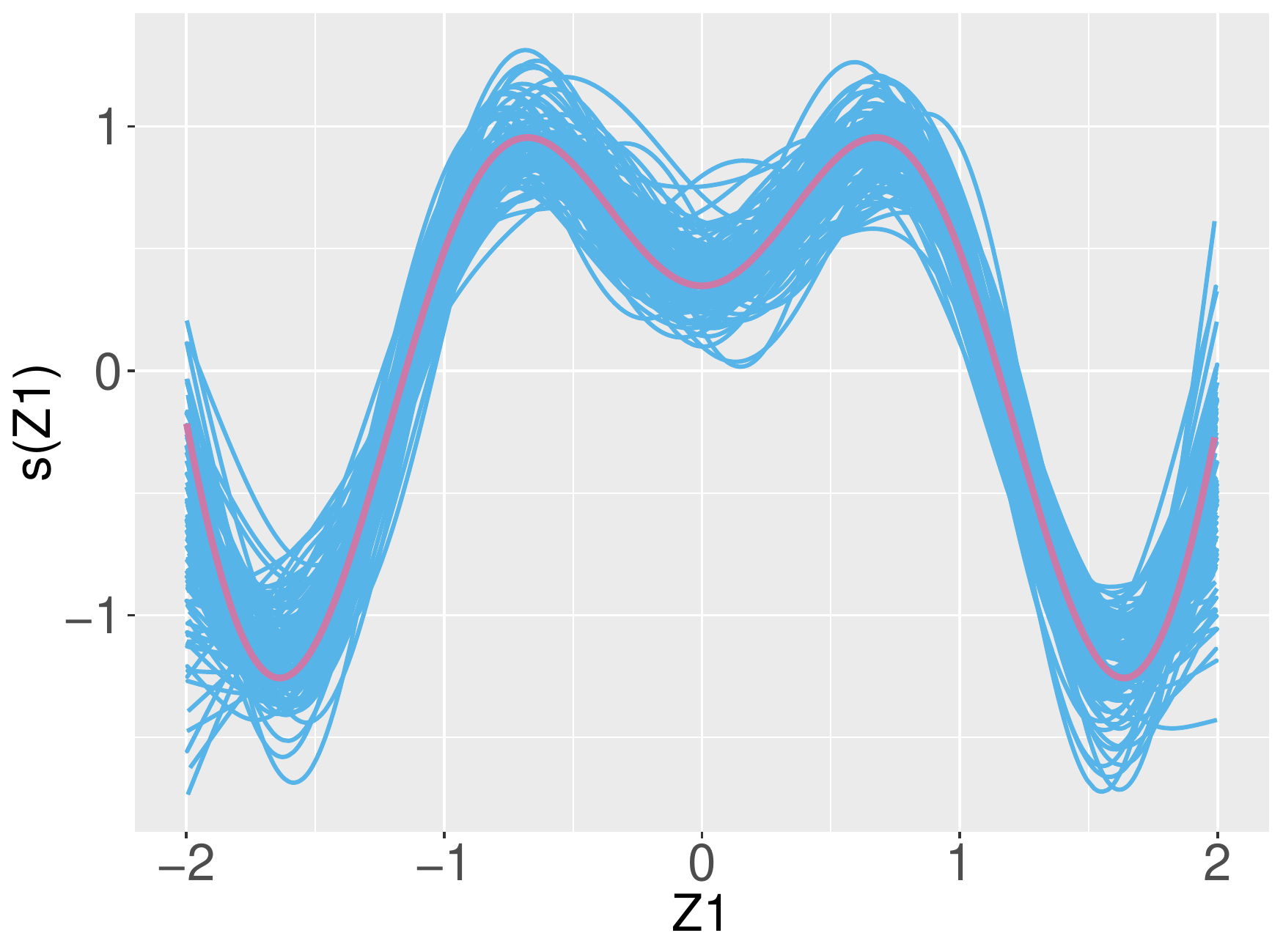}
    \caption{$s_1^{\mu_1}(\nu_1)$}
  \end{subfigure}%
  ~ 
  \begin{subfigure}{0.5\textwidth}
   % \hspace{-5mm}%
    \includegraphics[scale = 0.35]{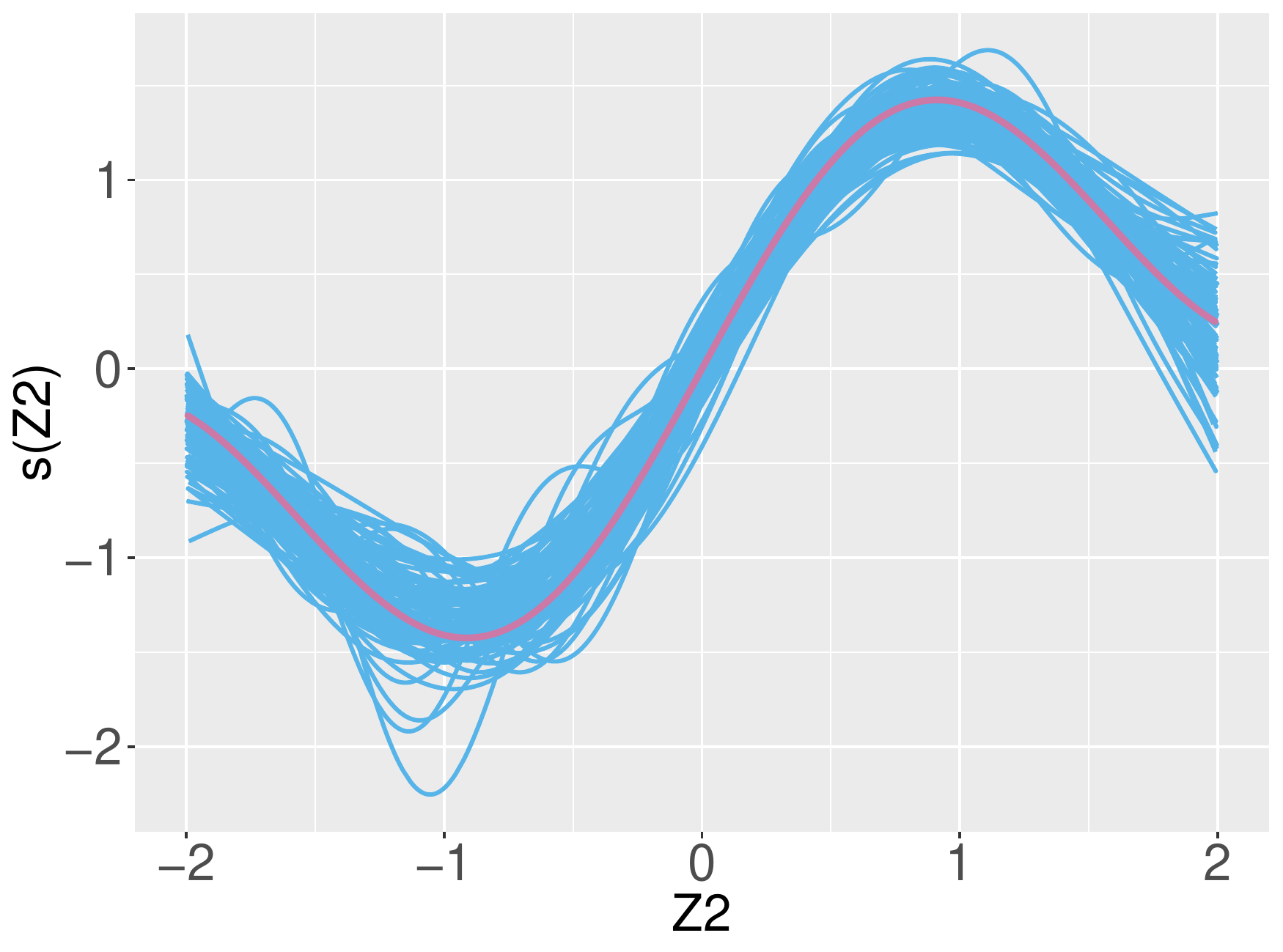}
    \caption{$s_2^{\mu_1}(\nu_2)$}
  \end{subfigure}
  \begin{subfigure}{0.5\textwidth}
  %  \hspace{-10mm}
    \includegraphics[scale = 0.35]{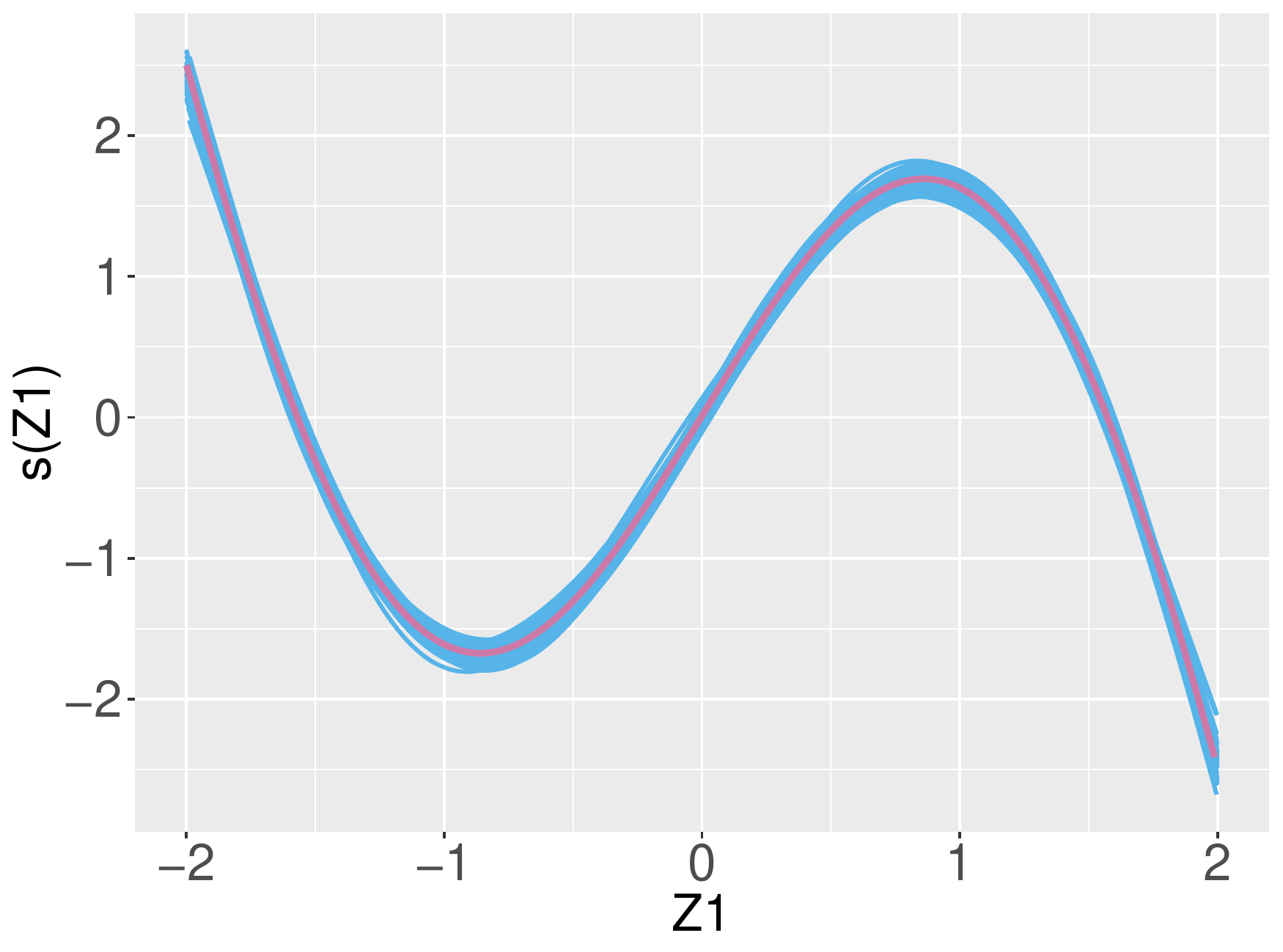}
    \caption{$s_3^{\mu_2}(\nu_1)$}
  \end{subfigure}%
  ~
  \begin{subfigure}{0.5\textwidth}
  %    \hspace{-5mm}%
    \includegraphics[scale = 0.35]{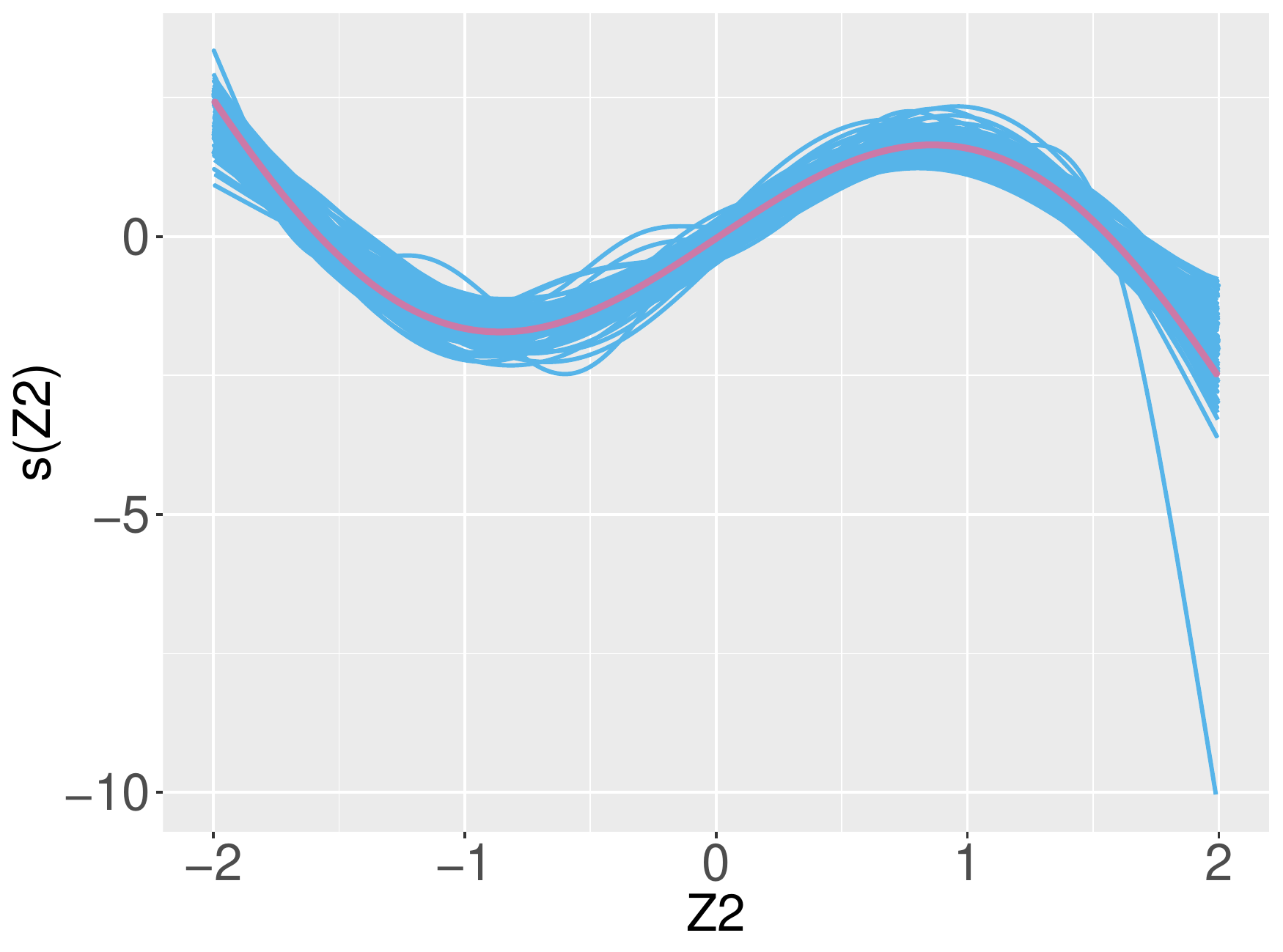}
    \caption{$s_3^{\gamma}(\nu_2)$}
  \end{subfigure}
\caption{Simulation results for linear and non-linear effects (case $n = 1000$) in scenario 3. The boxplots represent the estimated linear coefficients in $N = 100$ iterations. The true values of the coefficients are denoted by the black diamond symbols. The pink solid lines represent the true functions of the non-linear effects.} \label{fig:sim_smooth_scen3}
\end{figure}

%%%%%%%%

\begin{figure} 
    %\centering%
  \begin{subfigure}{0.5\textwidth}
 %    \hspace{-10mm}
    \includegraphics[scale = 0.35]{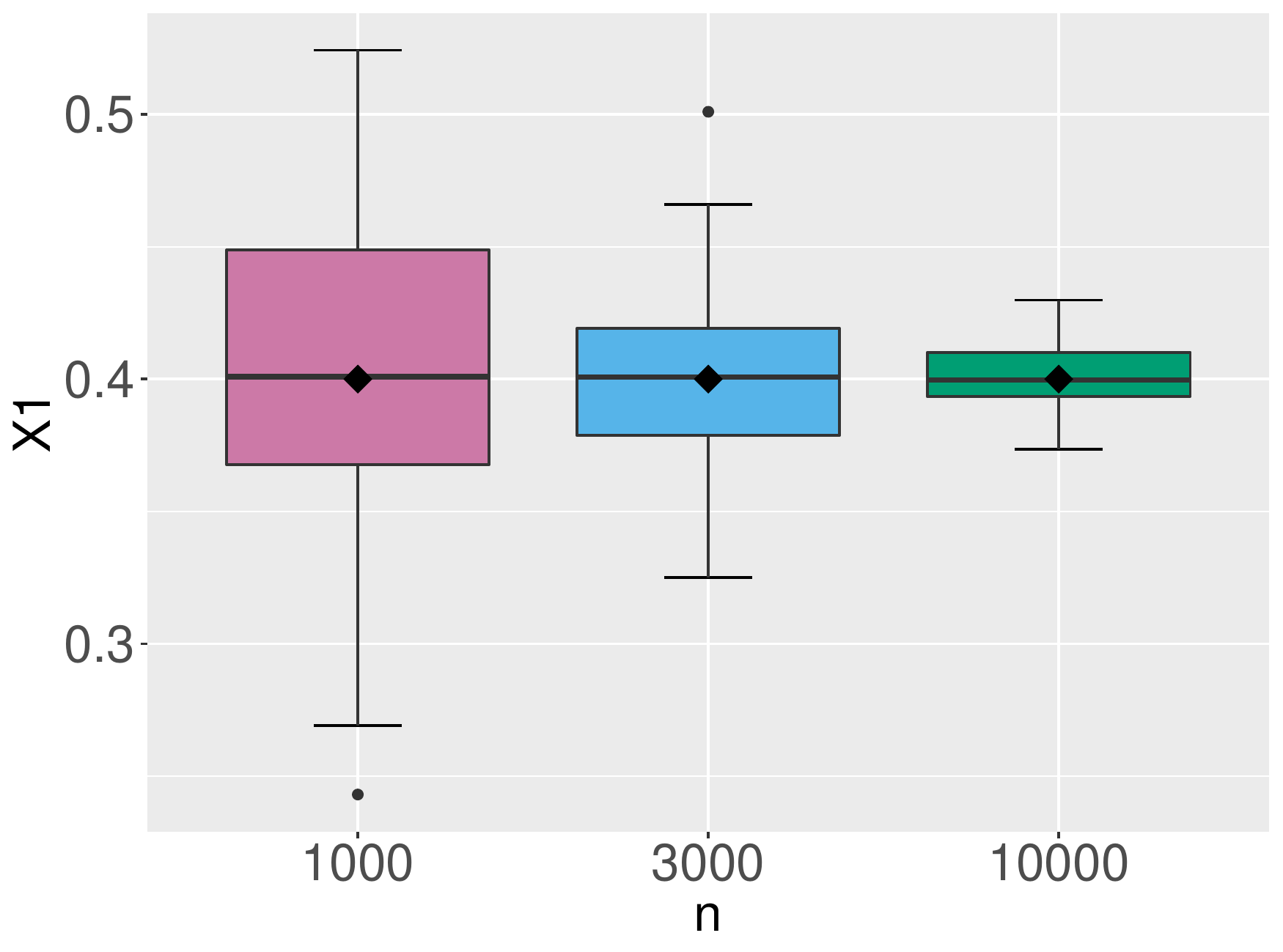}
    \caption{coefficients for $x_1$ of $\mu_1$}
  \end{subfigure}%
  ~ 
  \begin{subfigure}{0.5\textwidth}
%    \hspace{-5mm}%
    \includegraphics[scale = 0.35]{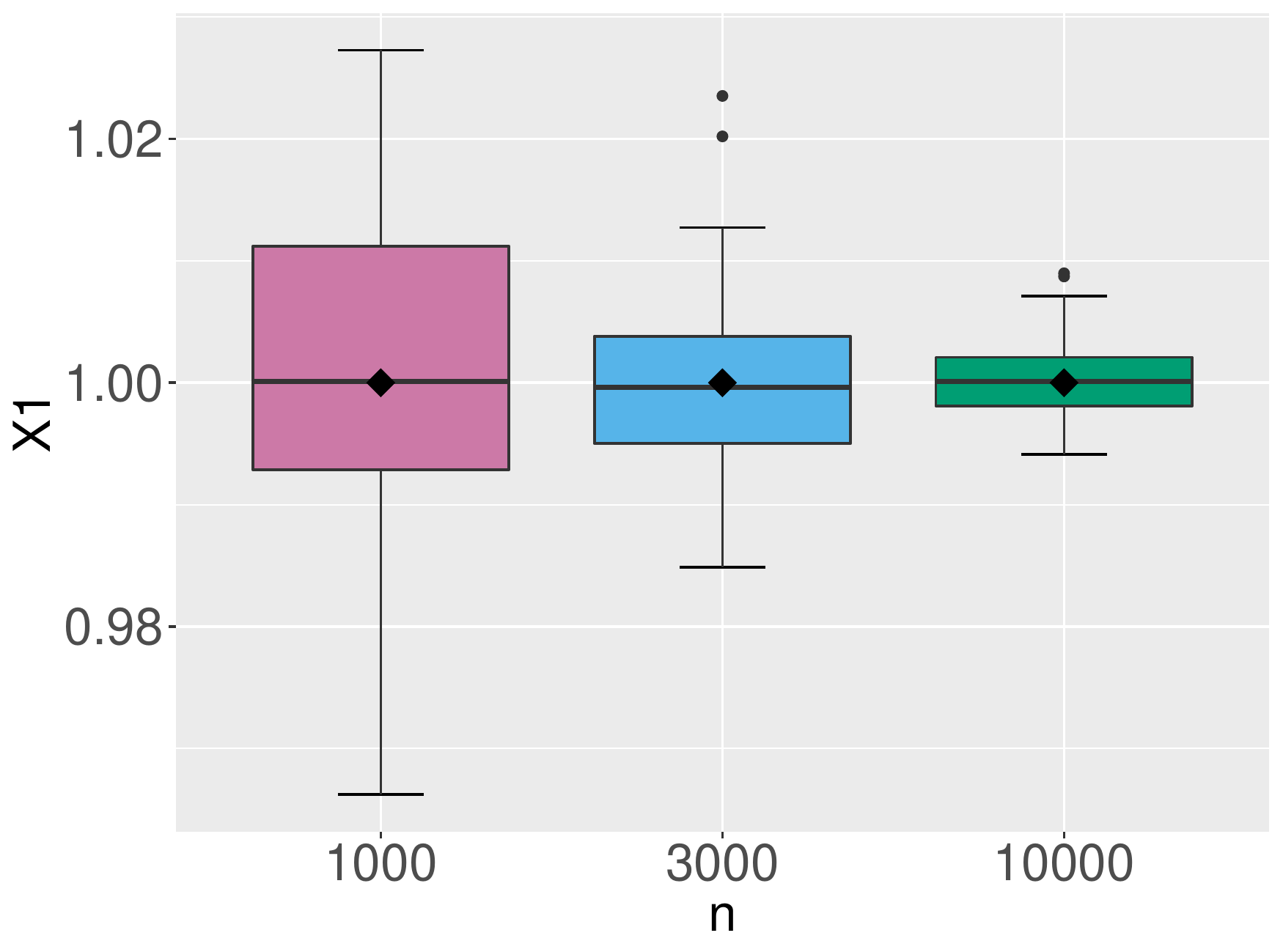}
    \caption{coefficients for $x_1$ of $\mu_2$}
  \end{subfigure}
  \begin{subfigure}{0.5\textwidth}
 %   \hspace{-10mm}
    \includegraphics[scale = 0.35]{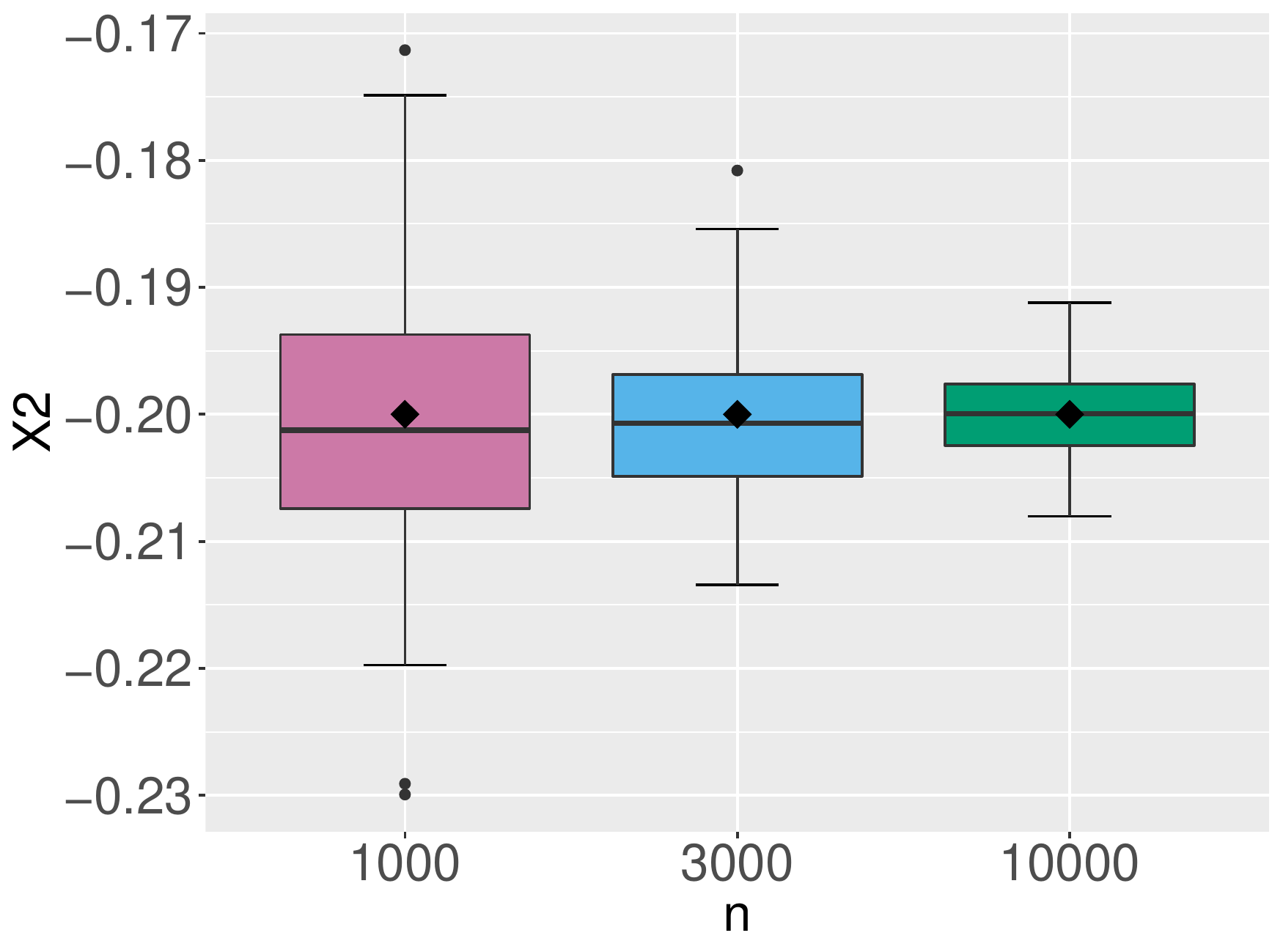}
    \caption{coefficients for $x_2$ of $\mu_2$}
  \end{subfigure}%
  ~
  \begin{subfigure}{0.5\textwidth}
%      \hspace{-5mm}%
    \includegraphics[scale = 0.35]{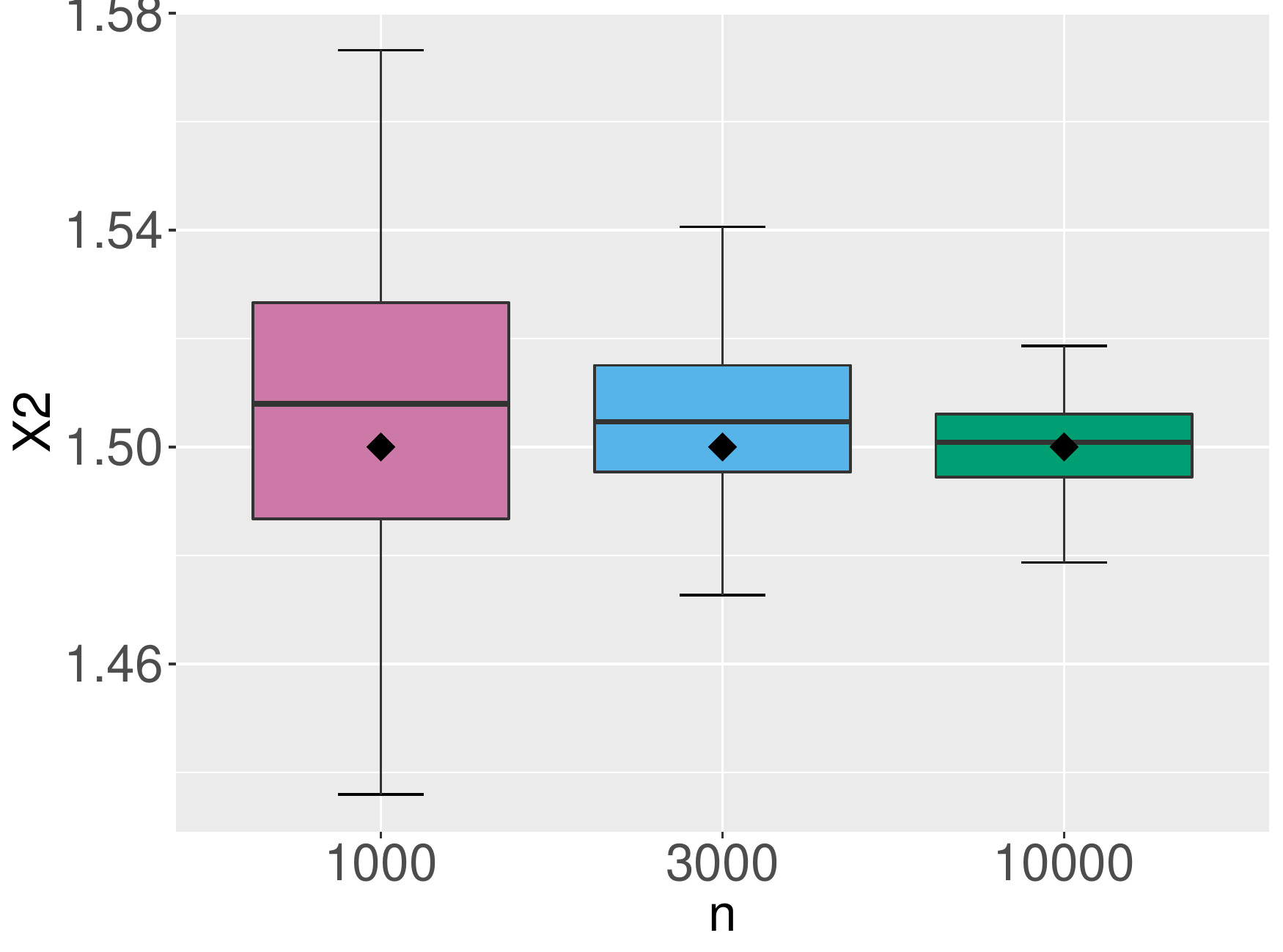}
    \caption{coefficients for $x_3$ of $\sigma_2$}
  \end{subfigure}
  \begin{subfigure}{0.5\textwidth}
 %    \hspace{-10mm}
    \includegraphics[scale = 0.35]{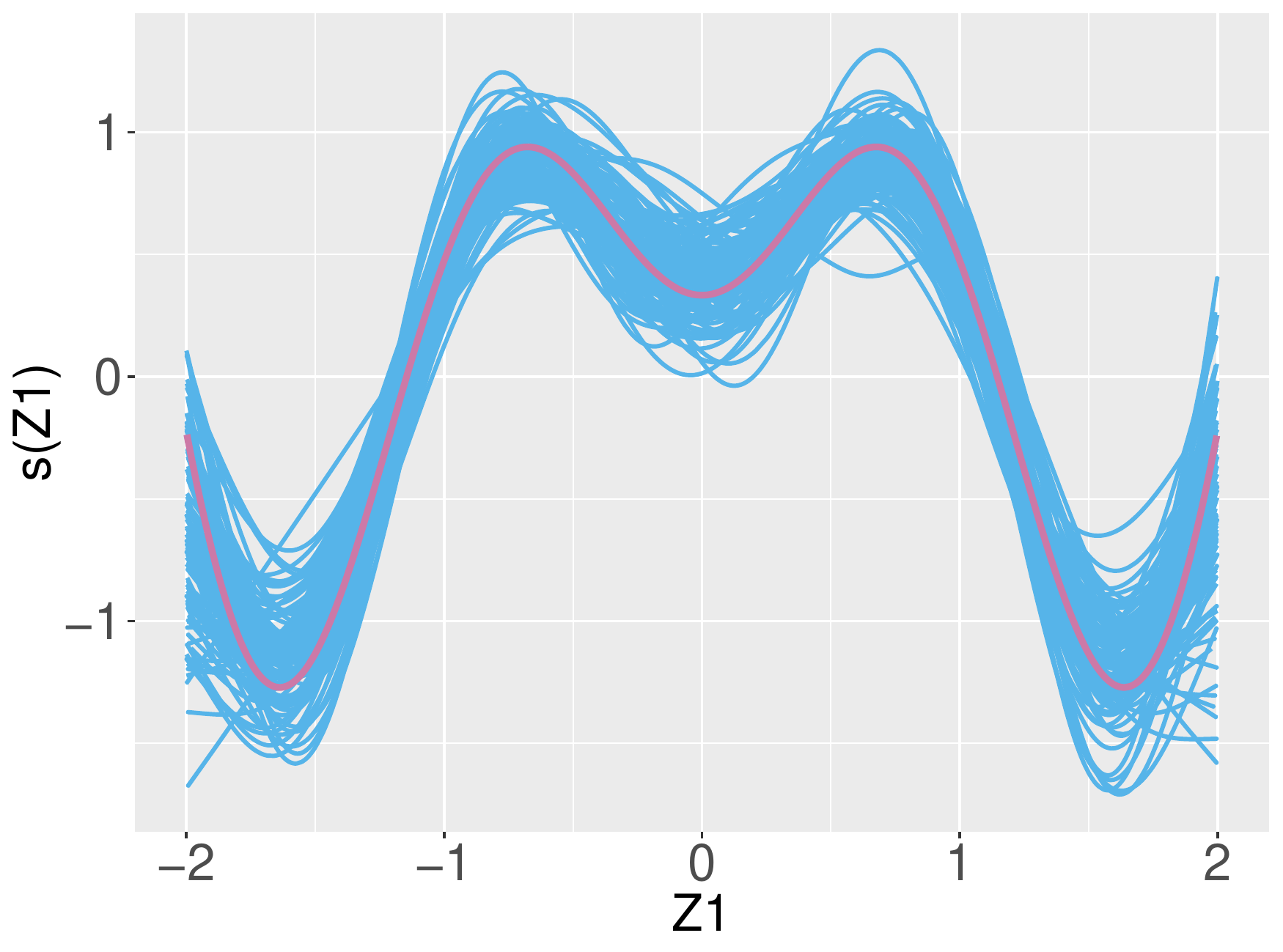}
    \caption{$s_1^{\mu_1}(\nu_1)$}
  \end{subfigure}%
  ~ 
  \begin{subfigure}{0.5\textwidth}
 %   \hspace{-5mm}%
    \includegraphics[scale = 0.35]{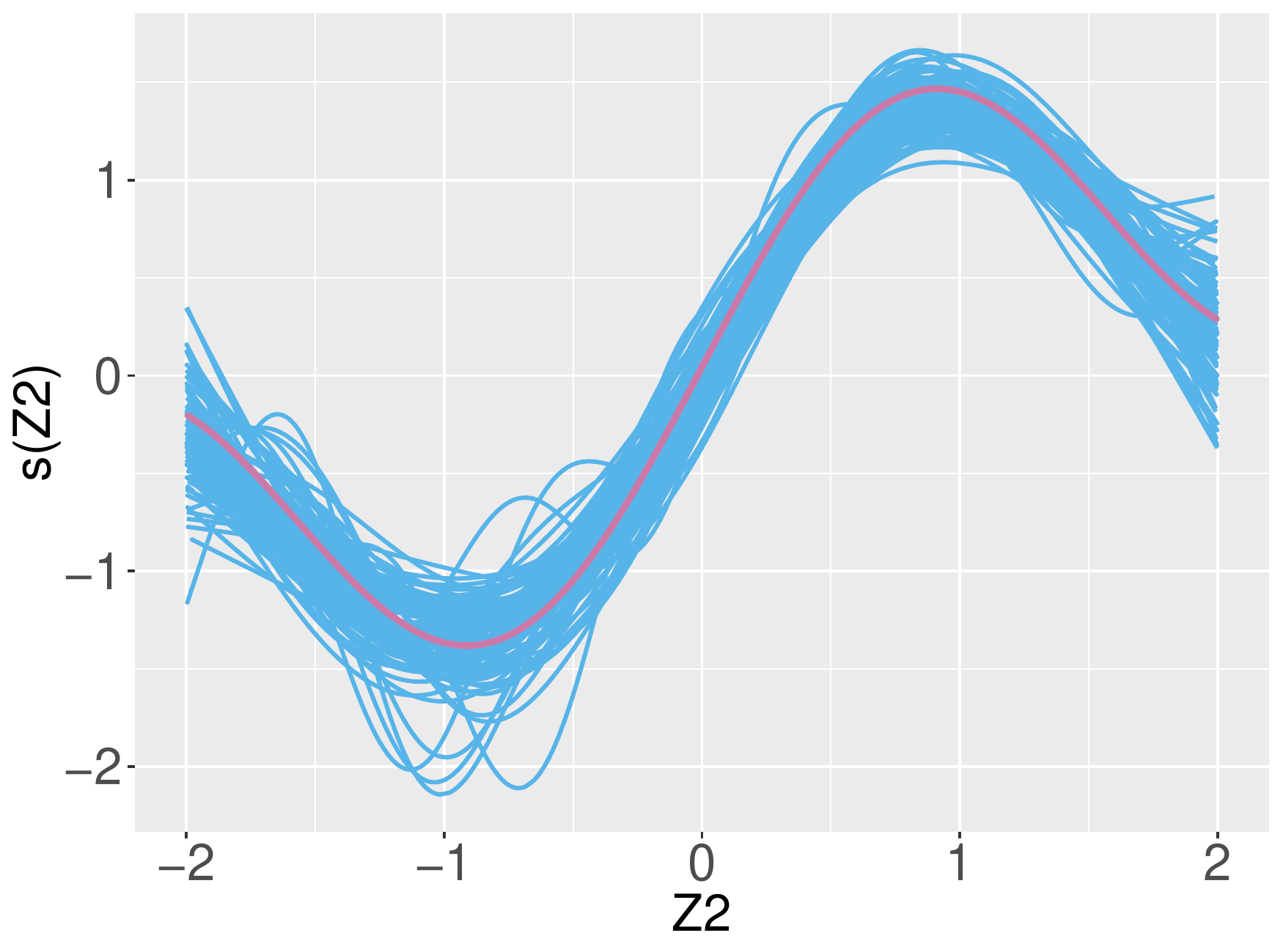}
    \caption{$s_2^{\mu_1}(\nu_2)$}
  \end{subfigure}
  \begin{subfigure}{0.5\textwidth}
%    \hspace{-10mm}
    \includegraphics[scale = 0.35]{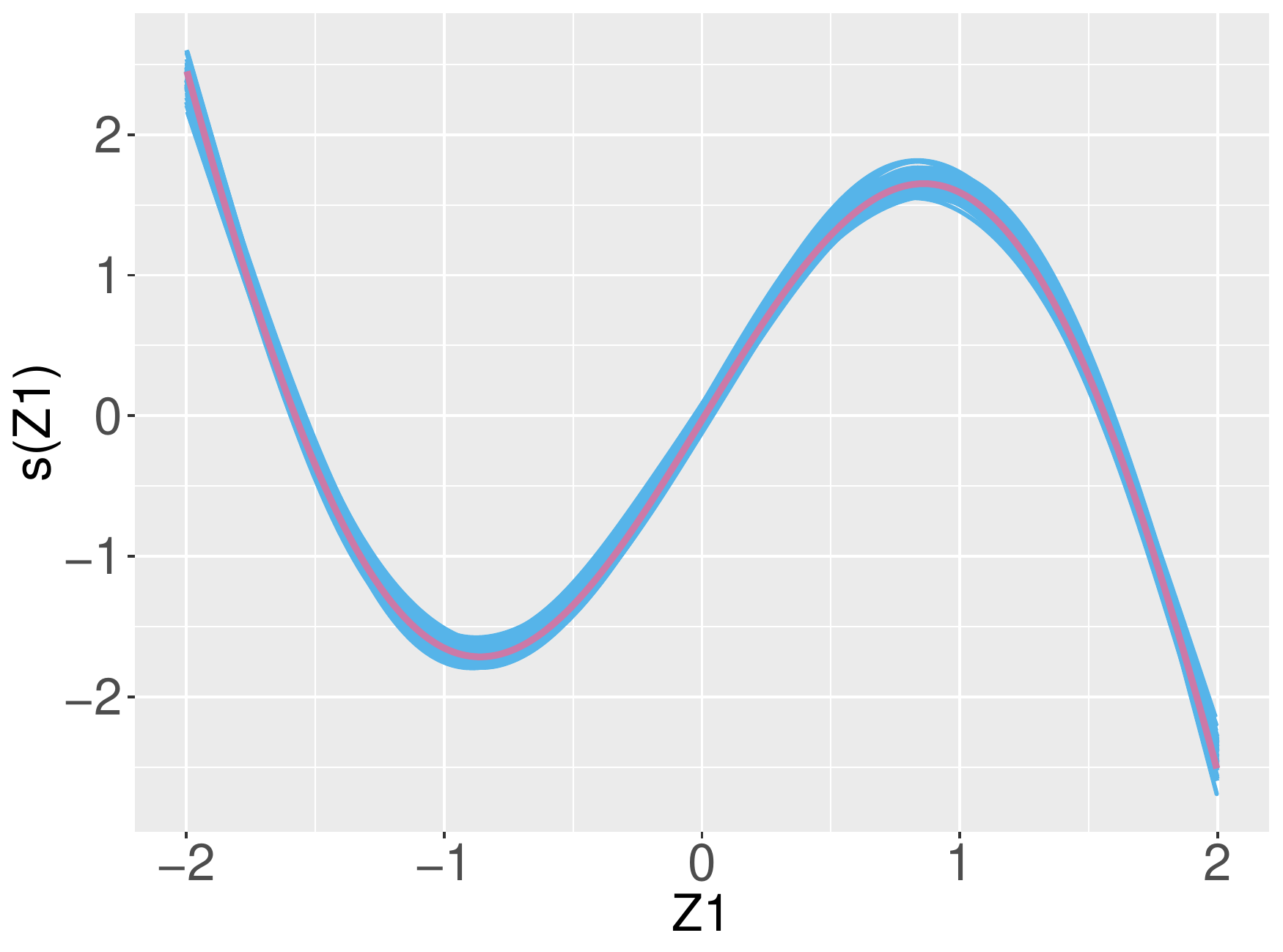}
    \caption{$s_3^{\mu_2}(\nu_1)$}
  \end{subfigure}%
  ~
  \begin{subfigure}{0.5\textwidth}
 %     \hspace{-5mm}%
    \includegraphics[scale = 0.35]{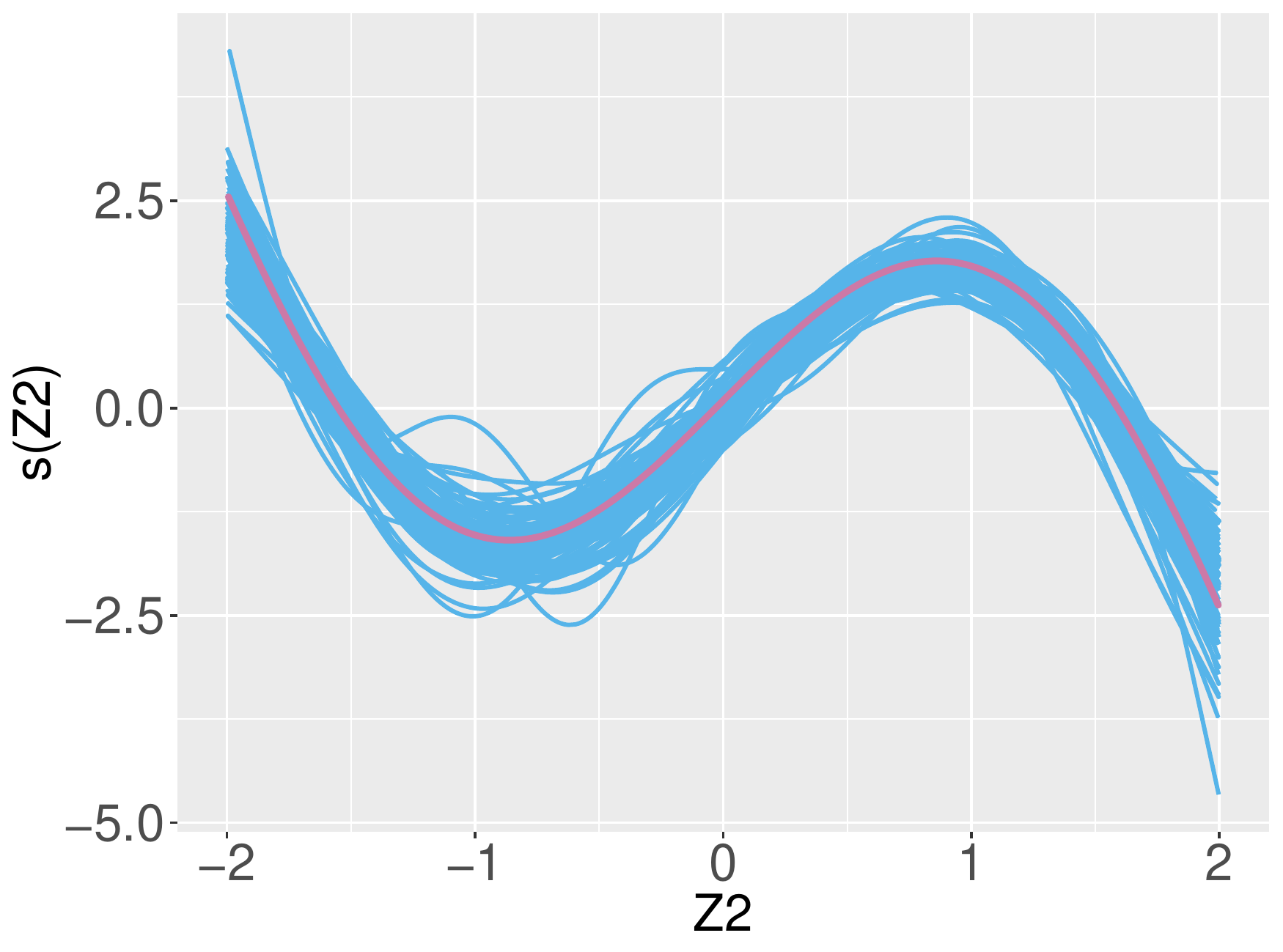}
    \caption{$s_3^{\gamma}(\nu_2)$}
  \end{subfigure}
\caption{Simulation results for linear and non-linear effects (case $n = 1000$) in scenario 4. The boxplots represent the estimated linear coefficients in $N = 100$ iterations. The true values of the coefficients are denoted by the black diamond symbols. The pink solid lines represent the true functions of the non-linear effects.} \label{fig:sim_smooth_scen4}
\end{figure}

%%%%%%%%%%%%%%%%%%%%%%%%%%%%%%%%%%%%%%%%%%%%%%%%%%%%%%%%%%%%%%%%%%%%%%
%%%%%%%%%%%%%%%%%%%%%%%%%%%%%%%%%%%%%%%%%%%%%%%%%%%%%%%%%%%%%%%%%%%%% 
\newpage

\section{Additional material for the application case\label{apx_application}}

%\subsection{Regression results for education and income}

\begin{figure}[ht]
    \centering
    \includegraphics[scale=0.6]{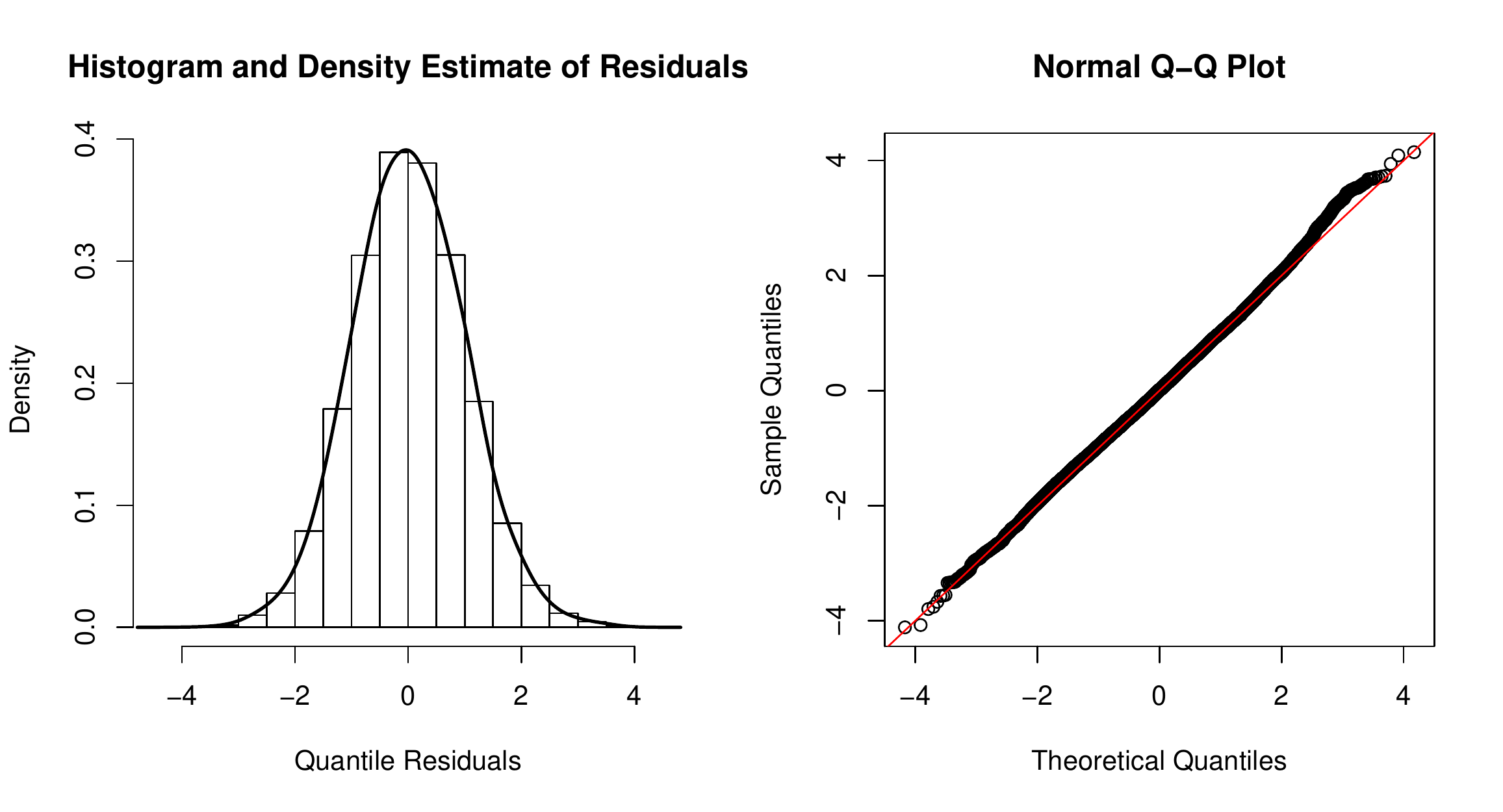}
        \caption{Histogram and normal Q-Q plots for the log-normal continuous margin of the final bivariate copula model.}
    \label{fig:post_check_biv}
\end{figure}

\begin{table}[ht]
\centering
\caption{Effects on education}\label{tab:educ}
\begin{tabular}{rrrrr}
  \hline
 & Estimate & Std. Error & z value & Pr($>$$|$z$|$) \\ 
  \hline
cutoff 1 & -3.170 & 0.071 & -44.348 & 0.000 \\ 
  cutoff 2 & -0.201 & 0.008 & -24.059 & 0.000 \\ 
  cutoff 3 & 0.761 & 0.006 & 129.930 & 0.000 \\ 
  cutoff 4 & 2.618 & 0.006 & 413.481 & 0.000 \\ 
  marital status of household head: married & -0.460 & 0.063 & -7.272 & 0.000 \\ 
  marital status of household head: separated & -0.637 & 0.089 & -7.128 & 0.000 \\ 
  marital status of household head: widowed & -0.448 & 0.073 & -6.151 & 0.000 \\ 
  household head is male & 0.185 & 0.042 & 4.353 & 0.000 \\ 
  urban dummy & 1.042 & 0.022 & 48.383 & 0.000 \\ 
  number of children: 1 & -0.166 & 0.028 & -5.969 & 0.000 \\ 
  number of children: 2 & -0.131 & 0.030 & -4.366 & 0.000 \\ 
  number of children: 3 & -0.168 & 0.043 & -3.936 & 0.000 \\ 
  number of children: 4-7 & -0.332 & 0.067 & -4.949 & 0.000 \\ 
  number of elderly: 1 & -0.003 & 0.031 & -0.111 & 0.912 \\ 
  number of elderly: 2 or 3 & 0.191 & 0.059 & 3.250 & 0.001 \\ 
  religion: Christian & 1.002 & 0.047 & 21.175 & 0.000 \\ 
  religion: Hinduism and other & -0.028 & 0.048 & -0.576 & 0.565 \\ 
   \hline
   \vspace{0.5pt}\\
\multicolumn{5}{p{0.98\linewidth}}{\footnotesize{\textit{Note:} Base categories for marital status is ``not yet married'', for number of children ``no children'', for number of elderly  ``no elderly'', and for religion ``Islam''.}}
\end{tabular}
\end{table}

\begin{table}[ht]
\centering
\caption{Effects on $\mu$ of the income distribution} \label{tab:income}
\begin{tabular}{rrrrr}
  \hline
 & Estimate & Std. Error & z value & Pr($>$$|$z$|$) \\ 
  \hline
Intercept & 14.430 & 0.031 & 465.128 & 0.0000 \\ 
  marital status of household head: married & -0.142 & 0.021 & -6.674 & 0.0000 \\ 
  marital status of household head: separated & -0.140 & 0.029 & -4.834 & 0.0000 \\ 
  marital status of household head: widowed & -0.145 & 0.024 & -5.997 & 0.0000 \\ 
  household head is male & -0.060 & 0.013 & -4.754 & 0.0000 \\ 
  education of household head: primary & 0.127 & 0.022 & 5.734 & 0.0000 \\ 
  education of household head: middle school & 0.245 & 0.023 & 10.713 & 0.0000 \\ 
  education of household head: high school & 0.388 & 0.023 & 16.944 & 0.0000 \\ 
  education of household head: tertiary education & 0.663 & 0.026 & 25.048 & 0.0000 \\ 
  urban dummy & 0.120 & 0.007 & 17.852 & 0.0000 \\ 
  number of children: 1 & -0.259 & 0.008 & -32.190 & 0.0000 \\ 
  number of children: 2 & -0.414 & 0.009 & -47.637 & 0.0000 \\ 
  number of children: 3 & -0.559 & 0.012 & -45.353 & 0.0000 \\ 
  number of children: 4-7 & -0.738 & 0.019 & -39.319 & 0.0000 \\ 
  number of elderly: 1 & -0.189 & 0.009 & -21.082 & 0.0000 \\ 
  number of elderly: 2 or 3 & -0.312 & 0.017 & -18.660 & 0.0000 \\ 
  religion: Christian & 0.059 & 0.016 & 3.818 & 0.0001 \\ 
  religion: Hinduism and other & 0.166 & 0.026 & 6.441 & 0.0000 \\ 
   \hline
      \vspace{0.5pt}\\
\multicolumn{5}{p{0.98\linewidth}}{\footnotesize{\textit{Note:} Base categories for marital status is ``not yet married'', for education ``no schooling'', for number of children ``no children'', for number of elderly  ``no elderly'', and for religion ``Islam''.}}
\end{tabular}
\end{table}

\begin{table}[ht]
\centering
\caption{Effects on $\sigma$ of the income distribution}\label{tab:income_sigma}
\begin{tabular}{rrrrr}
  \hline
 & Estimate & Std. Error & z value & Pr($>$$|$z$|$) \\ 
  \hline
Intercept & -0.426 & 0.022 & -19.394 & 0.000 \\ 
  marital status of household head: married & -0.176 & 0.023 & -7.617 & 0.000 \\ 
  marital status of household head: separated & -0.075 & 0.033 & -2.262 & 0.024 \\ 
  marital status of household head: widowed & -0.116 & 0.026 & -4.458 & 0.000 \\ 
  number of children: 1 & -0.065 & 0.011 & -6.132 & 0.000 \\ 
  number of children: 2 & -0.081 & 0.011 & -7.021 & 0.000 \\ 
  number of children: 3 & -0.066 & 0.017 & -4.000 & 0.000 \\ 
  number of children: 4-7 & -0.076 & 0.026 & -2.958 & 0.003 \\ 
  religion: Christian & 0.046 & 0.019 & 2.394 & 0.017 \\ 
  religion: Hinduism and other & 0.035 & 0.026 & 1.310 & 0.190 \\ 
   \hline
      \vspace{0.5pt}\\
\multicolumn{5}{p{0.98\linewidth}}{\footnotesize{\textit{Note:} Base categories for marital status is ``not yet married'', for number of children ``no children'', for number of elderly  ``no elderly'', and for religion ``Islam''.}}
\end{tabular}
\end{table}

\begin{table}[ht]
\centering
\caption{Effects on the copula parameter}\label{tab:gamma}
\begin{tabular}{rrrrr}
  \hline
 & Estimate & Std. Error & z value & Pr($>$$|$z$|$) \\ 
  \hline
Intercept & -0.179 & 0.045 & -3.998 & 0.000 \\ 
  marital status of household head: married & 0.143 & 0.034 & 4.219 & 0.000 \\ 
  marital status of household head: separated & 0.170 & 0.048 & 3.517 & 0.000 \\ 
  marital status of household head: widowed & 0.177 & 0.039 & 4.543 & 0.000 \\ 
  education of household head: primary & 0.061 & 0.033 & 1.877 & 0.060 \\ 
  education of household head: middle school & 0.104 & 0.036 & 2.906 & 0.004 \\ 
  education of household head: high school & 0.161 & 0.033 & 4.853 & 0.000 \\ 
  education of household head: tertiary education & 0.120 & 0.036 & 3.337 & 0.001 \\ 
  urban dummy & 0.101 & 0.013 & 7.816 & 0.000 \\ 
  number of children: 1 & 0.018 & 0.016 & 1.128 & 0.259 \\ 
  number of children: 2 & 0.041 & 0.017 & 2.411 & 0.016 \\ 
  number of children: 3 & 0.084 & 0.025 & 3.391 & 0.001 \\ 
  number of children: 4-7 & 0.063 & 0.036 & 1.746 & 0.081 \\ 
  number of elderly: 1 & 0.075 & 0.017 & 4.402 & 0.000 \\ 
  number of elderly: 2 or 3 & 0.113 & 0.030 & 3.748 & 0.000 \\ 
   \hline
      \vspace{0.5pt}\\
\multicolumn{5}{p{0.98\linewidth}}{\footnotesize{\textit{Note:} Base categories for marital status is ``not yet married'', for education ``no schooling'', for number of children ``no children'', and for number of elderly  ``no elderly''.}}
\end{tabular}
\end{table}

%\begin{center}
%\begin{figure}
%\centering
%\includegraphics[scale=0.7]{figures/contour_age.pdf} \\
%\includegraphics[scale=0.6]{figures/contour_marstat.pdf} 
%\end{figure}
%\vskip 12pt
%\begin{minipage}{\textwidth}
%    \captionof{figure}{Contour plots for an average individual at different ages and with different marital status}
%\end{minipage}

\begin{figure}
\centering
\includegraphics[scale=0.6]{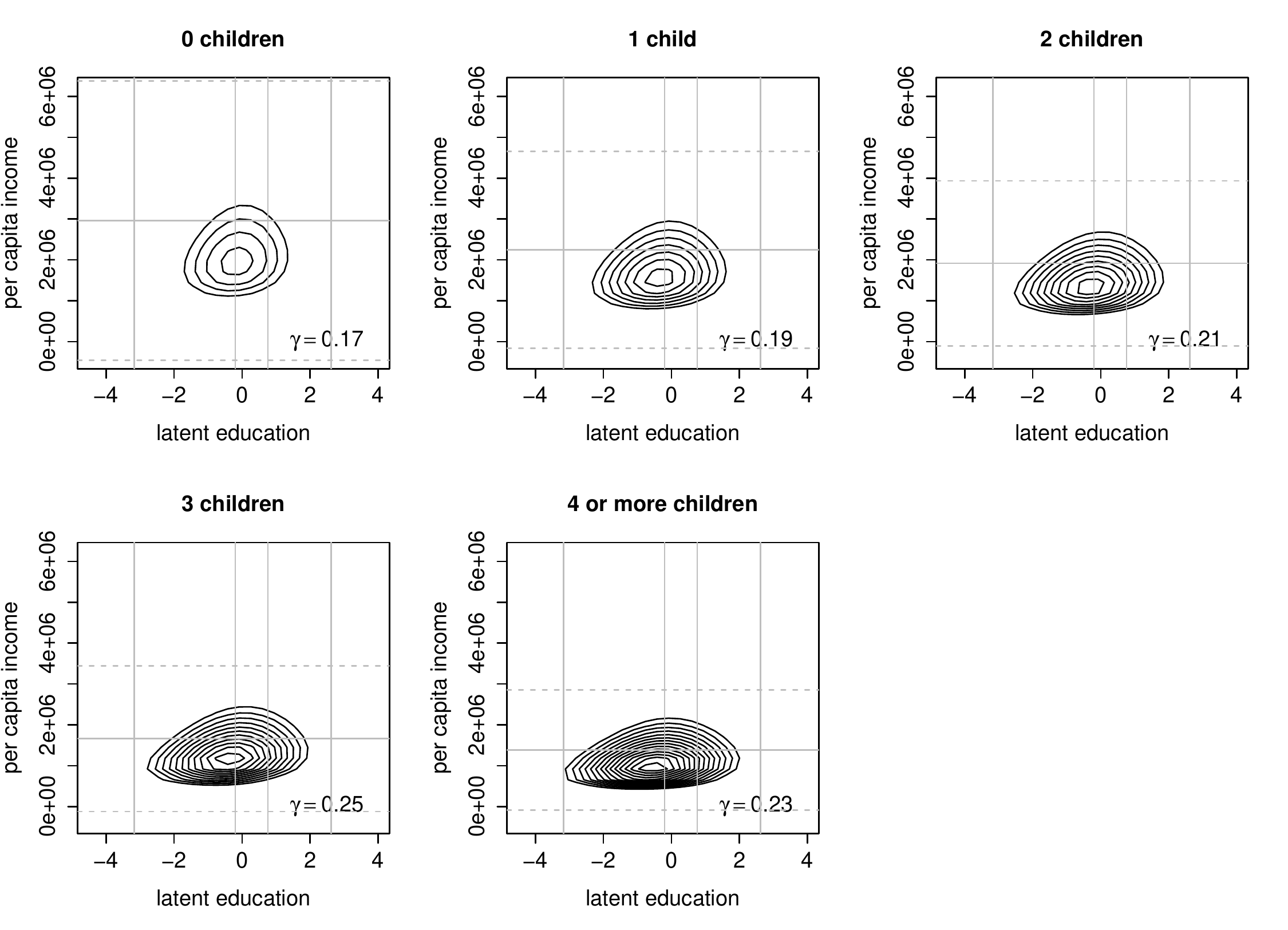} \\
%\hspace{4.8cm}
\includegraphics[scale=0.6]{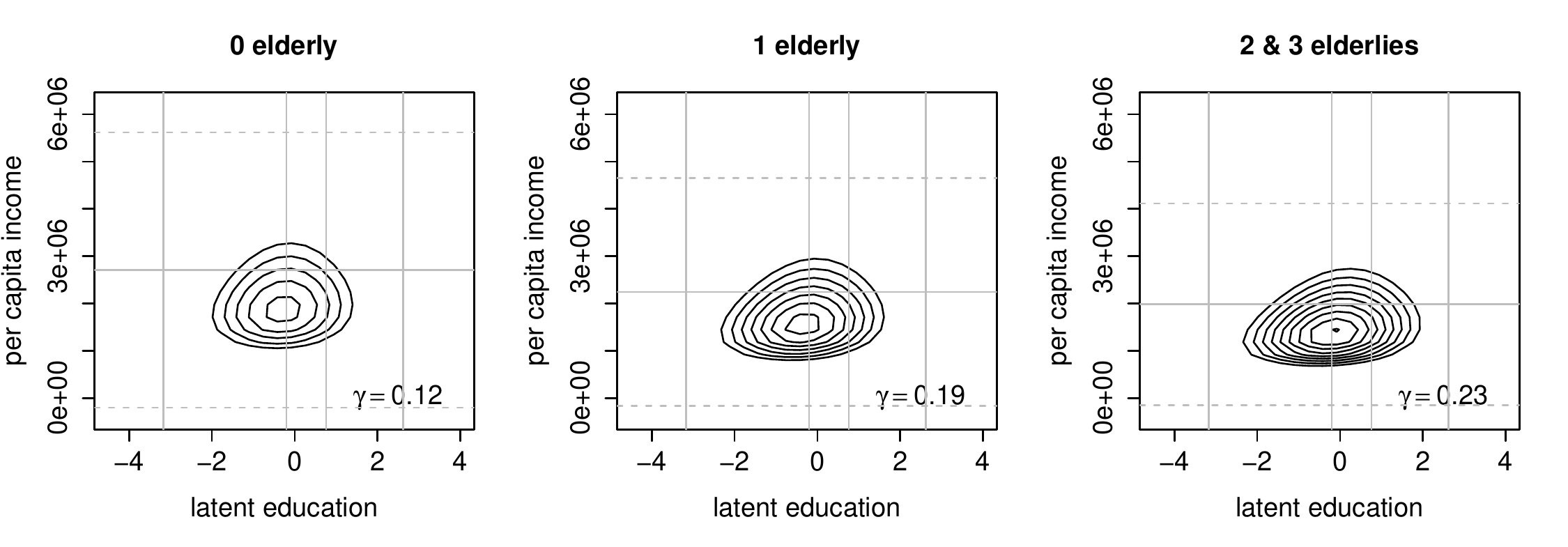} \\
%\caption{Contour plots for covariates in dependence equation \textcolor{red}{display only one set here, remaining ones in appendix, which ones?}}
\vskip 12pt
\begin{minipage}{\textwidth}
    \captionof{figure}{Contour plots for (education, income)' and a Gaussian copula by different numbers of children and elderly living in the same household. Contour lines of densities are at
    levels from 0.00000005 to 0.00000025 in 0.00000001 steps. The vertical straight lines represents the cut off values for the education categories, horizontal straight lines are the consumption average, and dashed horizontal line are at two standard deviations around this average. }
\end{minipage}
\end{figure}
%\end{center}

%%%%%%%%%%%%%%%%%%%%%%%%%%%%%%%%%%%%%%%%%%%%%%%%%%%%%%%%%%%%%%%%%%%%%%
%%%%%%%%%%%%%%%%%%%%%%%%%%%%%%%%%%%%%%%%%%%%%%%%%%%%%%%%%%%%%%%%%%%%% 
\clearpage
\section{Software}

We incorporated the models proposed this paper into the \texttt{GJRM} package \citep{GJRM} in \texttt{R} \citep{R.2019}. A mixed ordered-continuous model is called by setting the option \texttt{ordinal = TRUE}. The main fitting function, \texttt{gjrm()}, is very easy to use in that its syntax follows those of linear models, generalized linear models, or generalized additive models. The function \texttt{CopulaCLM()} is called internally to fit this specific model. An example of model specification is given below.   

\begin{verbatim}
eq.educ <- educ_att ~ s(age) + as.factor(hhmarstat) + as.factor(hhmale) + 
                      as.factor(urban) + as.factor(num_child) + 
                      as.factor(elderly) + as.factor(relig)  

eq.mu   <- pce.defl ~ s(age) + as.factor(hhmarstat) + as.factor(hhmale) + 
                      as.factor(urban) + as.factor(num_child) + 
                      as.factor(elderly) + as.factor(relig) + 
                      as.factor(hheduc) + s(prov, bs = "mrf", xt = xt1, k = 15)

eq.si   <-          ~ s(age) + as.factor(hhmarstat) + 
                      as.factor(num_child) + as.factor(elderly) + 
                      as.factor(relig) + s(prov, bs = "mrf", xt = xt1, k = 15)
 
eq.theta <-         ~ s(age) + as.factor(hhmarstat) + 
                      as.factor(urban) + as.factor(num_child) + 
                      as.factor(elderly) + as.factor(hheduc) +
                      s(prov, bs = "mrf", xt = xt1, k = 15)

form.list <- list(eq.educ, eq.mu, eq.si, eq.theta)

mod.edu <- gjrm(form.list, data = na.omit(df), ordinal = TRUE, 
                     Model = "B", BivD = "N", margins = c("logit", "LN"), 
                     drop.unused.levels = FALSE, gamlssfit = TRUE) 
\end{verbatim}

Te user first specifies the four equations for the model's parameters of the marginal distributions and of the copula, which are stored in the list \texttt{form.list}. Continuous variables enter the model specifications via smooth effects \texttt{s()} represented (by default) via thin-plate regression splines (argument \texttt{bs = "tp"}) with ten basis function and second order derivative penalties. Spatial effects of the provinces are modeled using Markov random fields with neighbourhood structure \texttt{xt1} and 15 knots (argument \texttt{bs = "mrf"}). Argument \texttt{ Model = "B"} specifies that a bivariate model will be estimated, \texttt{ margins = c("logit", "LN")} gives the marginal distributions and \texttt{BivD = "N"} specifies the Gaussian copula. The argument \texttt{ordinal} must be set to \texttt{TRUE} in order to fit a mixed ordered-continuous model and the ordinal outcome \texttt{educ\_att} must be numeric. The optional argument \texttt{gamlssfit = TRUE} uses starting values obtained from a univariate gamlss and \texttt{drop.unused.levels = FALSE} is needed because not all of the provinces specified via the Markov random fields have observations in the data frame \texttt{df}. Functions \texttt{summary()}, \texttt{plot()}, \texttt{AIC()} and \texttt{BIC()} can employed in the usual manner. It is advisable to use \texttt{post.check()} after fitting the model to produce plots of normalized quantile residuals. More details, options, and the available choices for the marginal distributions and copula functions can be found in the documentation of the \texttt{GJRM} package.

\end{appendices}

\end{document}